\documentclass[a4paper,12pt]{article}
\usepackage{amsmath,amssymb,mathtools,algpseudocode,algorithm,indentfirst,amsthm}
\usepackage[x11names,table]{xcolor}
\usepackage{etoolbox,graphicx}
\usepackage[toc,page]{appendix}
\usepackage{subcaption}
\usepackage[english]{babel}
\usepackage{tikz}
\usepackage{verbatim}
\usepackage[margin=3cm]{geometry}
\usepackage{url}

\DeclareMathOperator*{\argmax}{argmax}
\DeclareMathOperator*{\argmin}{argmin}

\DeclarePairedDelimiter\norm{\|}{\|}

\newcommand{\rcol}[2]{R_#1^{(#2)}}

\algnewcommand{\Initialize}[1]{
  \State \textbf{Initialize:}\:#1
}

\newtheorem{theorem}{Theorem}
\newtheorem{lemma}[theorem]{Lemma}
\newtheorem{proposition}[theorem]{Proposition}
\newtheorem{corollary}[theorem]{Corollary}
\theoremstyle{definition}

\newtheorem{definition}[theorem]{Definition}
\numberwithin{equation}{section}
\numberwithin{theorem}{section}

\newcommand{\rv}[1]{
\textcolor{red}{#1}%
}
\newcommand{\rvc}[2]{
\protect\textcolor{red}{#1}%
    \IfNoValueF{#2}{\protect\endnote{#2}%
    }%
}

\title{Common lines ab-initio reconstruction of $D_2$-symmetric molecules}
\author{{Eitan Rosen} and {Yoel Shkolnisky}}

\begin{document}
\maketitle

\begin{abstract}

Cryo-electron microscopy is a state-of-the-art method for determining high-resolution three-dimensional models of molecules, from their two-dimensional projection images taken by an electron microscope. A crucial step in this method is to determine a low-resolution model of the molecule using only the given projection images, without using any three-dimensional information, such as an assumed reference model. For molecules without symmetry, this is often done by exploiting common lines between pairs of images. Common lines algorithms have been recently devised for molecules with cyclic symmetry, but no such algorithms exist for molecules with dihedral symmetry.

In this work, we present a common lines algorithm for determining the structure of molecules with $D_{2}$ symmetry. The algorithm exploits the common lines between all pairs of images simultaneously, as well as common lines within each image. We demonstrate the applicability of our algorithm using experimental cryo-electron microscopy data.

\end{abstract}

\section{Introduction}\label{sec:introduction}
Cryo-electron microscopy (cryo-EM) is a technique for acquiring two-dimensional projection images of biological macromolecules~\cite{cryoBook}. In this technique, a large number of copies of the same molecule is rapidly frozen in a thin layer of ice, fixing each molecule in some random unknown orientation.
The frozen specimen is then imaged by an electron microscope, producing a set of two-dimensional projection images (defined below).
Once the imaging orientations of the frozen molecules which produced the images are obtained, the three-dimensional structure of the molecule can be recovered from the projection images by standard tomographic procedures.

Formally, choosing some arbitrary fixed coordinate system of $\mathbb{R}^3$,
the orientations of the imaged molecules at the moment of freezing can be described by a set of rotation matrices
\begin{equation}\label{Intro:Rots}
R_i=\left(\begin{array}{ccc} | & | & | \\R_i^{(1)} & R_i^{(2)} & R_i^{(3)} \\ |& | & | \end{array}\right)\in SO(3),\quad i\in[N],
\end{equation}
where we denote by $[N]$ the set $\{1,\ldots,N\}$, and $SO(3)$ is the group of $3\times 3$ rotation matrices.
We denote the density function of the molecule by $\phi(r):\mathbb{R}^3\to\mathbb{R}$, where $r=(x,y,z)^T$, and by
$P_{R_i}$ the image generated by the microscope when imaging a copy of $\phi$ rotated by $R_i$. The image $P_{R_i}$ is then given by the line integrals of $\phi(r)$ along the lines parallel to $R_i^{(3)}$, namely
\begin{equation}\label{Intro:ProjEq}
P_{R_i}(x,y)=\int_{-\infty}^{\infty}\phi(R_ir)dz=\int_{-\infty}^{\infty}\phi(xR^{(1)}_i+yR^{(2)}_i+zR^{(3)}_i)dz.
\end{equation}
The orthogonal unit vectors $R^{(1)}_i$ and $R^{(2)}_i$, which span the plane perpendicular to $R_i^{(3)}$, form the $(x,y)$ coordinate
system for the image $P_{R_i}$ from the point of view of an observer looking from the direction of the electron beam.
We refer to $R_i^{(3)}$ and its perpendicular plane as the beaming direction and the projection plane of $P_{R_i}$, respectively.

We can now state the ``orientation assignment problem'' as the task of finding a set of $N$ matrices $\{R_1,\ldots,R_N\} \in SO(3)$ such that~\eqref{Intro:ProjEq} is satisfied for all
$i\in[N]$, given only the images $P_{R_{1}},\ldots,P_{R_{N}}$ (in particular, $\phi$ in~\eqref{Intro:ProjEq} is unknown). In this work, we address the task of determining the orientations of a set of projection images obtained from a molecule with  $D_2$ symmetry.

In plain language, a $D_2$-symmetric molecule has three mutually perpendicular symmetry axes, where after we rotate the molecule by $180^\circ$ about any of these axes, the molecule looks exactly the same.
Formally, let
\begin{equation}\label{Sec1:gpMem}
g_2= \begin{pmatrix*}[r]
			1 & 0 & 0 \\
			0 & -1 & 0 \\
			0 & 0 & -1\end{pmatrix*},\quad
g_3=\begin{pmatrix*}[r]
			-1 & 0 & 0 \\
			0 & \phantom{+}1 & 0 \\
			0 & 0 & -1\end{pmatrix*} ,\quad
g_4= \begin{pmatrix*}[r]
			-1 & 0 & 0 \\
			0 & -1 & 0 \\
			0 & 0 & \phantom{+}1\end{pmatrix*}
\end{equation}
denote the three rotation matrices by $180^\circ$ about the $x,y$ and $z$ axes, respectively.
For notational convenience we also denote the $3\times3$ identity matrix by $g_1$.
Choosing a coordinate system in which the rotational symmetry axes of the molecule coincide with the $x,y$ and $z$ axes, the $D_2$ symmetry property implies that
\begin{equation}\label{Intro:volEquivalence}
\phi(r)=\phi(g_1r)=\phi(g_2r)=\phi(g_3r)=\phi(g_4r)
\end{equation}
for any $r\in\mathbb{R}^3$. Considering any projection image $P_{R_i}(x,y)$, by~\eqref{Intro:volEquivalence} we have
\begin{equation}\label{Intro:ProjIdentity}
P_{R_{i}}(x,y) =  \int_{-\infty}^{\infty}\phi(R_{i}r)dz=\int_{-\infty}^{\infty}\phi(g_mR_{i}r)dz= P_{g_mR_{i}}(x,y)
\end{equation}
for $m=2,3,4$. Equation~\eqref{Intro:ProjIdentity} shows that a $D_2$-symmetric molecule induces an ambiguity in which all orientation assignments of the form
$\{g_{m_i}R_i\}_{i=1}^N$, $g_{m_i}\in\{g_1,g_2,g_3,g_4\}$ are consistent with the same set of images $\{P_{R_i}\}_{i=1}^N$.

An additional ambiguity inherent to cryo-EM arises from the well known fact that the handedness (chirality) of a molecule cannot be established from its projection images.
Denoting by $J=\text{diag}(1,1,-1)$ the reflection matrix through the $xy$-plane, we define by $\psi(r)=\phi(Jr)$ the mirror image of the molecule $\phi(r)$.
Since $J^2=I$, we also have that $\phi(r)=\psi(Jr)$, and thus, by~\eqref{Intro:ProjEq} we have
\begin{equation*}
P_{R_i}(x,y)=\int_{-\infty}^{\infty}\phi(R_ir)dz =\int_{-\infty}^{\infty}\psi(JR_ir)dz=\int_{-\infty}^{\infty}\psi(JR_iJJr)dz,
\end{equation*}
where $r=(x,y,z)^T$.
Noting that $Jr=(x,y,-z)^T$, and changing the variable $z$ to $z'=-z$ we have
\begin{equation}\label{Intro:Jamb}
\int_{-\infty}^{\infty}\psi((JR_iJ)Jr)dz=\int_{-\infty}^{\infty}\psi(JR_iJ\begin{pmatrix}x \\ y \\ z'\end{pmatrix})dz'=\widetilde{P}_{JR_iJ}(x,y),
\end{equation}
which shows that any projection image $P_{R_i}$ of the molecule $\phi(r)$ is identical to the projection $\widetilde{P}_{JR_iJ}$ of its mirror image molecule $\psi(r)$.
Thus, both orientation assignments $\{R_k\}_{k=1}^N$ and $\{JR_kJ\}_{k=1}^N$ are consistent with the same set of projection images $\{P_{R_1},\ldots,P_{R_N}\}$.

The ``orientation assignment problem'' for a $D_2$-symmetric molecule can now be stated as the task of finding one of the sets of matrices $\{{R}_i\}_{i=1}^N$ or
$\{J{R}_iJ\}_{i=1}^N$ satisfying~\eqref{Intro:ProjEq}, where each matrix $R_i$ can be independently replaced by a matrix $\widetilde{R}_i\in\{g_mR_i\}_{m=2}^4$.

\section{Common lines and their $D_2$ induced geometry}\label{sec:D2 geometry}
One of the principal approaches to solving the orientation assignment problem, which we also employ in this work,
relies on the well know Fourier slice theorem~\cite{Natr2001a}, which states that the two-dimensional Fourier transform of a projection image $P_{R_i}$ is a central planar slice of the three-dimensional Fourier transform of the molecule $\phi(r)$. Formally, denoting by $\hat{\phi}$ the three-dimensional Fourier transform of $\phi(r)$, and by
$\hat{P}_{R_i}$ the two-dimensional Fourier transform of the image $P_{R_i}$, we have that
\begin{equation}\label{eq:Fourier slice theorem}
\hat{P}_{R_i}(\omega_x,\omega_y)=\hat{\phi}(\omega_x\rcol{i}{1}+\omega_y\rcol{i}{2})\:,\quad (\omega_x,\omega_y)\in\mathbb{R}^2.
\end{equation}
Note that the central slice of $\hat{\phi}$ in~\eqref{eq:Fourier slice theorem} is spanned by the vectors $\rcol{i}{1}$ and $\rcol{i}{2}$ (the first two columns of the rotation matrix $R_i$ of~\eqref{Intro:Rots}). Since any two non-coinciding central planes intersect at a single line through the origin, it follows that any pair of transformed projection images $\hat{P}_{R_i}$ and $\hat{P}_{R_j}$ for which $ | \langle \rcol{i}{3},\rcol{j}{3}  \rangle |\neq 1$ share a unique pair of identical central lines. Formally, consider the unit vector
\begin{equation}
q_{ij}=\frac{\rcol{i}{3}\times\rcol{j}{3}}{\|\rcol{i}{3}\times\rcol{j}{3}\|}.
\end{equation}
By the definition of the cross product, $q_{ij}$ is perpendicular to the normal vectors of both projection planes of $\hat{P}_{R_i}$ and $\hat{P}_{R_j}$, and thus it gives the direction of their common line (see \cite{sync} for a detailed discussion). We can express $q_{ij}$ using its local coordinates on both projection planes by
\begin{equation}\label{Sec2:LocalCoorDef}
q_{ij}=\cos\alpha_{ij}\rcol{i}{1}+\sin\alpha_{ij}\rcol{i}{2}=\cos\alpha_{ji}\rcol{j}{1}+\sin\alpha_{ji}\rcol{j}{2},
\end{equation}
where $\alpha_{ij}$ and $\alpha_{ji}$ are the angles between $q_{ij}$ and the local $x$-axes of the planes.
Using this notation, the common line property implies that
\begin{equation}
\hat{P}_{R_i}(\xi\cos\alpha_{ij},\xi\sin\alpha_{ij})=\hat{P}_{R_j}(\xi\cos\alpha_{ji},\xi\sin\alpha_{ji})\:, \quad \xi\in\mathbb{R}.
\end{equation}

Now, let $P_{R_i}$ and $P_{R_j}$ be a pair of images of a $D_2$-symmetric molecule.
By~\eqref{Intro:ProjIdentity}, the image $P_{R_j}$ (also denoted $P_{g_1R_j}$) is identical to the three images $P_{g_2R_j},P_{g_3R_j}$ and $P_{g_4R_j}$. However, each of these four images corresponds to a plane, and by~\eqref{Intro:ProjEq} the four planes are different from each other since their rotation matrices are different. Each of these planes has a unique common line with the projection plane of ${P}_{R_i}$, and thus, we conclude that there are four different pairs of identical central lines between the images $\hat{P}_{R_i}$ and $\hat{P}_{R_j}$. The directions of all four common lines are given by the unit vectors
\begin{equation}\label{Sec2:PhysCls}
q_{ij}^m=\frac{\rcol{i}{3}\times g_m\rcol{j}{3}}{\|\rcol{i}{3}\times g_m\rcol{j}{3}\|}\:,\quad m=1,2,3,4,
\end{equation}
where for convenience we denote $q_{ij}^1=q_{ij}$, since $g_1$ is the identity matrix.
We write the local coordinates of the common lines of the pairs of images $\hat{P}_{R_i}$ and $\hat{P}_{g_mR_j}$, on the respective projection plane of each image (see~\eqref{Sec2:LocalCoorDef}), as
\begin{equation}\label{Sec2:ClsCoor}
\begin{split}
C(R_i,g_mR_j)&=\begin{pmatrix}\cos{\alpha_{ij}^m} , \sin{\alpha_{ij}^m} \end{pmatrix}=
\begin{pmatrix}\langle \rcol{i}{1},q_{ij}^m\rangle , \langle \rcol{i}{2},q_{ij}^m\rangle \end{pmatrix}, \\
C(g_mR_j,R_i)&=\begin{pmatrix}\cos{\alpha_{ji}^m } , \sin{\alpha_{ji}^m} \end{pmatrix}=
\begin{pmatrix}\langle g_m\rcol{j}{1},q_{ij}^m\rangle , \langle g_m\rcol{j}{2},q_{ij}^m\rangle \end{pmatrix},
\end{split}
\end{equation}
for $m=1,2,3,4$. Thus, the common line property for $D_2$-symmetric molecules implies that
\begin{equation}\label{Sec2:ClsAmb}
\hat{P}_{R_i}(\xi\cos\alpha^m_{ij},\xi\sin\alpha^m_{ij})=\hat{P}_{R_j}(\xi\cos\alpha^m_{ji},\xi\sin\alpha^m_{ji}),\quad \xi\in\mathbb{R},
\end{equation}
for $m=1,2,3,4$.
Throughout the following sections, we refer to all four common lines as the common lines of the images $P_{R_i}$ and $P_{R_j}$.

An additional feature of symmetric molecules, and in particular of $D_2$-symmetric molecules, is self common lines. For any $i\in[N]$, the images $P_{R_i},P_{g_2R_i},P_{g_3R_i}$ and $P_{g_4R_i}$ of a $D_2$-symmetric molecule, are identical. However, each image corresponds to a different projection plane, and thus, the projection plane of $P_{R_i}$ has a unique common line with each of the planes of  $P_{g_2R_i},P_{g_3R_i}$ and $P_{g_4R_i}$. The directions of these common lines are given by
\begin{equation}\label{Sec2:PhysSelfCls}
q_{ii}^m=\frac{\rcol{i}{3}\times g_m\rcol{i}{3}}{\|\rcol{i}{3}\times g_m\rcol{i}{3}\|},\quad m=2,3,4.
\end{equation}
Thus, there are three different pairs of central lines in $\hat{P}_{R_i}$, corresponding to the common lines between the pairs $\{P_{R_i},P_{g_2R_i}\},\{P_{R_i},P_{g_3R_i}\}$ and $\{P_{R_i},P_{g_4R_i}\}$, which satisfy
\begin{equation}\label{Sec2:SelfClsAmb}
\hat{P}_{R_i}(\xi\cos\alpha_{ii}^{1m},\xi\sin\alpha_{ii}^{1m})=\hat{P}_{R_i}(\xi\cos\alpha_{ii}^{m1},\xi\sin\alpha_{ii}^{m1}),\:\; \xi\in\mathbb{R},
\end{equation}
for $m=2,3,4$, where
\begin{equation}\label{Sec2:SelfClsCoor}
\begin{split}
C(R_i,g_mR_i)&=\begin{pmatrix}\cos\alpha_{ii}^{1m}, \sin\alpha_{ii}^{1m} \end{pmatrix}=\begin{pmatrix}\langle \rcol{i}{1},q_{ii}^m\rangle , \langle \rcol{i}{2},q_{ii}^m\rangle \end{pmatrix},\\
C(g_mR_i,R_i)&=\begin{pmatrix}\cos\alpha_{ii}^{m1} , \sin\alpha_{ii}^{m1} \end{pmatrix}=\begin{pmatrix}\langle g_m\rcol{i}{1},q_{ii}^m\rangle , \langle g_m\rcol{i}{2},q_{ii}^m\rangle \end{pmatrix},
\end{split}
\end{equation}
for $m=2,3,4$. We refer to all three common lines in~\eqref{Sec2:SelfClsAmb} as the self common lines of the image $P_{R_i}$.

\section{Related work}\label{sec:previous}
Many of the current common line methods for orientation assignment are based on the method of angular reconstitution by Van Heel \cite{angRecon}.
The core idea of the angular reconstitution method is rooted at the observation that the intersection of any three non-coinciding central planes establishes the orientations of all three, relative to each other. In particular, Van Heel shows how given a triplet of projection images $\{P_{R_i},P_{R_j},R_{R_k}\}$, one can obtain either of the sets of relative rotation matrices
$\{R_i^TR_j, R_i^TR_k, R_j^TR_k\}$ or $\{JR_i^TR_jJ, JR_i^TR_kJ, JR_j^TR_kJ\}$, by using the common lines between $P_{R_i},P_{R_j}$ and $P_{R_k}$.
Both choices of relative rotation matrices are equally consistent with the images $P_{R_i},P_{R_j}$ and $P_{R_k}$, and are just the manifestation of the handedness ambiguity discussed in the previous section. The angular reconstitution method then makes the assumption that (without loss of generality) $R_i=I$, which immediately establishes $R_j$ and $R_k$ from $R_i^TR_j$ and $R_i^TR_k$.
The orientations of the rest of the images $P_{R_l}$  for $l\neq i,j,k$ are obtained by fixing the pair of images $P_{R_i}$ and $P_{R_j}$, and applying the same method sequentially to each triplet of images $\{P_{R_i},P_{R_j},P_{R_l}\}$ to retrieve $R_i^TR_l$, which immediately establishes $R_l$ by $R_l=R_i^TR_l$. Note that since all relative rotations and subsequently the rotations themselves were obtained by combining each of the images $R_l$ with the same pair of images $\{P_{R_i},P_{R_j}\}$, the angular reconstitution method ensures that we obtain a hand-consistent assignment, i.e.\ we either obtain $\{R_i\}_{i=1}^N$ or $\{JR_iJ\}_{i=1}^N$. Thus, we can recover either the original molecule or its mirror image.

The most commonly used procedure for estimating the common line of a pair of images $P_{R_i}$ and $P_{R_j}$, is to calculate their Fourier transforms and then find a maximally correlated pair of central lines between the transformed images, see e.g.~\cite{voting}.
As the images obtained by cryo-EM are contaminated with high levels of noise,
in practice the probability of correctly detecting their common lines, and subsequently their rotations using the angular reconstitution method, is low.
In~\cite{sync}, Shkolnisky and Singer describe an algorithm for estimating the rotations which achieves robustness to noise by employing an approach known as 'synchronization'. In this approach, the rotations $\{R_i\}_{i=1}^N$ are estimated using all the relative rotations $\{R_i^TR_j\}_{i<j\in[N]}$ together at once.
The authors use the set $\{R_i^TR_j\}_{i<j\in[N]}$ to construct a $3N\times3N$ block matrix $M$ known as a 'synchronization matrix', whose $(i,j)^{th}$ block $M_{ij}$ of size $3\times3$ is given by $R_i^TR_j$, that is
\begin{equation}
M_{ij}=R_i^TR_j \:,\quad  i,j\in[N].
\end{equation}
Defining the matrix $U=\left(R_1,\ldots,R_N\right)$, we see that
\begin{equation}\label{eq:M}
M=U^TU,
\end{equation}
and thus we can obtain $U=(R_1,\ldots,R_N)$ by factoring $M$ using SVD.

A method for estimating the set of relative rotations $\{R_i^TR_j\}_{i<j\in[N]}$ in a non-sequential manner is given in~\cite{voting}, and takes advantage of the following observation:
the relative rotation of $P_{R_i},P_{R_j}$ can be estimated from the common lines between these images and any of the $N-2$ images $P_{R_k}$, where $k\neq i,j$.
The authors of~\cite{voting} show how to obtain a robust estimate of the relative rotation of each pair $P_{R_i}$ and $P_{R_j}$, by taking a majority vote over all these $N-2$ estimates.
Note that in order to construct the synchronization matrix $M$, one has to obtain a hand-consistent set of relative rotations, i.e.\ either the set $\{R_i^TR_j\}_{i<j\in[N]}$ or the set
$\{JR_i^TR_jJ\}_{i<j\in[N]}$. However, if the estimates $R_i^TR_j$ or $JR_i^TR_jJ$ are obtained independently for each pair of images, this cannot be guaranteed.
A solution to this issue is given in~\cite{Jsync}.

A direct application of any of the common lines based methods described above to a symmetric molecule encounters substantial difficulties stemming from the ambiguity described by~\eqref{Sec2:ClsAmb}. Suppose that given a pair of images $P_{R_i}$ and $P_{R_j}$ of a $D_2$-symmetric molecule, we wish to estimate the relative rotation $R_i^TR_j$ by applying the angular reconstitution method. By~\eqref{Sec2:ClsAmb}, we can detect four different common lines between the images, corresponding to four different pairs of projection planes, and we have no way of knowing which common line corresponds to which pair of planes.
Thus, combining the common line of $P_{R_i}$, $P_{R_j}$ together with common lines with
$P_{R_k}$, gives rise to $4^3$ different possible combinations of common line triplets, many of which do not submit an intersection of three planes. For instance, we can erroneously consider a combination of the common lines between the pairs $\{P_{R_i},P_{R_j}\}$, $\{P_{R_i},P_{R_k}\}$, and a third pair $\{P_{R_j},P_{g_2R_k}\}$, from which we cannot establish the relative rotations of $P_{R_i},P_{R_j}$ and $P_{R_k}$, since we are not considering the correct common lines triplet between the projection planes of these images.

In~\cite{voting}, the authors derive a simple condition by which one can determine whether a triplet of common lines can be realized as the intersection of three central planes.
Still, even common line triplets which do satisfy this condition can generate any of the rotations $\{R_i^TR_j,R_i^Tg_2R_j,R_i^Tg_3R_j,R_i^Tg_4R_j\}$, between which we cannot distinguish. Thus, to construct the synchronization matrix $M$ in~\eqref{eq:M}, one would have to devise a way to obtain a set of estimates
$\{\widetilde{R}_i^T\widetilde{R}_j\}_{i<j\in[N]}$, in which for each $i\in[N]$ all the relative rotations $\widetilde{R}_i^T\widetilde{R}_j$ for $j\neq i$ collectively 'agree' on the identity of $\widetilde{R}_i\in \{R_i,g_2R_i,g_3R_i,g_4R_i\}$.

A further difficulty stems from the fact that pairs of central lines which are adjacent to each common line between a pair of images are also highly correlated. Thus, attempting to estimate 4 common lines between a pair of noisy images $\hat{P}_{R_i}$ and $\hat{P}_{R_j}$ by simply trying to detect the 4 best correlated central lines is most likely to fail (see~\cite{cn} for a detailed discussion).

In Section~\ref{sec:relative rotations} we present a different approach for estimating all four common lines and respective relative rotations of a pair of projection images of a $D_2$-symmetric molecule, inspired by maximum likelihood methods.
In Section~\ref{sec:rotation matrices} we outline an algorithm for extracting the rotations $R_i$ of~$\eqref{Intro:ProjEq}$ from the relative rotations estimated by the procedure described in Section~\ref{sec:relative rotations}. This approach encounters 3 major obstacles which are resolved in Sections~\ref{sec:handedness},~\ref{sec:rows synchronization} and~\ref{sec:signs synchronization}. In Section~\ref{sec:experiments} we demonstrate the applicability of our method to experimental cryo-EM data. Finally, in Section~\ref{sec:summary} we summarize and discuss future work.

\section{Relative rotations estimation}\label{sec:relative rotations}
In this section we present a method for estimating the set of relative rotations $\{R_i^Tg_mR_j\}_{m=1}^4$ for a pair of images $P_{R_i}$ and $P_{R_j}$ of a $D_2$-symmetric molecule. Similarly to the maximal correlations approach described in the previous section, we begin by computing the 2D Fourier transform of each image $P_{R_i}$. By~\eqref{Sec2:ClsAmb}, in the noiseless case, each pair of transformed images $\hat{P}_{R_i}$ and $\hat{P}_{R_j}$ has exactly four pairs of perfectly correlated central lines. Thus, in principle, we can detect these common lines by computing correlations between pairs of central lines in $\hat{P}_{R_i}$ and $\hat{P}_{R_j}$, and choosing the four maximally correlated pairs. However, as was explained in the previous section, this approach encounters several substantial difficulties. We now describe a different approach inspired by maximum likelihood methods.

Consider the set
\begin{equation}
\mathcal{D}_c=\{\{Q_l^Tg_mQ_r\}_{m=1}^4\:|\:Q_l,Q_r\in SO(3)\:,\:|\!<Q_l^{(3)}\!,Q_r^{(3)}\!>\!|\neq1\},
\end{equation}
of quadruplets of relative rotations generated from all pairs of rotations $Q_l,Q_r\in SO(3)$ with non-coinciding beaming directions, and let us denote the members of $\mathcal{D}_c$ by  $Q_{lr}=\{Q_l^Tg_mQ_r\}_{m=1}^4$. Given a pair of images $P_{R_i}$ and $P_{R_j}$, we now show how one can use the common lines of the images to assign a score $\pi_{ij}(Q_l,Q_r)$ to each element $Q_{lr}\in\mathcal{D}_c$, which indicates how well it approximates the quadruplet $\{R_i^Tg_mR_j\}_{m=1}^4$.
Since $\{R_i^Tg_mR_j\}_{m=1}^4 \in \mathcal{D}_c$, it can be detected by searching over $D_c$ for a candidate $Q_{lr}$ with the best score $\pi_{ij}(Q_l,Q_r)$. We henceforth refer to $D_c$ as the relative rotations search space.

First, to relate each candidate $Q_{lr}\in\mathcal{D}_c$ to the common lines of $P_{R_i}$ and $P_{R_j}$, let us compute the vectors
\begin{equation}\label{Sec4:ClsPhys}
\tilde{q}_{lr}^m=\frac{Q_l^{(3)}\times g_mQ_r^{(3)}}{\|Q_l^{(3)}\times g_mQ_r^{(3)}\|},\quad m\in\{1,2,3,4\},
\end{equation}
analogously to~\eqref{Sec2:PhysCls}. If $Q_{lr}=\{R_i^Tg_mR_j\}_{m=1}^4$, then the set $\{\tilde{q}_{lr}^m\}_{m=1}^4$ corresponds to the direction vectors of the common lines of $P_{R_i}$ and $P_{R_j}$. We subsequently refer to $\{\tilde{q}_{lr}^m\}_{m=1}^4$ as the set of common lines directions of the quadruplet $Q_{lr}$ (corresponding to the pair $Q_l,Q_r\in SO(3)$).
Next, for each candidate $Q_{lr}$, we use the set $\{\tilde{q}_{lr}^m\}_{m=1}^4$ to compute the coordinate vectors
\begin{equation}\label{Sec4:ClsCoor}
C(Q_l,g_mQ_r)=\begin{pmatrix}\cos{\tilde{\alpha}_{lr}^m } , \sin{\tilde{\alpha}_{lr}^m} \end{pmatrix},\quad C(g_mQ_r,Q_l)=\begin{pmatrix}\cos{\tilde{\alpha}_{rl}^m} , \sin{\tilde{\alpha}_{rl}^m} \end{pmatrix},
\end{equation}
for $m=1,2,3,4$, analogously to~\eqref{Sec2:ClsCoor}. If $Q_{lr}=\{R_i^Tg_mR_j\}_{m=1}^4$, the coordinates in~\eqref{Sec4:ClsCoor} correspond to the local coordinates of the common lines of $P_{R_i}$ and $P_{R_j}$ on the respective projection planes of the images.
Denote by
\begin{equation*}
\nu_{n,\theta}(\xi)=\hat{P}_{R_n}(\xi\cos\theta,\xi\sin\theta),\quad \xi\in(0,\infty),
\end{equation*}
the half line (known as a Fourier ray) in the direction which forms an angle $\theta$ with the $x$-axis of the transformed image $\hat{P}_{R_n}$.
We then compute the normalized cross correlations
\begin{equation}
\rho_{ij}(\tilde{\alpha}_{lr}^m,\tilde{\alpha}_{rl}^m)=\frac{\int_0^\infty (\nu_{i,\tilde{\alpha}_{lr}^m}(\xi))^*\nu_{j,\tilde{\alpha}_{rl}^m}(\xi)d\xi}
{\|\nu_{i,\tilde{\alpha}_{lr}^m}(\xi)\|_{L_2}\|\nu_{j,\tilde{\alpha}_{rl}^m}(\xi)\|_{L_2}}\:,\quad m=1,2,3,4,
\end{equation}
of each pair of rays given by the direction vectors in~\eqref{Sec4:ClsCoor}.
We use rays instead of lines, since the correlation $\rho_{ij}(\theta,\varphi)$  between each pair of rays $\nu_{i,\theta}(\xi)$ and $\nu_{j,\varphi}(\xi)$ is identical to the correlation value $\rho_{ij}(\theta+\pi,\varphi+\pi)$ of their anti-podal rays. We then assign to each quadruplet $Q_{lr}$ the score
\begin{equation}\label{Sec4:ScoreFuncCls}
\pi_{ij}(Q_l,Q_r)=\prod_{m=1}^4\rho_{ij}(\tilde{\alpha}_{lr}^m,\tilde{\alpha}_{rl}^m).
\end{equation}
By~\eqref{Sec2:ClsAmb}, if $Q_{lr}=\{R_i^Tg_mR_j\}_{m=1}^4$ for some $Q_l,Q_r\in SO(3)$, then $\pi_{ij}(Q_l,Q_r)=1$. Thus,
the quadruplet $Q_{lr}$ with $\pi_{ij}(Q_l,Q_r)=1$ is declared as $\{R_i^Tg_mR_j\}_{m=1}^4$.

We remark that since in practice we use a discretization of $\mathcal{D}_c$ to estimate the relative rotations of pairs of noisy images (see Section~\ref{sec:implementation} for details), $\pi_{ij}$ is never exactly 1, and so we simply choose a candidate
$Q_{lr}$ which maximizes $\pi_{ij}(Q_l,Q_r)$ as an approximation for $\{R_i^Tg_mR_j\}_{m=1}^4$.

However, we can obtain more robust estimates to $\{R_i^Tg_mR_j\}_{m=1}^4$ by also combining self common lines into the score~\eqref{Sec4:ScoreFuncCls}.
As was explained in Section~\ref{sec:D2 geometry}, each image $P_{R_i}$ of a $D_2$-symmetric molecule has three self common lines, given by~\eqref{Sec2:SelfClsAmb}, which are the intersections of the projection plane of $P_{R_i}$ with the projection planes of $P_{g_2R_i}$, $P_{g_3R_i}$ and $P_{g_3R_i}$. We now show how to adjust the score $\pi_{ij}(Q_l,Q_r)$ of each candidate $Q_{lr}\in\mathcal{D}_c$ to account for the self common lines of each image in the pair $P_{R_i}$ and $P_{R_j}$.

For each candidate $Q_{lr}\in\mathcal{D}_c$, we first compute the vectors
\begin{equation}\label{Sec4:SclsPhys}
\tilde{q}_{ll}^m=\frac{Q_l^{(3)}\times g_mQ_l^{(3)}}{\|Q_l^{(3)}\times g_mQ_l^{(3)}\|},\quad
\tilde{q}_{rr}^m=\frac{Q_r^{(3)}\times g_mQ_r^{(3)}}{\|Q_r^{(3)}\times g_mQ_r^{(3)}\|}, \quad m=2,3,4,
\end{equation}
analogously to~\eqref{Sec2:PhysSelfCls}. If $Q_{lr}=\{R_i^Tg_mR_j\}_{m=1}^4$, then the sets $\{\tilde{q}_{ll}^m\}_{m=2}^4$ and $\{\tilde{q}_{rr}^m\}_{m=2}^4$ correspond to the directions of the self common lines of the images $P_{R_i}$ and $P_{R_j}$, respectively (see~\eqref{Sec2:PhysSelfCls}). Next, we use $\{\tilde{q}_{ll}^m\}_{m=2}^4$ and $\{\tilde{q}_{rr}^m\}_{m=2}^4$ of~\eqref{Sec4:SclsPhys} to compute the coordinates
\begin{equation}\label{Sec4:SclsCoor}
\begin{aligned}
C(Q_l,g_mQ_l)&=\left ( \cos{\tilde{\alpha}_{ll}^{1m}} , \sin{\tilde{\alpha}_{ll}^{1m}} \right ), \quad C(g_mQ_l,Q_l)&=\left ( \cos{\tilde{\alpha}_{ll}^{m1}} ,  \sin{\tilde{\alpha}_{ll}^{m1}} \right ),\\
C(Q_r,g_mQ_r)&=\left ( \cos{\tilde{\alpha}_{rr}^{1m}} , \sin{\tilde{\alpha}_{rr}^{1m}} \right ), \quad C(g_mQ_r,Q_r)&=\left ( \cos{\tilde{\alpha}_{rr}^{m1}} , \sin{\tilde{\alpha}_{rr}^{m1}} \right ),
\end{aligned}
\end{equation}
for $m=2,3,4$, analogously to~\eqref{Sec2:SelfClsCoor}. If $Q_{lr}=\{R_i^Tg_mR_j\}_{m=1}^4$, then the coordinates in~\eqref{Sec4:SclsCoor} correspond to the local coordinates of the self common lines of $P_{R_i}$ and $P_{R_j}$, on their respective projection planes (see~\eqref{Sec2:SelfClsCoor}).
We then compute the set of normalized autocorrelations
\begin{equation}
\begin{aligned}
\rho_{ii}(\tilde{\alpha}_{ll}^{1m},\tilde{\alpha}_{ll}^{m1})&=\frac{\int_0^\infty (\nu_{i,\tilde{\alpha}_{ll}^{1m}}(\xi))^*\nu_{i,\tilde{\alpha}_{ll}^{m1}}(\xi)d\xi}
{\|\nu_{i,\tilde{\alpha}_{ll}^{1m}}(\xi)\|_{L_2}\|\nu_{i,\tilde{\alpha}_{ll}^{m1}}\|_{L_2}},\\
\rho_{jj}(\tilde{\alpha}_{rr}^{1m},\tilde{\alpha}_{rr}^{m1})&=\frac{\int_0^\infty (\nu_{j,\tilde{\alpha}_{rr}^{1m}}(\xi))^*\nu_{j,\tilde{\alpha}_{rr}^{m1}}(\xi)d\xi}
{\|\nu_{j,\tilde{\alpha}_{rr}^{1m}}(\xi)\|_{L_2}\|\nu_{j,\tilde{\alpha}_{rr}^{m1}}\|_{L_2}},
\end{aligned}
\end{equation}
for $m\in\{2,3,4\}$, and if $Q_{lr}=\{R_i^Tg_mR_j\}_{m=1}^4$ for some $Q_l,Q_r\in SO(3)$, then by~\eqref{Sec2:SelfClsAmb} we have $\prod_{m=2}^4\rho_{ii}(\tilde{\alpha}_{ll}^{1m},\tilde{\alpha}_{ll}^{m1})\rho_{jj}(\tilde{\alpha}_{rr}^{1m},\tilde{\alpha}_{rr}^{m1})=1$. We therefore redefine the score $\pi_{ij}(Q_l,Q_r)$ in~\eqref{Sec4:ScoreFuncCls} to be
\begin{equation}\label{Sec4:CombinedScorePi}
\pi_{ij}(Q_l,Q_r)=\prod_{m=1}^4\rho_{ij}(\tilde{\alpha}_{lr}^m,\tilde{\alpha}_{rl}^m)
\prod_{m=2}^4\rho_{ii}(\tilde{\alpha}_{ll}^{1m},\tilde{\alpha}_{ll}^{m1})\rho_{jj}(\tilde{\alpha}_{rr}^{1m},\tilde{\alpha}_{rr}^{m1}).
\end{equation}
For each $i<j\in[N]$, we then set
\begin{equation}\label{Sec4:MaxProb}
Q^{ij}=\argmax_{Q_{lr}\in \mathcal{D}_c}\pi_{ij}(Q_l,Q_r).
\end{equation}
to be $\{R_i^Tg_mR_j\}_{m=1}^4$.

We remark, that as was explained in Section~\ref{sec:introduction} and Section~\ref{sec:previous} (see~\eqref{Intro:Jamb}), the images $\hat{P}_{R_n}$ and $\hat{P}_{JR_nJ}$ are identical. Thus, the self common lines of each image $\hat{P}_{R_n}$ are identical to the self common lines of $\hat{P}_{JR_nJ}$, and the common lines between each pair of images $\hat{P}_{R_i}$ and $\hat{P}_{R_j}$ are identical to the common lines between $\hat{P}_{JR_iJ}$ and $\hat{P}_{JR_jJ}$. Hence, for each $i<j\in[N]$, the set $\{Q_l^Tg_mQ_r\}_{m=1}^4\in \mathcal{D}_c$ which maximizes the score $\pi_{ij}$  has the same score as $\{JQ_l^Tg_mQ_rJ\}_{m=1}^4\in \mathcal{D}_c$. Thus, in~\eqref{Sec4:MaxProb}, for each $i<j\in[N]$, we either estimate $\{R_i^Tg_mR_j\}_{m=1}^4$ or $\{JR_i^Tg_mR_jJ\}_{m=1}^4$, independently from other pairs of $i$ and $j$.

The procedure for estimating of the sets of relative rotations for each $i<j\in[N]$ is summarized in Algorithm~\ref{alg:relative rotations estimation}.
\begin{algorithm}
	\caption{$D_2$ relative rotations estimation}
	\begin{algorithmic}[1]	
		\Require{A set of images $\hat{P}_{R_1},\ldots,\hat{P}_{R_N}$}, and a discretization of $SO(3)$ $Q_1,\ldots,Q_L\in SO(3)$
		\For{$l<r\in[L]$} \Comment{Compute common lines induced by $D_c$}
			\If{$|<Q_l^{(3)},Q_r^{(3)}>|\neq1$}
				\For{$m=1 \text{ to } 4$}
					\State{$Q_{lr}^m=Q_l^Tg_mQ_r$}
					\State{$\tilde{q}_{lr}^m=\frac{Q_l^{(3)}\times  g_mQ_r^{(3)}}{\|Q_l^{(3)}\times g_mQ_r^{(3)}\|}$} \Comment{See~\eqref{Sec4:ClsPhys}}
					\State{$(\cos\tilde{\alpha}_{lr}^m,\sin\tilde{\alpha}_{lr}^m)=(<Q_l^{(1)},\tilde{q}_{lr}^m>,<Q_l^{(2)},\tilde{q}_{lr}^m>)$} \Comment{See~\eqref{Sec4:ClsCoor}}
					\State{$(\cos\tilde{\alpha}_{rl}^m,\sin\tilde{\alpha}_{rl}^m)=(<g_mQ_r^{(1)},\tilde{q}_{lr}^m>,<g_mQ_r^{(2)},\tilde{q}_{lr}^m>)$}
				\EndFor
				\EndIf
		\EndFor

		\For{$l=1 \text{ to } L$}
			\For{$m=2 \text{ to } 4$}
				\State{$\tilde{q}_{ll}^m=\frac{Q_l^{(3)}\times g_mQ_l^{(3)}}{\|Q_l^{(3)}\times g_mQ_l^{(3)}\|}$} \Comment{See~\eqref{Sec4:SclsPhys}}
				\State{$(\cos{\tilde{\alpha}_{ll}^{1m}},\sin{\tilde{\alpha}_{ll}^{1m}})=(<Q_l^{(1)},\tilde{q}_{ll}^m>,<Q_l^{(2)},\tilde{q}_{ll}^m>)$} \Comment{See~\eqref{Sec4:SclsCoor}}
				\State{$(\cos{\tilde{\alpha}_{ll}^{m1}},\sin{\tilde{\alpha}_{ll}^{m1}})=(<g_mQ_l^{(1)},\tilde{q}_{ll}^m>,<g_mQ_l^{(2)},\tilde{q}_{ll}^m>)$}
			\EndFor
		\EndFor
		
		\For{$i<j\in[N]$}
			\State{$Q^{ij}=\argmax_{l<r\in[L],|<Q_l^{(3)},Q_r^{(3)}>|\neq1}\pi_{ij}(Q_l,Q_r)$} \Comment{See~\eqref{Sec4:CombinedScorePi}}
		\EndFor	
		\Ensure{$Q^{ij}$, for all $i<j\in[N]$.}
	\end{algorithmic}
    \label{alg:relative rotations estimation}
\end{algorithm}

\section{Estimating the rotation matrices}\label{sec:rotation matrices}
In the previous section, we have shown how to estimate for each $i<j\in[N]$ either $\{R_i^Tg_mR_j\}_{m=1}^4$ or $\{JR_i^Tg_mR_jJ\}_{m=1}^4$. In Section~\ref{sec:handedness} below, we will show how to resolve the handedness ambiguity. Thus, in this section we will assume w.l.o.g that we have the sets $\{R_i^Tg_mR_j\}_{m=1}^4$ for all $i<j\in[N]$, and outline how to recover the rotations $R_i$ in~\eqref{Intro:ProjEq} from these sets.

To recover the matrices $R_i$  row by row, we use the following observation. For each $m\in\{1,2,3\}$, denote by $I_m$ the $3\!\times\!3$ diagonal matrix
\begin{equation}\label{Sec5:ImDef}
(I_m)_{ij}=\begin{cases}
1 & i=j=m,\\
0 & \text{otherwise},
\end{cases}
\end{equation}	
and note that
\begin{equation}\label{Sec5:Sumgs}
\frac	{1}{2}(g_1+g_{l})=I_{l-1},\quad l=2,3,4,
\end{equation}
where $g_l$ were defined in~\eqref{Sec1:gpMem}.
Thus, for any pair of matrices $R_i$ and $R_j$ we have that
\begin{equation}\label{Sec5:ViVj}
\frac{1}{2}(R_i^TR_j+R_i^Tg_{m+1}R_j)=R_i^T\frac{1}{2}(g_1+g_{m+1})R_j=R_i^TI_mR_j=(v_i^m)^Tv_j^m,
\end{equation}
for $m=1,2,3$, where $v_i^m$ and $v_j^m$ are the $m^{\text{th}}$ rows of the matrices $R_i$ and $R_j$, respectively.
We can also compute the matrices $(v_i^m)^Tv_i^m$ for $m\in\{1,2,3\}$ and $i\in[N]$, by noting that since $v_j^m$ are rows of orthogonal matrices we have
\[(v_i^m)^Tv_i^m=(v_i^m)^Tv_j^m(v_j^m)^Tv_i^m,\quad j\in[N]\setminus\{i\}.\]
Since in practice the matrices $(v_i^m)^Tv_j^m$ are estimated from noisy images, we get a more robust
estimate for $(v_i^m)^Tv_i^m$ by using all $j$, that is, by setting
\begin{equation*}
(v_i^m)^Tv_i^m=\frac{\sum_{j\in[N]\backslash\{i\}}(v_i^m)^Tv_j^m(v_j^m)^Tv_i}{N-1},
\end{equation*}
for $m\in\{1,2,3\}$ and $i\in[N]$.

Next, for each $m\in\{1,2,3\}$, we construct the $3N\times3N$ matrix $H_m$ whose $(i,j)^{th}$ $3\times3$ block is given by the rank~1 matrix $(v_i^m)^Tv_j^m$, and note that
\begin{equation}\label{Sec5:HmDef}
H_m=v_m^Tv_m\:,\quad v_m=(v_1^m,\ldots,v_N^m)\:,\quad m=1,2,3.
\end{equation}
That is, $H_m$, $m=1,2,3$, are rank 1 matrices. We can now factorize each matrix $H_m$ using SVD, to obtain either the vector $v_m$ or $-v_m$,
hence retrieving either the set of rows $\{v_i^m\}_{i=1}^N$ or $\{-v_i^m\}_{i=1}^N$, for each $m\in\{1,2,3\}$.
Then, we can use these sets of rows to assemble the matrices $\{OR_i\}_{i=1}^N$ row by row, where $O\in O(3)$ is a diagonal matrix with $\pm1$ on its diagonal. If $\det{O}=-1$, we simply multiply all $OR_i$ by $-1$, and thus, we can assume w.l.o.g that $O$ is a rotation. The matrix $O$ is an inherent degree of freedom, since we can always ``rotate the world'' by any orthogonal matrix.


Unfortunately, the approach just described is not directly applicable, as we now explain.
Recall from Section~\ref{sec:previous}, that though we can recover the set of relative rotation matrices $\{R_i^Tg_mR_j\}_{m=1}^4$ from the common lines of $P_{R_i}$ and $P_{R_j}$, we have no way of knowing for each $m\in\{1,2,3,4\}$ which of the recovered matrices in the latter set is $R_i^Tg_mR_j$.
This implies, that for each $i<j\in[N]$, we can only obtain a permutation $(R_i^Tg_{\tau_{ij}(m)}R_j)_{m=1}^4$ of the 4-tuple $(R_i^Tg_mR_j)_{m=1}^4$ where $\tau_{ij}\in S_4$ is some unknown permutation of $(1,2,3,4)$. In~\eqref{Sec5:ViVj}, we computed the 3-tuples $((v_i^m)^Tv_j^m)_{m=1}^3$ by summing the first element of the 4-tuple $(R_i^Tg_mR_j)_{m=1}^4$ with the rest of its elements. Suppose for example, that we have the 4-tuple $(R_i^Tg_2R_j,R_i^Tg_3R_j,R_i^Tg_1R_j,R_i^Tg_4R_j)$ for a given a pair of images $P_{R_i}$ and $P_{R_j}$. One can easily verify by direct calculation that
\begin{equation}\label{Sec5:SumgsUnord}
\frac{1}{2}(g_{m_1}+g_{m_2})=-I_{m_3-1},\quad  (m_1,m_2,m_3)=\sigma(2,3,4),\quad \sigma\in S_3,
\end{equation}
where $S_3$ is the group of all permutations of a 3-tuple.
Now, observe that by~\eqref{Sec5:Sumgs} and~\eqref{Sec5:SumgsUnord} we have
\begin{equation*}
\begin{aligned}
&(\frac{1}{2}(R_i^Tg_2R_j+R_i^Tg_3R_j),\frac{1}{2}(R_i^Tg_2R_j+R_i^Tg_1R_j),\frac{1}{2}(R_i^Tg_2R_j+R_i^Tg_4R_j))\\
&\qquad=(R_i^T\frac{1}{2}(g_2+g_3)R_j,R_i^T\frac{1}{2}(g_2+g_1)R_j,R_i^T\frac{1}{2}(g_2+g_4)R_j)\\
&\qquad=(-R_i^TI_3R_j,R_i^TI_1R_j,-R_i^TI_2R_j)=(-(v_i^3)^Tv_j^3,(v_i^1)^Tv_j^1,-(v_i^2)^Tv_j^2).
\end{aligned}
\end{equation*}
This implies that the summation in~\eqref{Sec5:ViVj} of a permutation of the 4-tuple $(R_i^Tg_mR_j)_{m=1}^4$ results in a permutation of the respective 3-tuple $((v_i^m)^Tv_j^m)_{m=1}^3$ of rank 1 matrices, where some of the matrices have a spurious $-1$ factor. The following proposition, the proof of which is given Appendix~\ref{proof:Sec5:MainProp}, summarizes the effect of the aforementioned summation on a general permutation $(R_i^Tg_{\tau(m)}R_j)_{m=1}^4$ of a 4-tuple
$(R_i^Tg_mR_j)_{m=1}^4$, that is, on the order and signs of the respective 3-tuple $((v_i^m)^Tv_j^m)_{m=1}^3$.
\begin{proposition}\label{Sec5:MainProp}
Let $(R_i^Tg_{\tau(m)}R_j)_{m=1}^4$ for some $\tau\in S_4$ be a permutation of the $4$-tuple $(R_i^Tg_mR_j)_{m=1}^4$.
\begin{enumerate}
\item If $\tau(1)=1$, then the corresponding 3-tuple of rank 1 matrices is given by $((v_i^{m_r})^Tv_j^{m_r})_{r=1}^3$, where $(m_1,m_2,m_3)=(\tau(2)-1,\tau(3)-1,\tau(4)-1)$.
\item If $\tau(m)=1$ for $m>1$, then the corresponding 3-tuple of rank 1 matrices is given by 
\begin{equation*}
\begin{cases}
((v_i^{\tau(1)-1})^Tv_j^{\tau(1)-1},-(v_i^{\tau(4)-1})^Tv_j^{\tau(4)-1},-(v_i^{\tau(3)-1})^Tv_j^{\tau(3)-1}) & m=2,\\
(-(v_i^{\tau(4)-1})^Tv_j^{\tau(4)-1},(v_i^{\tau(1)-1})^Tv_j^{\tau(1)-1},-(v_i^{\tau(2)-1})^Tv_j^{\tau(2)-1}) & m=3,\\
(-(v_i^{\tau(3)-1})^Tv_j^{\tau(3)-1},-(v_i^{\tau(2)-1})^Tv_j^{\tau(2)-1},(v_i^{\tau(1)-1})^Tv_j^{\tau(1)-1}) & m=4.
\end{cases}
\end{equation*}
\end{enumerate}
\end{proposition}
Proposition~\ref{Sec5:MainProp} implies that we can only obtain the ordered triplets
\begin{equation}\label{Sec5.3Perms}
(\pm (v_i^{\sigma_{ij}(m)})^Tv_j^{\sigma_{ij}(m)})_{m=1}^3,\quad \sigma_{ij}\in S_3,\quad i<j\in [N],
\end{equation}
where the permutations $\sigma_{ij}\in S_3$ are unknown, and each of the matrices $(v_i^{\sigma_{ij}(m)})^Tv_j^{\sigma_{ij}(m)}$ is multiplied by $\pm 1$, which is also unknown and depends on $\tau_{ij}$.
Thus, we cannot construct the matrices $H_m$ in~\eqref{Sec5:HmDef} using the triplets in~\eqref{Sec5.3Perms} directly. We will show how to construct $H_m$ using the triplets~\eqref{Sec5.3Perms} in Sections~\ref{sec:rows synchronization} and~\ref{sec:signs synchronization}.

The following three sections are organized as follows. In Section~\ref{sec:handedness}, we show a method for handedness synchronization for $D_2$-symmetric molecules, which is adapted from a method proposed in~\cite{Jsync} for non-symmetric molecules. In Section~\ref{sec:rows synchronization}, we show how to partition the 3-tuples in~\eqref{Sec5.3Perms} into three sets of the form $\{s_{ij}^m(v_i^m)^Tv_j^m\}_{i<j\in[N]}$ for $m=1,2,3$ and some unknown signs $s_{ij}\in\{-1,1\}$.
Then, in Section~\ref{sec:signs synchronization} we show how to correct the signs $s_{ij}^m$ so that we can construct the $3N \times 3N$ matrices
\begin{equation}\label{Sec5:HmTildDef}
\widetilde{H}_m=(v_m)^Tv_m\:,\quad v_{m}=(s_1^mv_1^m,\ldots,s_N^mv^m_N)\:,\:m=1,2,3,
\end{equation}
where $s_i^m\in\{-1,1\}$ for $i\in[N]$ and $m\in\{1,2,3\}$. We can then factor the matrices $\widetilde{H}_m$ using SVD to obtain the vectors $v_m$ (of length $3N$), which give us the sets of rows $\{s_i^mv_i^m\}_{i=1}^N$ for $m\in\{1,2,3\}$, and assemble the matrices
\begin{equation*}
\hat{R}_i=\begin{pmatrix}
-s_i^1v_i^1-\\-s_i^2v_i^2-\\-s_i^3v_i^3-
\end{pmatrix}=
D_iR_i,\quad D_i=\textup{diag}(s_i^1,s_i^2,s_i^3),\quad i\in[N].
\end{equation*}
Since for each $i\in[N]$ either $D_i$ or $-D_i$ is in $\{g_m\}_{m=1}^4$, that is, $\hat{R}_i=\pm g_mR_i$ for some $m\in\{1,2,3,4\}$, we can compute the matrices
\begin{equation*}
\widetilde{R}_i=\begin{cases}
\:\:\;\hat{R}_i & \det(\hat{R}_i)=1,\\
-\hat{R}_i & \det(\hat{R}_i)=-1,
\end{cases}
\quad i\in[N],
\end{equation*}
by replacing all matrices $\hat{R}_i$ which have $\det(\hat{R}_i)=-1$ with $-\hat{R}_i$. The resulting set of matrices $\{\widetilde{R}_i\}_{i=1}^N$ satisfies $\widetilde{R}_i\in\{g_mR_i\}_{m=1}^4$ for all $i\in[N]$, and are therefore a solution for the orientation assignment problem which was stated at the end of Section~\ref{sec:introduction}.

\section{Handedness synchronization}\label{sec:handedness}
Following the discussion in the previous section, we now assume we have obtained a set of 4-tuples
\begin{equation}\label{Jsync:NonConSet}
\{(J^{\delta_{ij}}R_i^Tg_{\tau_{ij}(m)}R_jJ^{\delta_{ij}})_{m=1}^4\}_{i<j\in[N]},\quad \tau_{ij}\in S_4,\quad J=\text{diag}(1,1,-1),
\end{equation}
for some unknown $\delta_{ij}\in\{0,1\}$, by applying Algorithm~\ref{alg:relative rotations estimation}. We now explain how to extract one of the hand-consistent sets
\begin{equation}\label{Jsync:ConSet}
\{(R_i^Tg_{\tau_{ij}(m)}R_j)_{m=1}^4\}_{i<j\in[N]}\:\text{ or }\:\{(JR_i^Tg_{\tau_{ij}(m)}R_jJ)_{m=1}^4\}_{i<j\in[N]}
\end{equation}
from the set in~\eqref{Jsync:NonConSet}.

For all $i<j\in[N]$, we denote by
\begin{equation}\label{eq:Rij tuple}
R_{ij}\in\{(R_i^Tg_{\tau_{ij}(m)}R_j)_{m=1}^4,(JR_i^Tg_{\tau_{ij}(m)}R_jJ)_{m=1}^4\},\quad \tau_{ij}\in S_4,
\end{equation}
a 4-tuple of relative rotations consistent with a pair of images $P_{R_i}$ and $P_{R_j}$ of a $D_2$-symmetric molecule.
We also denote by $R_{ij}^m$, the $m^{th}$ relative rotation in the 4-tuple $R_{ij}$ for $m=1,2,3,4$, and by $JR_{ij}J$ the set $\{JR_{ij}^mJ\}_{m=1}^4$.
We now show how the set in~\eqref{Jsync:NonConSet} can be partitioned into two disjoint sets
\begin{equation}\label{Jsync:ClassDef}
\begin{split}
C_0&=\{R_{ij}|R_{ij}=(R_i^Tg_{\tau_{ij}(m)}R_j)_{m=1}^4\},\\
C_1&=\{R_{ij}|R_{ij}=(JR_i^Tg_{\tau_{ij}(m)}R_jJ)_{m=1}^4\}.
\end{split}
\end{equation}
Once we have the partition in~\eqref{Jsync:ClassDef}, we can compute the set of 4-tuples $\widetilde{C}_1=\{(JR_{ij}^mJ)_{m=1}^4|R_{ij}\in C_1\}$. Then, one of the hand-consistent sets in~\eqref{Jsync:ConSet} is given by $C_0\cup \widetilde{C}_1$.

The partition in~\eqref{Jsync:ClassDef} is derived by the following procedure. First, we construct a graph $\Sigma$ with vertices corresponding to the estimates $R_{ij}$ in~\eqref{eq:Rij tuple}, and with edges that encode which pairs of estimates $R_{ij}$ and $R_{kl}$ are in the same set in~\eqref{Jsync:ClassDef}, and which aren't (as will be explained shortly). Then, we derive the partition in~\eqref{Jsync:ClassDef} from the eigenvector of the leading eigenvalue of the adjacency matrix of $\Sigma$. The procedure we present is an adaptation of an algorithm that was derived in~\cite{Jsync} for non-symmetric molecules.

We next state a proposition, the proof of which is given in Appendix~\ref{proof:Jsync:MainProp}, which allows us to determine which estimates $R_{ij}$, $i<j\in[N]$, belong to the same set in~\eqref{Jsync:ClassDef}. Following the approach in~\cite{Jsync}, we look at triplets of estimates of the form $R_{ij},R_{jk},R_{ki}$ for all triplets $i<j<k\in[N]$, and determine which members of each such triplet are in the same set of~\eqref{Jsync:ClassDef} and which aren't.
\begin{proposition}\label{Jsync:MainProp}
For any $i<j<k\in[N]$ consider the triplet of estimates
\begin{equation}\label{Jsync:SameClassTrip}
R_{ij}=(R_i^Tg_{\tau_{ij}(m)}R_j)_{m=1}^4,\quad R_{jk}=(R_j^Tg_{\tau_{jk}(l)}R_k)_{l=1}^4,\quad R_{ki}=(R_k^Tg_{\tau_{ki}(r)}R_i)_{r=1}^4,
\end{equation}
that is, $R_{ij},R_{jk}$ and $R_{ki}$  are all in the same set of~\eqref{Jsync:ClassDef}. Then, exactly 16 of the $4^3$ matrix products in the set
\begin{equation}\label{Jsync:TripProds}
\{R_{ij}^mR_{jk}^rR_{ki}^l\:|\quad (m,l,r)\in\{1,2,3,4\}^3\}
\end{equation}
satisfy
\begin{equation}\label{Jsync:EstProds}
R_{ij}^mR_{jk}^lR_{ki}^r=I.
\end{equation}
\end{proposition}

Now, consider a triplet of estimates
\begin{equation}\label{Jsync:Trip}
\begin{split}
R_{ij}&\in\{(R_i^Tg_{\tau_{ij}(m)}R_j)_{m=1}^4,(JR_i^Tg_{\tau_{ij}(m)}R_jJ)_{m=1}^4\}, \\
R_{jk}&\in\{(R_j^Tg_{\tau_{jk}(l)}R_k)_{l=1}^4,(JR_j^Tg_{\tau_{jk}(l)}R_kJ)_{l=1}^4\},\\
R_{ki}&\in\{(R_k^Tg_{\tau_{ki}(r)}R_i)_{r=1}^4,(JR_k^Tg_{\tau_{ki}(r)}R_iJ)_{r=1}^4\},
\end{split}
\end{equation}
and note that since each estimate is either in the set $C_0$ or in the set $C_1$ of~\eqref{Jsync:ClassDef}, it must be that either all estimates are in the same set, or two estimates are in one set and the third estimate is in the other. We define the ``set configuration'' of a triplet $(R_{ij},R_{jk},R_{ki})$ by the row vector
\begin{equation}
d_{ijk}=\begin{cases}\label{Jsync:Config}
(0,0,0) & R_{ij},R_{jk},R_{ki}\text{ are in the same set of~\eqref{Jsync:ClassDef}},\\
(1,0,0) & R_{ij} \text{ is in a different set from } R_{jk}\text{ and }R_{ki},\\
(0,1,0) & R_{jk} \text{ is in a different set from } R_{ij}\text{ and }R_{ki},\\
(0,0,1) & R_{ki} \text{ is in a different set from } R_{ij}\text{ and }R_{jk},\\
\end{cases}
\end{equation}
and denote
\begin{equation}\label{Jsync:ConfigSet}
\mathcal{C}=\{c_0=(0,0,0),c_1=(1,0,0),c_3=(0,1,0),c_4=(0,0,1)\}.
\end{equation}
Note that if, for example, $R_{ij}$ is in one set of~\eqref{Jsync:ClassDef} and $R_{jk}$ and $R_{ki}$ are in another, then we have that $JR_{ij}J,R_{jk}$ and $R_{ki}$ are all in the same set.
We remark that we have found experimentally that whenever three estimates $R_{ij}$,$R_{jk}$ and $R_{ki}$ are not in the same set of~\eqref{Jsync:ClassDef}, then all the products in~\eqref{Jsync:EstProds} are far from $I$ in norm.
Thus, Proposition~\ref{Jsync:MainProp} suggests that we can find the set configuration of a triplet of estimates by the following procedure.
First, we compute the four sets of $4^3$ norms
\begin{equation}\label{Jsync:NormSets}
\begin{split}
&\mathcal{N}_{ijk}^0=\{\norm{R_{ij}^mR_{jk}^lR_{ki}^r-I}_F\: :\:(m,l,r)\in\{1,2,3,4\}^3\},  \\
&\mathcal{N}_{ijk}^1=\{\norm{JR_{ij}^mJR_{jk}^lR_{ki}^r-I}_F\: :\:(m,l,r)\in\{1,2,3,4\}^3\}, \\
&\mathcal{N}_{ijk}^2=\{\norm{R_{ij}^mJR_{jk}^lJR_{ki}^r-I}_F\: :\:(m,l,r)\in\{1,2,3,4\}^3\}, \\
&\mathcal{N}_{ijk}^3=\{\norm{R_{ij}^mR_{jk}^lJR_{ki}^rJ-I}_F\: :\:(m,l,r)\in\{1,2,3,4\}^3\}, \\
\end{split}
\end{equation}
where $\|\cdot\|_F$ is the Frobenius norm. Next, we sort the norms in each set $\mathcal{N}_{ijk}^0,\mathcal{N}_{ijk}^1,\mathcal{N}_{ijk}^2$ and $\mathcal{N}_{ijk}^3$ in~\eqref{Jsync:NormSets} in ascending order, and denote the resulting ascending sequences by $\mathcal{S}_{ijk}^0,\mathcal{S}_{ijk}^1,\mathcal{S}_{ijk}^2$ and $\mathcal{S}_{ijk}^3$, respectively. Finally, we compute the scores
\begin{equation}\label{Jsync:scores}
\hat{\mathcal{S}}_{ijk}^p=\sum_{n=1}^{16}(\mathcal{S}_{ijk}^p)_n,\quad  p=0,1,2,3.
\end{equation}
By Proposition~\ref{Jsync:MainProp}, there are exactly 16 norms with value 0 in the set $\mathcal{N}_{ijk}^p$ of~\eqref{Jsync:NormSets} which corresponds to the correct set configuration $d_{ijk}$ of a triplet $(R_{ij},R_{jk},R_{ki})$ (see~\eqref{Jsync:Config}). Thus,
we set $d_{ijk}=c_p$ for $p\in\{0,1,2,3\}$ such that $\hat{\mathcal{S}}_{ijk}^p$ is the minimal score in~\eqref{Jsync:scores}.

Once we have computed $d_{ijk}$ for all $i<j<k\in[N]$, we construct a graph $\Sigma$ whose vertices correspond to the estimates $R_{ij}$ in~\eqref{eq:Rij tuple}, and whose edges are defined by the $\binom{N}{2} \times \binom{N}{2}$ adjacency matrix (which we also denote by $\Sigma$)
\begin{equation}
\Sigma_{(i,j)(k,l)}=\begin{cases}
\quad 1 &\quad\text{if } |\{i,j\}\cap\{k,l\}|=1 \text{ and } R_{ij}\\
  &\quad \text{and } R_{kl} \text{ are in the same set of}~\eqref{Jsync:ClassDef},\\
\,-1 &\quad\text{if } |\{i,j\}\cap\{k,l\}|=1 \text{ and } R_{ij}\\
  &\quad\text{and } R_{kl} \text{ are in different sets of}~\eqref{Jsync:ClassDef},\\
\quad 0 &\quad\text{if } |\{i,j\}\cap\{k,l\}|\neq1.
\end{cases}
\end{equation}

Finally, we compute the eigenvector $u_s$ which corresponds to the leading eigenvalue of the matrix $\Sigma$. In~\cite{Jsync}, it is shown that the leading eigenvalue of $\Sigma$ is simple, and that $u_s$ is of the form $\{-1,1\}^{\binom{N}{2}}$ (up to normalization), where the sign of each entry encodes the set membership in~\eqref{Jsync:ClassDef} of each estimate $R_{ij}$. The procedure for handedness synchronization for $D_2$-symmetric molecules is summarized in Algorithm~\ref{alg:handedness synchronization}.
\begin{algorithm}
\caption{$D_2$ handedness synchronization}
\begin{algorithmic}[1]	
\Require{A set of $\binom{N}{2}$ 4-tuples $R_{ij}$ defined in~\eqref{eq:Rij tuple}}
\Initialize{$\binom{N}{2} \times \binom{N}{2}$ matrix $\Sigma$, with all entries set to zero}
\For{$i<j<k\in[N]$}
\For{$(m,l,r)\in\{1,2,3,4\}^3$}\Comment{See~\eqref{Jsync:NormSets}.}
\State{$\mathcal{N}_{ijk}^0(m,l,r)=\norm{R_{ij}^mR_{jk}^lR_{ki}^r-I}_F$}
\State{$\mathcal{N}_{ijk}^1(m,l,r)=\norm{JR_{ij}J^mR_{jk}^lR_{ki}^r-I}_F$}
\State{$\mathcal{N}_{ijk}^2(m,l,r)=\norm{R_{ij}^mJR_{jk}J^lR_{ki}^r-I}_F$}
\State{$\mathcal{N}_{ijk}^3(m,l,r)=\norm{R_{ij}^mR_{jk}^lJR_{ki}^rJ-I}_F$}
\EndFor
\EndFor
\For{$p=1$ to $4$}
\State{$\mathcal{S}_{ijk}^p=\textup{sort}(\mathcal{N}_{ijk}^p)$}\Comment{Sort in ascending order}
\State{$\hat{\mathcal{{S}}}_{ijk}^p=\sum_{n=1}^{16}(\mathcal{S}_{ijk}^p)_n$}
\EndFor
\For{$i<j<k\in[N]$}
\State{$m=\argmin\limits_{p\in\{0,1,2,3\}}\hat{\mathcal{S}}_{ijk}^p$}
\State{$d_{ijk}=c_m$} \Comment{See~\eqref{Jsync:Config},~\eqref{Jsync:ConfigSet}}
\State{$\Sigma_{(i,j),(j,k)}=(-1)^{\max((d_{ijk})_1,(d_{ijk})_2)}$}
\State{$\Sigma_{(j,k),(k,i)}=(-1)^{\max((d_{ijk})_1,(d_{ijk})_3)}$}
\State{$\Sigma_{(k,i),(i,j)}=(-1)^{\max((d_{ijk})_2,(d_{ijk})_3)}$}
\EndFor
\State{$\Sigma=\Sigma+\Sigma^T$}
\State{$u_s=\argmax\limits_{\|v\|=1}v^T\Sigma v$}\Comment{$u_s$ is the leading eigenvector of $\Sigma$}
\For{$i<j\in[N]$}
\If{$(u_s)_{ij}<0$}
\State{$R_{ij}=JR_{ij}J$}
\EndIf
\EndFor
\Ensure{$Rij$, for all $i<j\in[N]$.}
\end{algorithmic}
\label{alg:handedness synchronization}
\end{algorithm}

\section{Rotations' rows synchronization}\label{sec:rows synchronization}
At this point, in light of Sections~\ref{sec:rotation matrices} and~\ref{sec:handedness}, we assume that we have obtained one of the hand-consistent sets of 4-tuples in~\eqref{Jsync:ConSet}. Let us assume without loss of generality that we have the set $\{(R_i^Tg_{\tau_{ij}(m)}R_j)_{m=1}^4\}_{i<j\in[N]}$. As was explained in Section~\ref{sec:rotation matrices}, for each $i<j\in[N]$, we now form a 3-tuple of matrices by summing the first element of $(R_i^Tg_{\sigma_{ij}(m)}R_j)_{m=1}^4$ with each of the rest of its elements. By Proposition~\ref{Sec5:MainProp}, this results in a set of triplets
\begin{equation*}
\{(\pm(v_i^{\sigma_{ij}(m)})^Tv_j^{\sigma_{ij}(m)})_{m=1}^3\}_{i< j\in[N]},\quad \sigma_{ij}\in{S_3},
\end{equation*}
which was defined in~\eqref{Sec5.3Perms}, where $\sigma_{ij}$ are unknown and the $\pm1$ signs are also unknown. In this section, we will show how to partition this set of triplets into three disjoint sets
\begin{equation}\label{Clr.ClassPar}
C_m= \{s_{ij}^m(v_i^m)^Tv_j^m\}_{i<j\in[N]},\quad m\in\{1,2,3\},
\end{equation}
where $s_{ij}^m$ are the (unknown) signs of $(v_i^m)^Tv_j^m$. That is, for each $m\in\{1,2,3\}$, the set $C_m$ contains all outer products between the $m^{th}$ rows of the rotation matrices $R_i$ and $R_j$ for $i<j\in[N]$, up to sign. This partition will be obtained by casting it as a graph partitioning problem.

In Section~\ref{sec:graph partitioning}, we show how to encode the partition in~\eqref{Clr.ClassPar} as a graph in which each vertex corresponds to one of the matrices in~\eqref{Sec5.3Perms}. In Section~\ref{sec:constructing Omega}, we construct the adjacency matrix of the graph, and in Section~\ref{sec:unmixing}, we show how to extract the partition in~\eqref{Clr.ClassPar} from the leading eigenvectors of the graphs' adjacency matrix.

\subsection{Graph partitioning formulation}\label{sec:graph partitioning}
In what follows, we denote the $3 \times 3$ matrices in~\eqref{Clr.ClassPar} by
\begin{equation}\label{Clr.VertDef}
v_{ij}^m=s_{ij}^m(v_i^m)^Tv_j^m,\quad m\in\{1,2,3\}.
\end{equation}
We now construct a weighted graph $\Omega=(V,E)$ from which the partition in~\eqref{Clr.ClassPar} can be inferred. Each vertex in $V$ corresponds to one of the matrices in~\eqref{Clr.VertDef} (henceforth, we shall refer to both the matrix $s_{ij}^m(v_i^m)^Tv_j^m$ and its corresponding vertex in $V$ using the notation $v_{ij}^m$).
Thus, we have that (see~\eqref{Clr.VertDef})
\begin{equation}\label{Clr.Color.Class}
V=\cup_{m=1}^3C_m, \quad C_m=\{v_{ij}^m\}_{i<j\in[N]},\quad m=1,2,3.
\end{equation}
We define the set of weighted edges $E$ of the graph $\Omega$ by its $3\binom{N}{2} \times 3\binom{N}{2}$ adjacency matrix, which we also denote by $\Omega$, as follows
\begin{equation}\label{Clr.Omega.Def}
\Omega(v_{ij}^{m},v_{kl}^{r})=\begin{cases}
\quad 1 & |\{i,j\}\cap\{k,l\}|=1 \text{ and } m=r,\\
\, -1 & |\{i,j\}\cap\{k,l\}|=1 \text{ and } m\neq r,\\
\quad 0 & \text{otherwise}.
\end{cases}
\end{equation}
That is, we only connect by an edge vertices which have exactly one index in common. We give this edge a weight $+1$ if its incident vertices are in the same set $C_m$ of~\eqref{Clr.Color.Class}, and weight $-1$ otherwise (see Fig. \ref{fig:TriSameClass}).

Note that the weights on the edges $E$ of $\Omega$
induce a partition of the vertex set $V$ into the sets of~\eqref{Clr.Color.Class}, by grouping together vertices which are connected by edges with a weight of $+1$.
Recovering the partition in~\eqref{Clr.Color.Class} corresponds to coloring the vertices $V$ of $\Omega$ with 3 colors, say, red, green, and blue, where the vertices of $C_1$ are colored red, of $C_2$ green and of $C_3$ blue.

\begin{figure}[!tbp]
	\begin{subfigure}[b]{0.5\textwidth}
		\begin{center}
			\begin{tikzpicture}
			\path (0,0) node[circle,draw=red](x) {$v_{ij}^1$}
			(2,1)node[circle,draw=Green3](y) {$v_{jk}^2$}
			(2,-1) node[circle,draw=DodgerBlue2](z) {$v_{ki}^3$}
			(4,0) node[circle,draw=DodgerBlue2](w) {$v_{ij}^3$};
			\draw (x) -- (y) node[pos=0.5,above=0.07cm]{\textcolor{red}{-}1};
			\draw (x) -- (z) node[pos=0.5,below=0.1cm]{\textcolor{red}{-}1};
			\draw (y) -- (z) node[pos=0.5,left=0.05cm]{\textcolor{red}{-}1};
			\draw (y) -- (w) node[pos=0.5,above=0.1cm]{\textcolor{red}{-}1};
			\draw (z) -- (w) node[pos=0.5,below=0.1cm]{1};
			\end{tikzpicture}
			\caption{}
            \label{fig:nodes different classes}
		\end{center}	
	\end{subfigure}	
	\begin{subfigure}[b]{0.5\textwidth}
		\begin{center}
			\begin{tikzpicture}
			\path (0,1.5) node[circle,draw=Green3](x) {$v_{ij}^1$}
			(-1.5,0) node[circle,draw=Green3](y) {$v_{jk}^1$}
			(1.5,0) node[circle,draw=Green3](z) {$v_{ki}^1$};
			
			\draw (x) -- (y) node[pos=0.5,left=0.01cm]{1};
			\draw (y) -- (z) node[pos=0.5,below=0.01cm]{1};
			\draw (z) -- (x) node[pos=0.5,right=0.01cm]{1};
			\end{tikzpicture}
			\caption{}
            \label{fig:nodes same class}
		\end{center}
	\end{subfigure}		
	\caption{ (\subref{fig:nodes different classes}) Edges in $\Omega$. The vertices $v_{ik}^3$ and $v_{ij}^3$ are in the same set of~\eqref{Clr.Color.Class} and have the index $i$ in common. The vertices $v_{ij}^1,v_{ij}^3$ which have both indices $i$ and $j$ in common are disconnected. Vertices from different sets of~\eqref{Clr.Color.Class} with one index in common  are connected by edges with weight $-1$. (\subref{fig:nodes same class}) A triangle formed by vertices in the same set of~\eqref{Clr.Color.Class}.}
	\label{fig:TriSameClass}
\end{figure}
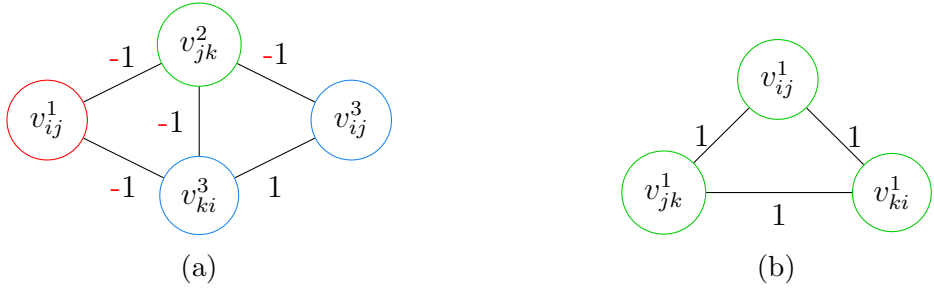

In what follows, we show that the partition of $V$ to the sets $C_m$ in~\eqref{Clr.Color.Class} can be derived from eigenvectors of the matrix $\Omega$.
\begin{definition}
	Let the eigenvalues of an $n\times n$ matrix A be $\lambda_1>\lambda_2>\ldots>\lambda_r$, with their respective multiplicities given by $n_1,\ldots,n_r$. We denote the spectrum of A by $\Lambda(A)$, and write
	\begin{equation*}
	\Lambda(A)=\begin{pmatrix}
	\lambda_1&\lambda_2&\cdots&\lambda_r\\n_1&n_2&\cdots& n_r
	\end{pmatrix}.
	\end{equation*}
\end{definition}
\begin{theorem}\label{Clr.MainThrm}
	The spectrum of the matrix $\Omega$ is given by
	\begin{equation}\label{Clr:OmegaSpec}
	\begin{pmatrix}4(N-2)& 2(N-4)& 2& -4& -(N-4)& -2(N-2)\\
	2& 2(N-1) & {\binom{N}{2}}-N& 2\left(\binom{N}{2}-N\right) & N-1&  1 \end{pmatrix}.
	\end{equation}
\end{theorem}
The proof of Theorem~\ref{Clr.MainThrm} is given in Appendix~\ref{proof:Clr.MainThrm}.

\begin{definition}\label{Clr.uAlphDef}
	Define $\alpha=(2{\binom{N}{2}})^{-\frac{1}{2}}$, $\beta=(6{\binom{N}{2}})^{-\frac{1}{2}}$, and define the pair of vectors $u_{\alpha},u_{\beta}\in\mathbb{R}^{3{\binom{N}{2}}}$ by
	\begin{equation*}
	u_{\alpha}(v_{ij}^{m})=\begin{cases}
	\begin{aligned}
	\alpha& & m=1,\\
	0& & m=2,\\
	-\alpha& & m=3,
	\end{aligned}
	\end{cases}\quad
	u_{\beta}(v_{ij}^{m})=\begin{cases}
	\begin{aligned}
	\beta& & m=1,\\
	-2&\beta & m=2,\\
	\beta& & m=3,
	\end{aligned}
	\end{cases}
	\end{equation*}
where for any $w\in\mathbb{R}^{3{\binom{N}{2}}}$ we denote by $w(v_{ij}^m)$ the entry of $w$ with the same index as the row of $\Omega$ which corresponds to the vertex $v_{ij}^m$.
\end{definition}

Throughout Section~\ref{sec:rows synchronization}, for any column vector $w\in \mathbb{R}^{3{\binom{N}{2}}}$, we denote
\begin{equation}\label{Clr:w_ij}
(w)_{ij}=(w(v_{ij}^{\sigma_{ij}(1)}),w(v_{ij}^{\sigma_{ij}(2)}),w(v_{ij}^{\sigma_{ij}(3)}))^T,\quad i<j\in[N].
\end{equation}
That is, $(w)_{ij}\in \mathbb{R}^3$ is the column vector that corresponds to the entries of the triplet $(v_{ij}^{\sigma_{ij}(1)},v_{ij}^{\sigma_{ij}(2)},v_{ij}^{\sigma_{ij}(3)})$ in $w$.

\begin{proposition}\label{Clr.MainProp}
The vectors $u_{\alpha}$ and $u_{\beta}$ in Definition~\ref{Clr.uAlphDef} are orthogonal eigenvectors of $\Omega$, corresponding to the eigenvalue $\mu_c=4(N-2)$.
\end{proposition}
We prove Proposition~\ref{Clr.MainProp} in Appendix~\ref{proof:Clr.MainProp}. The
immediate consequence of Theorem~\ref{Clr.MainThrm} and Proposition~\ref{Clr.MainProp} is the following corollary.
\begin{corollary}\label{Clr.MainCorl}
	The eigenspace of $\mu_c$ is spanned by $u_{\alpha}$ and $u_{\beta}$.
\end{corollary}

Note that $u_{\alpha}$ is a unit vector which exactly encodes the partition in~\eqref{Clr.Color.Class}, where the entries $+\alpha,0$ and $-\alpha$ encode the color of each vertex $v_{ij}^m$. The unit vector $u_\beta$ is 'color blind' in the sense that it can only distinguish between 2 colors. Obviously, in any 3-coloring of a graph we can always permute the colors, e.g., switch the color of all red vertices to green, green vertices to blue, and blue vertices to red. This is manifested in the following proposition.
\begin{proposition}
	Define the '3-color' and '2-color' vectors by
	\begin{equation}\label{Clr.Unmix.CVecs}
	u_{3c}=(\alpha,0,-\alpha)^T\;\textup{and}\; u_{2c}=(\beta,-2\beta,\beta)^T,
	\end{equation}
	respectively. For any $\sigma\in S_3$, we define the vectors $u_\alpha^\sigma,u_\beta^\sigma\in \mathbb{R}^{3{\binom{N}{2}}}$ by
	\begin{equation}\label{Clr.Unmix.uSig}
	u_\alpha^\sigma(v_{ij}^m)=\begin{cases}
	u_{3c}(\sigma(1)) & m=1,\\
	u_{3c}(\sigma(2)) & m=2,\\
	u_{3c}(\sigma(3)) & m=3,
	\end{cases}\quad
	u_\beta^\sigma(v_{ij}^m)=\begin{cases}
	u_{2c}(\sigma(1)) & m=1,\\
	u_{2c}(\sigma(2)) & m=2,\\
	u_{2c}(\sigma(3)) & m=3,
	\end{cases}
	\end{equation}
	where $u_\alpha^\sigma(v_{ij}^m)$ and $u_\beta^\sigma(v_{ij}^m)$ are the entries of $u_\alpha^\sigma$ and $u_\beta^\sigma$ with the same index as the row of $\Omega$ which corresponds to the vertex $v_{ij}^m$. Then, $u_\alpha^\sigma$ and $u_{\beta}^\sigma$ are orthogonal eigenvectors of $\Omega$ of~\eqref{Clr.Omega.Def} in the eigenspace of $\mu_c=4(N-2)$.
\end{proposition}

\begin{proof}
		Observe that $u_{\alpha}^\sigma$ is obtained from $u_{\alpha}$ of Definition \ref{Clr.uAlphDef} by replacing all $+\alpha$ entries with $u_{3c}(\sigma(1))$, all 0 entries with $u_{3c}(\sigma(2))$, and all $-\alpha$ entries with $u_{3c}(\sigma(3))$. The vector $u_{\beta}^\sigma$ is obtained from $u_{\beta}$ in a similar manner.
	The proposition then follows by repeating the method of proof applied in Proposition~\ref{Clr.MainProp} with $u_{\alpha}^\sigma$ and $u_{\beta}^\sigma$.

\end{proof}

Following Corollary~\ref{Clr.MainCorl}, we recover $u_\alpha$ (up to a color permutation, i.e, one of the vectors $u_\alpha^\sigma$ of \eqref{Clr.Unmix.uSig}) in the following manner. We begin by constructing the matrix~$\Omega$. Then, we compute a pair of orthogonal eigenvectors $v_a$ and $v_b$ spanning the eigenspace of $\mu_c$ (the leading eigenvalue of $\Omega$). In general, each of these eigenvectors is an orthogonal linear combination of $u_{\alpha}$ and $u_{\beta}$, and thus, we cannot read the partition in~\eqref{Clr.Color.Class} directly from any one of them. In Section~\ref{sec:unmixing}, we show how to 'unmix' $v_a$ and $v_b$ and retrieve $u_{\alpha}$. In practice, due to noise, we can only compute an approximation of $u_{\alpha}$, and thus, we never get the exact values $\alpha$, $0$ and $-\alpha$. We explain how to deal with this issue in Section~\ref{sec:unmixing}. In the following section, we show how to construct $\Omega$ of~\eqref{Clr.Omega.Def} using the set of matrices in~\eqref{Sec5.3Perms}.

\subsection{Constructing $\Omega$}\label{sec:constructing Omega}
We now derive a procedure for constructing the matrix $\Omega$ of~\eqref{Clr.Omega.Def}.
For any two pairs of indices $i<j\in[N]$ and $k<l\in[N]$, we denote by
\begin{equation}\label{Clr.Omega.BlockDef}
\Omega_{(i,j)(k,l)}=\begin{pmatrix}
\Omega(v_{ij}^{\sigma_{ij}(1)},v_{kl}^{\sigma_{kl}(1)}) &
\Omega(v_{ij}^{\sigma_{ij}(1)},v_{kl}^{\sigma_{kl}(2)}) &
\Omega(v_{ij}^{\sigma_{ij}(1)},v_{kl}^{\sigma_{kl}(3)})\\
\Omega(v_{ij}^{\sigma_{ij}(2)},v_{kl}^{\sigma_{kl}(1)}) &
\Omega(v_{ij}^{\sigma_{ij}(2)},v_{kl}^{\sigma_{kl}(2)}) &
\Omega(v_{ij}^{\sigma_{ij}(2)},v_{kl}^{\sigma_{kl}(3)})\\
\Omega(v_{ij}^{\sigma_{ij}(3)},v_{kl}^{\sigma_{kl}(1)}) &
\Omega(v_{ij}^{\sigma_{ij}(3)},v_{kl}^{\sigma_{kl}(2)}) &
\Omega(v_{ij}^{\sigma_{ij}(3)},v_{kl}^{\sigma_{kl}(3)})
\end{pmatrix}
\end{equation}
the $3\times3$ matrix given by the rows of $\Omega$ corresponding to the vertices $v_{ij}^{\sigma_{ij}(1)},v_{ij}^{\sigma_{ij}(2)}$ and $v_{ij}^{\sigma_{ij}(3)}$, and columns of $\Omega$ corresponding to the vertices $v_{kl}^{\sigma_{kl}(1)},v_{kl}^{\sigma_{kl}(2)}$ and $v_{kl}^{\sigma_{kl}(3)}$. By \eqref{Clr.Omega.Def} and \eqref{Clr.Omega.BlockDef}, we have
\begin{equation}\label{Clr.Omega.ZeroBlocksId}
\Omega_{(i,j),(k,l)}=0_{3\times3},\quad |\{i,j\}\cap\{k,l\}|\neq 1,
\end{equation}
where $0_{3\times3}$ is the $3 \times 3$ zero matrix. We will now show how to construct $\Omega$ block by block, by computing the blocks $\Omega_{(i,j),(k,l)}$ for which $|\{i,j\}\cap\{k,l\}|=1$.

The following lemma, the proof of which is given in Appendix~\ref{proof:Clr.Omega.IdxLemma}, characterizes the indices of the non-zero entries in $\Omega$ of~\eqref{Clr.Omega.Def}.

\begin{lemma}\label{Clr.Omega.IdxLemma}
	Define
    \begin{equation}\label{eq:def_of_A}
    	A=\{(i,j)(k,l)\;|\;|\{i,j\}\cap\{k,l\}|=1,\;i<j\in[N],\;k<l\in[N]\},\\	
    \end{equation}
    and for $i<j\in[N]$ define
    \begin{equation}\label{eq:def_of_A1_to_A4}
	\begin{alignedat}{4}
	A_{ij}^1&=&&\{(i,j)(k,j)\;|\;k<j,k\neq i\}, &\quad&  A_{ij}^2=\{(i,j)(j,k)\;|\;k>j\},\\
	A_{ij}^3&=&&\{(i,j)(k,i)\;|\;k<i\}, &\quad&	A_{ij}^4=\{(i,j)(i,k)\;|\;k>i,k\neq j\}.
	\end{alignedat}
    \end{equation}
	Moreover, for $i<j<k\in[N]$, define
	\begin{equation}\label{eq:A and Af}
	\begin{aligned}
	A_{ijk}&=\{(i,j)(j,k)\;,\;(i,j)(i,k)\;,\;(j,k)(i,k)\},\\
	A_{ijk}^f&=\{(j,k)(i,j)\;,\;(i,k)(i,j)\;,\;(i,k)(j,k)\}.
	\end{aligned}	
	\end{equation}
	Then, we have that
	\begin{equation}\label{Clr.Omega.IdxIdentity}
	A=\bigcup_{i<j\in[N]}A_{ij}^1\cup A_{ij}^2 \cup A_{ij}^3 \cup A_{ij}^4=\bigcup_{i<j<k\in[N]}A_{ijk}\cup A_{ijk}^f.
	\end{equation}
\end{lemma}
By \eqref{Clr.Omega.ZeroBlocksId} and the second equality in \eqref{Clr.Omega.IdxIdentity} of Lemma \ref{Clr.Omega.IdxLemma}, to construct $\Omega$, we only need to determine the blocks
\begin{equation}\label{Clr.Omega.SixBlocks}
\Omega_{(i,j)(j,k)},\; \Omega_{(i,j)(i,k)},\;\Omega_{(j,k)(i,k)},\; \Omega_{(j,k)(i,j)},\; \Omega_{(i,k)(i,j)},\;\Omega_{(i,k)(j,k)},\quad i<j<k\in[N],
\end{equation}
which we now show how to do.

For each triplet of indices $i<j<k\in[N]$, we consider the vertices of $\Omega$ corresponding to the pairs of indices $(i,j)$, $(j,k)$, and $(i,k)$, written in the rows of the table
\begin{equation}\label{Clr.TriTable}
\begin{tabular}{|l|c|r|}
\hline
$v_{ij}^{\sigma_{ij}(1)}$ & $v_{ij}^{\sigma_{ij}(2)}$ & $v_{ij}^{\sigma_{ij}(3)}$\\ \hline
$v_{jk}^{\sigma_{jk}(1)}$ & $v_{jk}^{\sigma_{jk}(2)}$ & $v_{jk}^{\sigma_{jk}(3)}$\\ \hline $v_{ik}^{\sigma_{ik}(1)}$ & $v_{ik}^{\sigma_{ik}(2)}$ & $v_{ik}^{\sigma_{ik}(3)}$\\
\hline
\end{tabular}
\end{equation}
For each pair of vertices from different rows in \eqref{Clr.TriTable}, we need to determine whether this pair belongs to the same set $C_m$ of \eqref{Clr.Color.Class} or to different sets. This corresponds to choosing between an edge with a weight of $+1$ or $-1$ for each of these pairs in~$\Omega$. We therefore show how to determine all edge weights between the vertices in~\eqref{Clr.TriTable} simultaneously. This procedure is then repeated for each triplet of indices $i<j<k\in [N]$.

First, observe that for any pair of matrices $v_{ij}^m$ and $v_{jk}^r$, given by \eqref{Clr.VertDef}, we have
\begin{equation}\label{Clr.VijVjk}
\begin{split}
v_{ij}^mv_{jk}^r&=\pm (v_i^m)^Tv_j^m(v_j^r)^Tv_k^r\\&=\pm<v_j^m,v_j^r>(v_i^m)^Tv_k^r=
\begin{cases}
\pm (v_i^m)^Tv_k^m & m=r,\\
0 & m\neq r,
\end{cases}
\end{split}
\end{equation}
since the row vectors $v_j^m$ and $v_j^r$ are rows of the orthogonal matrix $R_j$. This suggests that for each pair of vertices $v_{ij}^{\sigma_{jk}(m)}$ and $v_{jk}^{\sigma_{ij}(r)}$ with unknown $\sigma_{ij},\sigma_{jk}\in S_3$, we can determine whether they belong to the same set of \eqref{Clr.Color.Class}, by simply computing the norm of the product of the matrices which they represent.

Since in practice we work with noisy data, we next show how to get more robust estimates for the edge weights of $\Omega$, by leveraging the graph structure of $\Omega$ in conjunction with \eqref{Clr.VijVjk}. Denote by
\begin{equation}\label{Clr.Omega.TrpId}
v_{ji}^{m}=(v_{ij}^{m})^T, \quad m=1,2,3, \quad i<j\in[N],
\end{equation}
the transposed matrices of \eqref{Clr.VertDef}.
By \eqref{Clr.VijVjk} and \eqref{Clr.Omega.TrpId}, for each triplet of matrices $v_{ij}^m,v_{jk}^r$ and $v_{ki}^p=(v_{ik}^p)^T$, $i<j<k\in[N]$, we have
\begin{equation}\label{Clr.VijVjkVki}
v_{ij}^mv_{jk}^rv_{ki}^p=\begin{cases}
\pm (v_i^m)^Tv_i^m  & m=r=p,\\
0 & \text{otherwise},
\end{cases}
\end{equation}
and by \eqref{Clr.Omega.TrpId}, we also have that
\begin{equation}\label{Clr.VijVji}
v_{ij}^mv_{ji}^m=(v_i^m)^Tv_j^m(v_j^m)^Tv_i^m=(v_i^m)^Tv_i^m,\quad m=1,2,3.
\end{equation}
Note that a non-zero product of matrices in~\eqref{Clr.VijVjkVki} corresponds to a triplet of vertices in $\Omega$ in the same set of \eqref{Clr.Color.Class} (see Fig.~\ref{fig:nodes same class}).
We now infer which vertices in~\eqref{Clr.TriTable} belong to the same set of~\eqref{Clr.Color.Class}, by constructing a function that vanishes for all vertex triplets for which~\eqref{Clr.VijVjkVki} and~\eqref{Clr.VijVji} hold, namely, for all vertex triplets that belong to the same set of~\eqref{Clr.Color.Class}. Specifically, for each triplet of indices $i<j<k\in[N]$, we minimize the function (which will be explained shortly) $f_{ijk}:S_{3} \times S_{3} \to \mathbb{R}$ given by
\begin{equation}\label{Clr.OmegaConMax}
\begin{split}
f_{ijk}(\gamma,\delta)=\sum_{m=1}^3&\|v_{ij}^{\sigma_{ij}(m)}v_{jk}^{\sigma_{jk}(\gamma(m))}v_{ki}^{\sigma_{ik}(\delta(m))}\pm v_{ij}^{\sigma_{ij}(m)}v_{ji}^{\sigma_{ij}(m)}\|,
\end{split}
\end{equation}
over all $\gamma=(\gamma_1,\gamma_2,\gamma_3)$ and $\delta=(\delta_1,\delta_2,\delta_3)$ in $S_3$, and all choices of sign $\pm 1$ between the 2 terms in each norm (since by \eqref{Clr.VijVjkVki}, the sign of the right term in each norm is unknown), independently between the norms.


The rationale of minimizing~\eqref{Clr.OmegaConMax} can be demonstrated in the following manner. Writing down the vertices as in \eqref{Clr.TriTable}, we seek to rearrange the vertices in the second and third rows so that after rearrangement, the vertices in each column are in the same set of~\eqref{Clr.Color.Class} (see example in Fig.~\ref{fig:verTables}).
Whenever we choose a pair of permutations $\gamma^*,\delta^*\in S_3$ such that the vertices in each column in
\begin{equation}\label{Clr.TriTablePerm}
\begin{tabular}{|c|c|c|}
\hline
$v_{ij}^{\sigma_{ij}(1)}$ & $v_{ij}^{\sigma_{ij}(2)}$ & $v_{ij}^{\sigma_{ij}(3)}$\\ \hline
$v_{jk}^{\sigma_{jk}(\gamma^*(1))}$ & $v_{jk}^{\sigma_{jk}(\gamma^*(2))}$ & $v_{jk}^{\sigma_{jk}(\gamma^*(3))}$\\ \hline $v_{ik}^{\sigma_{ik}(\delta^*(1))}$ & $v_{ik}^{\sigma_{ik}(\delta^*(2))}$ & $v_{ik}^{\sigma_{ik}(\delta^*(3))}$\\
\hline
\end{tabular}
\end{equation}
are in the same set of~\eqref{Clr.Color.Class}, by~\eqref{Clr.VijVjkVki} and~\eqref{Clr.VijVji}, we have that all the terms in the sum~\eqref{Clr.OmegaConMax} equal zero. Otherwise, if there exists a column in~\eqref{Clr.TriTablePerm} in which there is a pair of vertices in different sets of~\eqref{Clr.Color.Class}, then we have that~\eqref{Clr.OmegaConMax} is strictly~$>0$. For example, if in each column of~\eqref{Clr.TriTablePerm}, there is a pair vertices in different classes, then by~\eqref{Clr.VijVjkVki}, the left term inside each norm in~\eqref{Clr.OmegaConMax} equals zero, while the right term in each norm equals $(v_{ij}^{\sigma_{ij}(m)})^Tv_{ji}^{\sigma_{ij}(m)}$, and we get that $f_{ijk}=\sum_{m=1}^3\|(v_{ij}^{\sigma_{ij}(m)})^Tv_{ji}^{\sigma_{ij}(m)}\|$.

\begin{figure}[!tbp]
	\begin{subfigure}[b]{0.5\textwidth}
		\begin{center}
			\begin{tabular}{|l|c|r|}
				\hline
				$v_{ij}^{2}$ & $v_{ij}^{3}$ & $v_{ij}^{1}$\\ \hline
				$v_{jk}^{3}$ & $v_{jk}^{1}$ & $v_{jk}^{2}$\\ \hline $v_{ik}^{2}$ & $v_{ik}^{1}$ & $v_{ik}^{3}$\\
				\hline
			\end{tabular}
		\end{center}
		\caption{}
        \label{tbl:unsync rows}
	\end{subfigure}
	\begin{subfigure}[b]{0.5\textwidth}
		\begin{center}
			\begin{tabular}{|l|c|r|}
				\hline
				$v_{ij}^{2}$ & $v_{ij}^{3}$ & $v_{ij}^{1}$\\ \hline
				$v_{jk}^{2}$&
				$v_{jk}^{3}$ & $v_{jk}^{1}$  \\ \hline
				$v_{ik}^{2}$ &
				$v_{ik}^{3}$ & $v_{ik}^{1}$  \\
				\hline
			\end{tabular}
		\end{center}
		\caption{}
        \label{tbl:sync rows}
	\end{subfigure}
	\caption{(\subref{tbl:unsync rows})~Example of unsynchronized rows, where $\sigma_{ij}=(2,3,1)$, $\sigma_{jk}=(3,1,2)$, and $\sigma_{ik}=(2,1,3)$. (\subref{tbl:sync rows})~The triplets in~(\subref{tbl:unsync rows}) after rearrangement of rows 2 and 3 in~\eqref{Clr.TriTablePerm}, with $\gamma^*=(3,1,2)$ and $\delta^*=(1,3,2)$. }
	\label{fig:verTables}
\end{figure}

Once we compute for each $i<j<k\in[N]$ a pair of permutations~$\gamma^*$ and~$\delta^*$ minimizing $f_{ijk}$ in \eqref{Clr.OmegaConMax}, the matrix $\Omega$ is set block by block by computing all the blocks of \eqref{Clr.Omega.SixBlocks} in the following manner. 
We first set the first three blocks of \eqref{Clr.Omega.SixBlocks}, that is,  $\Omega_{(i,j)(j,k)}$, $\Omega_{(i,j)(i,k)}$ and $\Omega_{(j,k)(i,k)}$. Consider \eqref{Clr.TriTablePerm}, which consists of all the vertices that are incident to the edges that constitute the aforementioned blocks (see \eqref{Clr.Omega.BlockDef}). The triplet of vertices $v_{ij}^{\sigma_{ij}(m)}, v_{jk}^{\sigma_{jk}(\gamma^*(m))}$ and $v_{jk}^{\sigma_{jk}(\delta^*(m))}$, which are all in the same column of \eqref{Clr.TriTablePerm}, are in the same set $C_{\sigma_{ij}(m)}$ of \eqref{Clr.Color.Class}, $m=1,2,3$. Thus, we assign a weight +1 to the edges of $\Omega_{(i,j)(j,k)}$, $\Omega_{(i,j)(i,k)}$, and $\Omega_{(j,k)(i,k)}$, which correspond to the pairs of vertices $(v_{ij}^{\sigma_{ij}(m)},v_{jk}^{\sigma_{jk}(\gamma^*(m))})$, $(v_{ij}^{\sigma_{ij}(m)},v_{ik}^{\sigma_{ik}(\delta^*(m))})$, and
$(v_{jk}^{\sigma_{jk}(\gamma^*(m))},v_{ik}^{\sigma_{ik}(\delta^*(m))})$, for $m=1,2,3$. All the edges of $\Omega_{(i,j)(j,k)}$, $\Omega_{(i,j)(i,k)}$, and $\Omega_{(j,k)(i,k)}$ which correspond to pairs of the form  $(v_{ij}^{\sigma_{ij}(m)},v_{jk}^{\sigma_{jk}(\gamma^*(r))})$, $(v_{ij}^{\sigma_{ij}(m)},v_{ik}^{\sigma_{ik}(\delta^*(r))})$, and
$(v_{jk}^{\sigma_{jk}(\gamma^*(m))},v_{ik}^{\sigma_{ik}(\delta^*(r))})$ where $m\neq r$, are assigned a weight $-1$ since they are in a different sets of \eqref{Clr.Color.Class}. As for the last three blocks of \eqref{Clr.Omega.SixBlocks}, that is,  $\Omega_{(j,k)(i,j)}$, $\Omega_{(i,k)(i,j)}$, and $\Omega_{(i,k)(j,k)}$, by \eqref{Clr.Omega.Def} we have that $\Omega(v_{ij}^m,v_{kl}^r)=\Omega(v_{kl}^r,v_{ij}^m)$ for all $i<j\in[N]$, $k<l\in[N]$, and $m,r\in\{1,2,3\}$. Thus, by \eqref{Clr.Omega.BlockDef} we have
\begin{equation}
\begin{split}
\Omega_{(k,l)(i,j)}&=\begin{pmatrix}
\Omega(v_{kl}^{\sigma_{kl}(1)},v_{ij}^{\sigma_{ij}(1)}) &
\Omega(v_{kl}^{\sigma_{kl}(1)},v_{ij}^{\sigma_{ij}(2)}) &
\Omega(v_{kl}^{\sigma_{kl}(1)},v_{ij}^{\sigma_{ij}(3)})\\
\Omega(v_{kl}^{\sigma_{kl}(2)},v_{ij}^{\sigma_{ij}(1)}) &
\Omega(v_{kl}^{\sigma_{kl}(2)},v_{ij}^{\sigma_{ij}(2)}) &
\Omega(v_{kl}^{\sigma_{kl}(2)},v_{ij}^{\sigma_{ij}(3)})\\
\Omega(v_{kl}^{\sigma_{kl}(3)},v_{ij}^{\sigma_{ij}(1)}) &
\Omega(v_{kl}^{\sigma_{kl}(3)},v_{ij}^{\sigma_{ij}(2)}) &
\Omega(v_{kl}^{\sigma_{kl}(3)},v_{ij}^{\sigma_{ij}(3)})
\end{pmatrix}\\
&=\begin{pmatrix}
\Omega(v_{ij}^{\sigma_{ij}(1)},v_{kl}^{\sigma_{kl}(1)}) &
\Omega(v_{ij}^{\sigma_{ij}(2)},v_{kl}^{\sigma_{kl}(1)}) &
\Omega(v_{ij}^{\sigma_{ij}(3)},v_{kl}^{\sigma_{kl}(1)})\\
\Omega(v_{ij}^{\sigma_{ij}(1)},v_{kl}^{\sigma_{kl}(2)}) &
\Omega(v_{ij}^{\sigma_{ij}(2)},v_{kl}^{\sigma_{kl}(2)}) &
\Omega(v_{ij}^{\sigma_{ij}(3)},v_{kl}^{\sigma_{kl}(2)})\\
\Omega(v_{ij}^{\sigma_{ij}(1)},v_{kl}^{\sigma_{kl}(3)}) &
\Omega(v_{ij}^{\sigma_{ij}(2)},v_{kl}^{\sigma_{kl}(3)}) &
\Omega(v_{ij}^{\sigma_{ij}(3)},v_{kl}^{\sigma_{kl}(3)})
\end{pmatrix}=(\Omega_{(i,j)(k,l)})^T.
\end{split}
\end{equation}
Thus, for every $i<j<k\in[N]$ it holds that
\begin{equation}\label{Clr.Omega.BlockTrpIdentity}
\Omega_{(j,k)(i,j)}=(\Omega_{(i,j)(j,k)})^T, \quad \Omega_{(i,k)(i,j)}=(\Omega_{(i,j)(i,k)})^T, \quad \Omega_{(i,k)(j,k)}=(\Omega_{(j,k)(i,k)})^T,
\end{equation}
and thus, $\Omega_{(j,k)(i,j)}$, $\Omega_{(i,k)(i,j)}$ and $\Omega_{(i,k)(j,k)}$ can be set according to \eqref{Clr.Omega.BlockTrpIdentity} after we compute $\Omega_{(i,j)(j,k)},\Omega_{(i,j)(i,k)}$ and $\Omega_{(j,k)(i,k)}$.

The procedure for constructing $\Omega$ of \eqref{Clr.Omega.Def} is summarized in Algorithm~\ref{alg:constructing Omega}. In the next section, we turn to the task of unmixing the eigenvectors corresponding to the maximal eigenvalue of $\Omega$, in order to extract $u_{\alpha}$ of Definition~\ref{Clr.uAlphDef}.

\begin{algorithm}
\caption{Constructing $\Omega$}
\begin{algorithmic}[1]	
\Require{The set of $\binom{N}{2}$ 3-tuples $\{(v_{ij}^{\sigma_{ij}(m)}=\pm(v_i^{\sigma_{ij}(m)})^Tv_j^{\sigma_{ij}(m)})_{m=1}^3\}_{i<j\in[N]}$}
\Initialize{$3{\binom{N}{2}}\times 3{\binom{N}{2}}$ matrix $\Omega$, with all entries set to zero}
\For{$i<j<k\in[N]$}
\State{$(\gamma^*,\delta^*)=\argmin\limits_{\gamma,\delta\in S_3}f_{ijk}(\gamma,\delta)$} \Comment{See \eqref{Clr.OmegaConMax}}
\For{$m=1 \text{ to } 3$} \Comment{Set $\Omega_{(i,j)(j,k)},\Omega_{(i,j)(i,k)}$ and $\Omega_{(j,k)(i,k)}$}
\State{$\Omega(v_{ij}^{\sigma_{ij}(m)},v_{jk}^{\sigma_{jk}(\gamma^*(m))})=1$
}
\State{$\Omega(v_{ij}^{\sigma_{ij}(m)},v_{ik}^{\sigma_{ik}(\delta^*(m))})=1$}
\State{$\Omega(v_{jk}^{\sigma_{jk}(\gamma^*(m))},v_{ik}^{\sigma_{ik}(\delta^*(m))})=1$}
\EndFor
\For{$m\neq r \in \{1,2,3\}$}
\State{$\Omega(v_{ij}^{\sigma_{ij}(m)},v_{jk}^{\sigma_{jk}(\gamma^*(r))})=-1$
}
\State{$\Omega(v_{ij}^{\sigma_{ij}(r)},v_{jk}^{\sigma_{jk}(\gamma^*(m))})=-1$
}
\State{$\Omega(v_{ij}^{\sigma_{ij}(m)},v_{ik}^{\sigma_{ik}(\delta^*(r))})=-1$}
\State{$\Omega(v_{ij}^{\sigma_{ij}(r)},v_{ik}^{\sigma_{ik}(\delta^*(m))})=-1$}
\State{$\Omega(v_{jk}^{\sigma_{jk}(\gamma^*(r))},v_{ik}^{\sigma_{ik}(\delta^*(m))})=-1$}
\State{$\Omega(v_{jk}^{\sigma_{jk}(\gamma^*(m))},v_{ik}^{\sigma_{ik}(\delta^*(r))})=-1$}

\EndFor

\State{$\Omega_{(j,k)(i,j)}=(\Omega_{(i,j)(j,k)})^T$} \Comment{See \eqref{Clr.Omega.BlockTrpIdentity}}
\State{$\Omega_{(i,k)(i,j)}=(\Omega_{(i,j)(i,k)})^T$}
\State{$\Omega_{(i,k)(j,k)}=(\Omega_{(j,k)(i,k)})^T$}
\EndFor
\end{algorithmic}
\label{alg:constructing Omega}
\end{algorithm}

\subsection{Unmixing the eigenvectors of $\mu_c$}\label{sec:unmixing}
By Corollary~\ref{Clr.MainCorl}, the eigenspace of $\Omega$ of~\eqref{Clr.Omega.Def} corresponding the eigenvalue $\mu_c=4(N-2)$ is spanned by $u_{\alpha}$ and $u_{\beta}$ of Definition~\ref{Clr.uAlphDef}. However, any orthogonal linear combination of $u_{\alpha}$ and $u_{\beta}$ is also an eigenvector, and so we can only compute two orthogonal eigenvectors which are linear combinations of $u_{\alpha}$ and $u_{\beta}$. In this section, we show how to 'unmix' these linear combinations to retrieve~$u_{\alpha}$.

Suppose that we have computed a pair of orthogonal unit eigenvectors
\begin{equation}\label{Clr.Unmix.VaVb}
v_a=a_1u_{\alpha}+a_2u_{\beta},\quad v_b=b_1u_{\alpha}+b_2u_{\beta},
\end{equation}
spanning the eigenspace of $\mu_c$.
Since $v_a$ and $v_b$ are unit vectors, by Proposition~\ref{Clr.MainProp} we have
\begin{equation}\label{Clr.Unmix.VaUnit}
1=\|v_a\|^2=<a_1u_{\alpha}+a_2u_{\beta},a_1u_{\alpha}+a_2u_{\beta}>=a_1^2+a_2^2,
\end{equation}
that is, $(a_1,a_2)^T$ is also a unit vector. Similarly, we have that $\|(b_1,b_2)^T\|=1$.
Furthermore, again by Proposition~\ref{Clr.MainProp}, we have
\begin{equation}\label{Clr.Unmix.VaVbOrth}
\begin{split}
0&=<v_a,v_b>=<a_1u_{\alpha}+a_2u_{\beta},b_1u_{\alpha}+b_2u_{\beta}>\\
&=a_1b_1+a_2b_2=<(a_1,a_2)^T,(b_1,b_2)^T>,
\end{split}
\end{equation}
i.e., the coefficients vectors $(a_1,a_2)^T$ and $(b_1,b_2)^T$ are unit orthogonal vectors.
Denoting
\begin{equation}\label{Clr.Unmix.RotMat}
R(\theta)=\begin{pmatrix}
\cos\theta&-\sin\theta\\ \sin\theta&\cos\theta
\end{pmatrix},\quad R^{ref}(\theta)=\begin{pmatrix}
-\cos\theta&\sin\theta\\ -\sin\theta&\cos\theta
\end{pmatrix},\quad \theta\in[0,2\pi),
\end{equation}
by~\eqref{Clr.Unmix.VaUnit} and~\eqref{Clr.Unmix.VaVbOrth} there exists an angle $\varphi\in[0,2\pi)$ such that either
\begin{equation}\label{Clr.Unmix.ab2theta}
\begin{pmatrix}
a_1 & b_1\\ a_2 & b_2
\end{pmatrix}=R(\varphi) \text{ or }
\begin{pmatrix}
a_1 & b_1\\ a_2 & b_2
\end{pmatrix}=R^{ref}(\varphi),
\end{equation}
and by~\eqref{Clr.Unmix.VaVb} and~\eqref{Clr.Unmix.ab2theta}, we have that
\begin{equation}\label{Clr.Unmix.Rot1}
(v_a\;\;v_b)=(u_\alpha \;\; u_\beta)R(\varphi) \;\text{ or }\; (v_a\;\;v_b)=(u_\alpha \;\; u_\beta)R^{ref}(\varphi).
\end{equation}
However, it can be easily verified that
\begin{equation}
(u_\alpha \;\; u_\beta)R^{ref}(\varphi)=(-u_\alpha \;\; u_\beta)R(\varphi),
\end{equation}
and thus,~\eqref{Clr.Unmix.Rot1} can be written as
\begin{equation}\label{Clr.Unmix.Rot}
(v_a\;\;v_b)=(u_\alpha \;\; u_\beta)R(\varphi) \;\text{ or }\; (v_a\;\;v_b)=(-u_\alpha \;\; u_\beta)R(\varphi).
\end{equation}
Equation~\eqref{Clr.Unmix.Rot} suggests that we can if we can recover $\varphi$, we can unmix $v_a$ and $v_b$ and recover either $u_\alpha$ or $-u_{\alpha}$. By Definition~\ref{Clr.uAlphDef}, $-u_{\alpha}$ is obtained from $u_{\alpha}$ by switching places between all $+\alpha$ and $-\alpha$ values in $u_{\alpha}$, and so both $u_\alpha$ and $-u_{\alpha}$ encode the same partition in~\eqref{Clr.Color.Class} of the vertices of $\Omega$. We now show how to find $\varphi$.

For any angle $\theta\in[0,2\pi)$, we write
\begin{equation}\label{eq:vtheta}
(v_a^{\theta}\; v_b^{\theta})=(v_a\;v_b)R(\theta).
\end{equation}
In addition, using the notation introduced in~\eqref{Clr:w_ij},
for any vector $w\in \mathbb{R}^{3{\binom{N}{2}}}$ and for all $i<j\in[N]$, we define
\begin{equation}\label{Clr.Unmix.SortVec}
M_{ij}(w)=\max\{(w)_{ij}\}, \quad  m_{ij}(w)=\min\{(w)_{ij}\},
\end{equation}
and we define by $d_{ij}(w)$ the value of $(w)_{ij}$ whose magnitude is between $M_{ij}$ and $m_{ij}$.
Then, we define the function $f_{c}:[0,2\pi) \to \mathbb{R}$ by
\begin{multline}\label{Clr.Unmix.Max}
f_c(\theta)=\sum_{i<j\in[N]} [(M_{ij}(v_a^{\theta})+m_{ij}(v_a^{\theta}))^2+d_{ij}(v_a^{\theta})^2]+\\
[(m_{ij}(v_b^{\theta})+2M_{ij}(v_b^{\theta}))^2+(m_{ij}(v_b^{\theta})+2d_{ij}(v_b^{\theta}))^2\\+(M_{ij}(v_b^{\theta})-d_{ij}(v_b^{\theta}))^2].
\end{multline}

The following proposition, the proof of which is given in Appendix~\ref{proof:Col.Unmix.MinUniq}, states that the minimum of~\eqref{Clr.Unmix.Max} over $\theta\in [0,2\pi)$ is obtained at $\theta$ for which $v_a^{\theta}=\pm u_{\alpha}$ and $v_b^{\theta}=u_{\beta}$ (up to a permutation of the vectors $u_{\alpha}$ and $u_{\beta}$ as in \eqref{Clr.Unmix.uSig}).
\begin{proposition}\label{Col.Unmix.MinUniq}
	Out of all orthogonal pairs of unit vectors in the eigenspace of $\mu_c$,
	the 12 pairs of vectors $\{(\pm u_{\alpha}^{\sigma}, u_{\beta}^{\sigma})\;|\;\sigma\in S_3\}$ are the unique minimizers of $f_c$ in~\eqref{Clr.Unmix.Max}, up to normalization.
\end{proposition}
Thus, if we denote by $(v_a^{\theta_*},v_b^{\theta_*})$ a minimizer of~\eqref{Clr.Unmix.Max}, then we declare $u_\alpha$ to be $v_a^{\theta^*}$. The partition in~\eqref{Clr.Color.Class} is then read from $u_\alpha$. In practice, due to noise in the input data, $v_a^{\theta^*}$ never exactly equals $u_\alpha$, and so after we compute $v_a^{\theta^*}$, we threshold its entries according to
\begin{equation}\label{Clr.Unmix.EntryThresh}
M_{ij}(v_a^{\theta^*})=\alpha,\quad d_{ij}(v_a^{\theta^*})=0,\quad m_{ij}(v_a^{\theta^*})=-\alpha,\quad i<j\in[N].
\end{equation}
The estimation of $u_\alpha$ is summarized in Algorithm~\ref{alg:ualpha}.

\begin{algorithm}
\caption{Estimation of the eigenvector $u_\alpha$ of $\Omega$}
\begin{algorithmic}[1]	
\Require{The matrix $\Omega$ of~\eqref{Clr.Omega.Def}}
\State{$v_a=\argmax\limits_{\|v\|=1}v^T\Omega v$} \Comment{$v_a$ and $v_b$, are orthogonal eigenvectors}
\State{$v_b=\argmax\limits_{\|v\|=1,v\bot u}v^T\Omega v$} \Comment{of the largest eigenvalue of $\Omega$}
\State{$(v_a^{\theta_*}\;v_b^{\theta_*})=\argmin_{\theta \in [0,2\pi)} f_c(\theta)$} \Comment{See~\eqref{Clr.Unmix.Max}}
\For{$i<j\in[N]$}
\State{$(M_{ij}(v_a^{\theta_*}),d_{ij}(v_a^{\theta_*}),m_{ij}(v_a^{\theta_*}))=(\alpha,0,-\alpha)$}\Comment{see \eqref{Clr.Unmix.EntryThresh}}
\EndFor
\Ensure{$v_a^{\theta^*}$}
\end{algorithmic}
\label{alg:ualpha}
\end{algorithm}

\section{Signs synchronization}\label{sec:signs synchronization}
Assuming we have obtained the partition in~\eqref{Clr.ClassPar} (see also~\eqref{Clr.Color.Class}), our final task is to adjust the signs $s_{ij}^m$ in each set $C_m$ of~\eqref{Clr.ClassPar}, so that we can construct the rank~1 matrices
$\widetilde{H}_m$ of~\eqref{Sec5:HmTildDef}. As was explained in Section 5, the matrices $\widetilde{H}_m$ can then be decomposed to retrieve all the rows of the matrices $R_i$ in~\eqref{Intro:Rots}, which can then be assembled from their constituent rows. Since all assertions we derive in this section apply identically and independently to each set $C_m$ in~\eqref{Clr.ClassPar}, throughout this section we refer to a single set of rank 1 matrices
\begin{equation}\label{SignsSync:Set}
\{s_{ij}v_i^Tv_j\}_{i<j\in[N]},\quad s_{ij}\in\{-1,1\},
\end{equation}
by dropping the superscript $m$ that indicates the set $C_m$ in~\eqref{Clr.ClassPar} to which the matrices belong. Our goal is therefore to estimate $v_{1},\ldots,v_{N}$ up to an arbitrary sign. To that end, we will adjust the signs $s_{ij}$ in~\eqref{SignsSync:Set} so that the matrix $\widetilde{H}$ of size $3N \times 3N$ whose $(i,j)^{th}$ block of size $3 \times 3$ is given by~\eqref{SignsSync:Set} has rank 1. The leading eigenvector of $\widetilde{H}$ will then give the vectors $v_{1},\ldots,v_{N}$ as required. We next describe the ``signs adjustment'' procedure and the construction of $\widetilde{H}$.

The task of adjusting the signs $s_{ij}$ in~\eqref{SignsSync:Set} consists of three steps. The first step of the signs adjustment procedure is computing the matrices $v_i^Tv_i$ for all $i\in[N]$ by observing that
\begin{equation}\label{SignsSync:v_ii}
(s_{ij}v_i^Tv_j) (s_{ij}v_i^Tv_j)^{T}= (s_{ij}v_i^Tv_j)(s_{ij}v_j^Tv_i)=v_i^Tv_i, \quad j\in[N]\backslash\{i\},
\end{equation}
since $s_{ij}\in \left \{-1,1 \right \}$ and so $s_{ij}^{2}=1$. For notational convenience, for all $i\in[N]$ we write $s_{ii}v_i^Tv_i$ instead $v_i^Tv_i$, since the `sign' $s_{ii}$ of $v_i^Tv_i$ equals 1.
In principle,~\eqref{SignsSync:v_ii} allows to compute $s_{ii}v_i^Tv_i$ for each $i\in[N]$ by using a single matrix $s_{ij}v_i^Tv_j$ for some arbitrarily chosen $j\in[N]\backslash\{i\}$. However, since in practice the input data is noisy, we obtain more robust estimates for $v_i^Tv_i$ by computing the averages
\begin{equation}\label{SignsSync:vii}
s_{ii}v_i^Tv_i=\frac{\sum_{j\in[N]\backslash\{i\}}(s_{ij}v_i^Tv_j)(s_{ij}v_i^Tv_j)^{T}}{N-1}, \quad i\in[N],
\end{equation}
followed by computing the best rank 1 approximation of~\eqref{SignsSync:vii} using SVD.

While $s_{ij}$ in~\eqref{SignsSync:Set} are defined only for $i<j$, for notational convenience we define $s_{ij}=s_{ji}$ whenever $i>j$, and as explained above $s_{ii}=1$. Thus, $s_{ij}$ are defined for all $i,j \in [N]$. We next outline steps 2 and 3 of the signs adjustment procedure, before giving a detailed description of these steps.
In step 2 of the procedure, we construct $N$ rank 1 matrices $H_1^s,\ldots,H_N^s$, which admit the decompositions
\begin{equation}\label{SignsSync:Hsn}
H_n^s=(v_n^s)^Tv_n^s\:,\quad v_n^s=(s_{n1}v_1,\ldots,s_{nN}v_N)\:,\quad n\in[N],
\end{equation}
for unknown $s_{nj}\in\{-1,1\}$, $j\in[N]$. For each $n\in[N]$, the matrix $H_n^s$ in~\eqref{SignsSync:Hsn} is constructed block by block from the matrices in~\eqref{SignsSync:Set} and~\eqref{SignsSync:vii}, in such a way that each $H_n^s$ is a rank 1 matrix which admits the decomposition~\eqref{SignsSync:Hsn}. This construction relies on Proposition~\ref{SignsSync:MainProp} stated below.
Each of the matrices $H_n^s$ in~\eqref{SignsSync:Hsn} can then be decomposed to recover all the rows $v_1,\ldots,v_N$ up to signs $s_{nj}$. Thus, in theory, we could construct only one of these matrices, say $H_1^s$, and recover $\{s_{11}v_1,\ldots,s_{1N}v_N\}$, which is our goal in this section. However, in practice, the estimates $s_{ij}v_i^Tv_j$ of~\eqref{SignsSync:Set} contain errors since they were estimated from noisy images. Moreover, for each particular $n\in[N]$, the set of estimates $\{s_{nj}v_n^Tv_j\}_{j\in[N]\backslash\{n\}}$ used to
construct the matrix $H_n^s$ critically depends on the common lines of the single noisy image $P_{R_n}$ in~\eqref{Intro:ProjEq} with each of the images $P_{R_j}$ for $j\in[N]\backslash\{n\}$, which due to the noise in the input images, may be highly inaccurate, leading to large errors in $(s_{11}v_1,\ldots,s_{1N}v_N)$ of~\eqref{SignsSync:Hsn}.

Thus, to use all available data in estimating $v_{1},\ldots,v_{N}$ of~\eqref{SignsSync:Hsn}, we first decompose all matrices $H_n^s$ in~\eqref{SignsSync:Hsn}, which results in $N$ independent estimates $\{s_{nj}v_n\}_{j=1}^N$ for each row $v_n$.
We then execute the third step of our signs adjustment procedure in which we
use all estimates $v_n^{s}$ together (see~\eqref{SignsSync:Hsn}) to obtain a set of signs $\widetilde{s}_{ij}\in\{-1,1\}$, which allows us to adjust the signs $s_{ij}$ of \eqref{SignsSync:Set} by multiplying each matrix $s_{ij}v_i^Tv_j$ by $\widetilde{s}_{ij}$, such that the matrix $\widetilde{H}$ of size $3N\times 3N$ whose $(i,j)^{th}$ $3\times3$ block is $(\widetilde{s}_{ij}s_{ij})v_i^Tv_j$ has rank 1. This latter procedure exploits all available data at once, improving the robustness of the estimation of the matrices $R_{i}$ of~\eqref{Intro:ProjEq} to noisy input data. This matrix $\widetilde{H}$ admits the decomposition
\begin{equation}\label{SignsSync:HTilde}
\widetilde{H}=(v^s)^Tv^s,\quad v^s=(s_1v_1,\ldots,s_nv_n)
\end{equation}
for some unknown signs $s_n\in\{-1,+1\}$. Recalling that we dropped the index $m\in\{1,2,3\}$ from $\widetilde{H}_m$ of \eqref{Sec5:HmTildDef}, and constructed $\widetilde{H}$ from the set in \eqref{SignsSync:Set}, we see that in fact we can construct $\widetilde{H}_m$ of~\eqref{Sec5:HmTildDef} using $\{s_{ij}^mv_i^Tv_j\}_{i<j\in[N]}$ for each $m\in\{1,2,3\}$, as required.
We can then decompose each $\widetilde{H}_m$, and recover the rotation matrices $R_{i}$ of~\eqref{Intro:ProjEq}, as was explained in Section 5, which is our task in this paper.
We now complete the details of the signs adjustment procedure described above, i.e., the construction of the matrices of~\eqref{SignsSync:Hsn} and~\eqref{SignsSync:HTilde}.

The following proposition, the proof of which is given in Appendix~\ref{proof:SignsSync:MainProp}, is the basis for the construction of the matrices of~\eqref{SignsSync:Hsn}.
\begin{proposition}\label{SignsSync:MainProp}
Let $H$ be a $3N \times 3N$ matrix whose $(i,j)^{th}$ block of size $3\times3$ is given by $s_{ij}v_i^Tv_j$ of~\eqref{SignsSync:Set} if $i<j$, by its transpose if $i>j$, and by $v_i^Tv_i$ of~\eqref{SignsSync:v_ii} if $i=j$. Then, H is rank 1 iff for each $n\in[N]$
\begin{equation}\label{SignsSync:CycleId}
s_{in}s_{nj}=s_{ij},
\end{equation}
where as noted above, for $i>j$ we define $s_{ij}=s_{ji}$. Furthermore, whenever~\eqref{SignsSync:CycleId} holds, we have
\begin{equation}\label{SignsSync:HnDecomp}
	H =(v_n^s)^Tv_n^s\:,\quad v_n^s=(s_{n1}v_1,\ldots,s_{nN}v_N), \quad n\in[N].
\end{equation}
\end{proposition}
Now, suppose we wish to construct $H_1^s$ of~\eqref{SignsSync:Hsn}. Proposition~\ref{SignsSync:MainProp}, and in particular~\eqref{SignsSync:CycleId}, suggest that we can construct $H_1^s$ using the set in~\eqref{SignsSync:Set}, by applying the following sign correction procedure. Recall that $v_i$ are rows of $3 \times 3$ orthogonal matrices, and thus $v_iv_i^T=1$ for all $i\in[N]$. For each pair $1<i<j\in[N]$, we compute the norm
\begin{equation}\label{SignsSync:SignIndNorm}
\|(s_{i1}v_i^Tv_1)(s_{1j}v_1^Tv_j)-s_{ij}v_i^Tv_j\|_F=\sqrt{|s_{i1}s_{1j}-s_{ij}|}\cdot\|v_i^Tv_j\|_F,
\end{equation}
and replace the matrix $s_{ij}v_i^Tv_j$ with $-s_{ij}v_i^Tv_j$ if~\eqref{SignsSync:SignIndNorm} is greater than zero. Let $\{\hat{s}_{ij}v_i^Tv_j\}_{i<j\in[N]}$ be the resulting set of rank 1 matrices after this signs correction procedure. By construction, it holds that
\begin{equation}\label{SignsSync:H1sCycle}
\hat{s}_{ij}=s_{i1}s_{1j}\:,\quad1<i<j\in[N].
\end{equation}
Now, let $H_1^s$ be the $3N\times3N$ matrix whose $(i,j)^{th}$ block of size $3\times3$ is given by $\hat{s}_{ij}v_i^Tv_j$. By~\eqref{SignsSync:H1sCycle} and Proposition~\ref{SignsSync:MainProp}, $H_1^s$ admits the decomposition
\begin{equation*}
	H_1^s=(v_1^s)^Tv_1^s,\quad v_1^s=(s_{11}v_1,\ldots,s_{1N}v_N).
 \end{equation*}
The matrices $H_2^s,\ldots,H_N^s$ are obtained in a similar manner.
In practice,~\eqref{SignsSync:SignIndNorm} is never exactly zero due to errors stemming from noise, and thus we also compute $\|(s_{i1}v_i^Tv_1)(s_{1j}v_1^Tv_j)+s_{ij}v_i^Tv_j\|_F$, and replace
$s_{ij}v_i^Tv_j$ with $-s_{ij}v_j^Tv_j$,  if
\begin{equation*}
\|(s_{i1}v_i^Tv_1)(s_{1j}v_1^Tv_j)-s_{ij}v_i^Tv_j\|_F>\|(s_{i1}v_i^Tv_1)(s_{1j}v_1^Tv_j)+s_{ij}v_i^Tv_j\|_F.
\end{equation*}

At this point, we can factor each of the matrices $\{H_n^s\}_{n=1}^N$ (e.g., using SVD), and obtain the set of vectors
\begin{equation}\label{SignsSync:SnVsn}
\hat{v}_n^s=s_n(s_{n1}v_1,\ldots,s_{nN}v_N)^T,\quad n\in[N],
\end{equation}
of~\eqref{SignsSync:Hsn}, where $s_n\in\{-1,+1\}$ are unknown. For each $i \in [N]$ we have
\begin{equation}\label{SignsSync:rowNotation}
(\hat{v}_{n}^s)_i=s_ns_{ni}v_i,\quad n\in[N],
\end{equation}
that is, we have $N$ estimates $\{s_ns_{ni}v_i\}_{n\in[N]}$ for the row $v_i$, where each estimate has an unknown sign $s_ns_{ni}$.
This concludes step 2 of the signs adjustment procedure outlined above.

For the third and final step of the signs adjustment procedure,
we define
\begin{equation}\label{SignsSync:s_ijTildeDef}
\widetilde{s}_{ij}=s_is_js_{ij},\quad i,j\in[N].
\end{equation}
Since $v_1,\ldots,v_N$ are unit row vectors, by~\eqref{SignsSync:rowNotation} and~\eqref{SignsSync:s_ijTildeDef}, we have
\begin{equation}\label{SignsSync:SijSjkTild}
\widetilde{s}_{ij}\widetilde{s}_{jk}=s_is_js_{ij}s_js_ks_{jk}=s_is_ks_{ij}s_{jk}=s_is_{ij}v_j(s_ks_{kj}v_j)^T=(\hat{v}_{i}^s)_j(\hat{v}_{k}^s)_j^T,
\end{equation}
for all $i\neq j\neq k\in[N]$. Thus, we can obtain all the products $\widetilde{s}_{ij}\widetilde{s}_{jk}$ in~\eqref{SignsSync:SijSjkTild} by taking dot products of the vectors in~\eqref{SignsSync:rowNotation}.

Now, suppose we computed the set $\{\widetilde{s}_{ij}\}_{i<j\in[N]}$ from the values $\widetilde{s}_{ij}\widetilde{s}_{jk}$ in~\eqref{SignsSync:SijSjkTild} (as will be explained shortly). Since $\widetilde{s}_{ij}=s_is_js_{ij}$, we have that
\begin{equation}\label{SignsSync:s_ijTildBys_ij}
\widetilde{s}_{ij}s_{ij}=s_is_j, \quad i<j\in[N].
\end{equation}
Thus, we can multiply each matrix $s_{ij}(v_i)^Tv_j$ in~\eqref{SignsSync:Set} by $\widetilde{s}_{ij}$, and obtain the set of matrices
$\{s_is_jv_i^Tv_j\}_{i<j\in[N]}$, and together with \eqref{SignsSync:vii}, we can construct  the $3N\times3N$ matrix $\widetilde{H}$, whose $(i,j)^{th}$ block of size $3\times3$ is given by
\begin{equation}\label{SignsSync:HTildeBlock}
(\widetilde{H})_{ij}=s_is_jv_i^Tv_j, \quad i,j\in[N].
\end{equation}
Then, $\widetilde{H}$ admits the decomposition in~\eqref{SignsSync:HTilde}, as required.
Thus, it only remains to show how to extract the set
$\{\widetilde{s}_{ij}\}_{i<j\in[N]}$ from the values $\{\widetilde{s}_{ij}\widetilde{s}_{jk}\}_{i\neq j\neq k\in[N]}$ in~\eqref{SignsSync:SijSjkTild}.

Let us define the ${\binom{N}{2}}\times{\binom{N}{2}}$ matrix
\begin{equation}\label{SignsSync:SignsSyncMat}
(S)_{(i,j)(k,l)}=
\begin{cases}
\widetilde{s}_{ij}\widetilde{s}_{kl} & |\{i,j\}\cap\{k,l\}|=1,\\
0 & \text{otherwise},
\end{cases}
\end{equation}
where $i<j\in[N]$, $k<l\in[N]$, and the products $\{\widetilde{s}_{ij}\widetilde{s}_{jk}\}_{i\neq j\neq k\in[N]}$ are computed using~\eqref{SignsSync:SijSjkTild}. The following proposition, the proof of which is given in Appendix~\ref{proof:SignsSync:SignsProp}, shows that the signs $\widetilde{s}_{ij}$ can be extracted from $S$ in~\eqref{SignsSync:SignsSyncMat}.
\begin{proposition}\label{SignsSync:SignsProp}
The leading eigenvalue of $S$ in~\eqref{SignsSync:SignsSyncMat} is $2(N-2)$ and it is simple. Moreover, define $u_s=(\widetilde{s}_{ij})_{i<j\in[N]}$ to be the vector of length $\binom{N}{2}$ with entries $\widetilde{s}_{ij}$. Then, $u_s$ is an eigenvector of $S$ corresponding to the eigenvalue $2(N-2)$.
\end{proposition}
By Proposition~\ref{SignsSync:SignsProp}, the eigenvector $u_s$ of $S$ in~\eqref{SignsSync:SignsSyncMat} gives the set $\{\widetilde{s}_{ij}\}_{i<j\in[N]}$.
The procedure for the signs adjustment (and the construction of $\widetilde{H}$ of~\eqref{SignsSync:HTilde}), is summarized in Algorithm~\ref{alg:signs}.
\begin{algorithm}
\caption{Signs adjustment procedure}
\begin{algorithmic}[1]	
\Require{A set of $\binom{N}{2}$ rank 1 matrices $\{v_{ij}=s_{ij}v_i^Tv_j\}_{i<j\in[N]}$}
\Initialize{${\binom{N}{2}}\times{\binom{N}{2}}$ matrix $S$, with all entries set to zero, and $N+1$ matrices \\\qquad $\{H_1^s,\ldots,H_N^s\}$ and $\widetilde{H}$ of size $3N\times3N$, with all entries set to zero}
\State Estimate $s_{ii}v_{i}^Tv_i$, $i=1,\ldots,N$  \Comment{See \eqref{SignsSync:vii}}
\For{$n=1$ to $N$}
\For{$i<j\in[N]$}
\If{$\|v_{in}v_{nj}-v_{ij}\|_F>\|v_{in}v_{nj}+v_{ij}\|_F$}
\State{$(H_n^s)_{ij}=-v_{ij}$}
\Else\Comment{$(H^s_n)_{ij}$ denotes the $(i,j)^{th}$ $3\times3$ block of $H_n^s$}
\State{$(H_n^s)_{ij}=v_{ij}$}
\EndIf
\EndFor
\State{$H_n^s=H_n^s+(H_n^s)^T$}
\For{$i=1\text{ to } N$}
\State{$(H_n^s)_{ii}=s_{ii}v_{i}^Tv_i$}
\EndFor
\State{$\hat{v}_n^s=\argmax\limits_{\|v\|=1}v^TH_n^sv$} \Comment{See~\eqref{SignsSync:SnVsn}}
\EndFor
\For{$n=1 \text{ to } N$}
\For{$i=1 \text{ to } N$}
\State{$(\hat{v}_{n}^s)_i=\hat{v}_n^s(3i-2,3i-1,3i)$} \Comment{See~\eqref{SignsSync:rowNotation}}
\EndFor
\EndFor
\For{$(i,j)(k,l) \in A$} \Comment{See~\eqref{eq:def_of_A}}
\State{$S_{(i,j)(k,l)}= \widetilde{s}_{ij}\widetilde{s}_{kl}$} \Comment{Using $\hat{v}_{n}^s$, see~\eqref{SignsSync:SijSjkTild} and~\eqref{SignsSync:SignsSyncMat}} \EndFor
\State{$u_s=\argmax\limits_{\|v\|=1}v^TSv$}
\For{$i<j\in[N]$}
\State{$\widetilde{s}_{ij}=u_s(v_{ij})$} \Comment{See Proposition~\ref{SignsSync:SignsProp}}
\State{$(\widetilde{H})_{ij}=\widetilde{s}_{ij}\cdot v_{ij}$} \Comment{See~\eqref{SignsSync:s_ijTildBys_ij} and~\eqref{SignsSync:HTildeBlock}}
\EndFor
\State{$\widetilde{H}=\widetilde{H}+\widetilde{H}^T$}
\For{$i=1\text{ to } N$}
\State{$(\widetilde{H})_{ii}=s_{ii}v_{i}^Tv_i$}
\EndFor
\State{$v^s=\argmax\limits_{\|v\|=1}v^T\widetilde{H}v$}\Comment{See~\eqref{SignsSync:HTilde}}
\Ensure{$v^s$}
\end{algorithmic}
\label{alg:signs}
\end{algorithm}

\section{Numerical experiments}\label{sec:experiments}
We implemented Algorithms~\ref{alg:relative rotations estimation}--\ref{alg:signs} in Matlab and tested them on a dataset of raw projection images of the beta-galactosidase enzyme~\cite{betaGala}, which has a $D_2$ symmetry. All tests were executed on a dual Intel Xeon E5-2683 CPU (32 cores in total), with
768GB of RAM running Linux, and four nVidia GTX TITAN XP GPU's. Section~\ref{sec:implementation} provides some of the implementation details for Algorithms~\ref{alg:relative rotations estimation}--\ref{alg:signs}, and Section~\ref{sec:betagal} presents the results on the experimental dataset.

\subsection{Implementation details}\label{sec:implementation}
To execute Algorithm~\ref{alg:relative rotations estimation}, we need to discretize the space of rotations $SO(3)$. To that end, we generated a pseudo-uniform spherical grid of $K$=1200 points $z_{k}$ on $S^2$, using the Saaf-Kuijlaars algorithm~\cite{saafKuij}. Then, for each $z_k=(a_k,b_k,c_k)^T\in S^2$ on the spherical grid, we computed the set of rotations
\begin{equation}
Q_{kl}=\begin{pmatrix} |&|&|\\ \cos(\theta_l)u_k+\sin(\theta_l)w_k&-\sin(\theta_l)u_k+\cos(\theta_l)w_k & z_k\\|&|&|\end{pmatrix},
\end{equation}
where $\theta_l=2\pi l/L$ for $l=0,1,\ldots,L-1$, and the vectors $u_k$ and $w_k$ are given by
\begin{equation}
u_k=\frac{(-b_k,a_k,0)^T}{\|(-b_k,a_k,0)^T\|},\quad w_k=\frac{u_k\times z_k}{\|u_k\times u_k\|}.
\end{equation}
It is easily verified that the vectors $w_k$, $u_k$ and $z_k$ form an orthonormal set, and that $Q_{kl}\in SO(3)$. The third column $z_k$ of each rotation $Q_{kl}$ is the beaming direction corresponding to the rotation $Q_{kl}$, and the vectors
\begin{equation}
\cos(2\pi l/L)u_k+\sin(2\pi l/L)w_k, \quad -\sin(2\pi l/L)u_k+\cos(2\pi l/L)w_k,
\end{equation}
are the coordinate systems for the plane perpendicular to $z_k$. Thus, each set of matrices $\{Q_{kl}\}_{l=0}^{L-1}$ where $k\in\{1,\ldots,K\}$ is a discretization of the set of rotations in $SO(3)$ with beaming direction $z_k$. We found experimentally that choosing $L=72$ (together with $K=1200$) is sufficient to obtain accurate results.

As for runtime, for a set of 500 projection images, it took 1512 seconds to compute all sets of relative rotations $\{R_i^Tg_mR_j\}_{m=1}^4$ (Algorithm~\ref{alg:relative rotations estimation}), 720 second to synchronize handedness (Algorithm~\ref{alg:handedness synchronization}), 5784 seconds to compute the partition in~\eqref{Clr.ClassPar}  (Algorithms~\ref{alg:constructing Omega} and~\ref{alg:ualpha}), and 1378 seconds to adjust the signs $s_{ij}$ of~\eqref{Clr.ClassPar} (Algorithm~\ref{alg:signs}).

\subsection{Beta-galactosidase experimental results}\label{sec:betagal}
We applied Algorithms~\ref{alg:relative rotations estimation}--\ref{alg:signs} to the EMPIAR-$10061$ dataset~\cite{EMPIAR-10061} from the EMPIAR archive~\cite{empiar}. The dataset consists of 41,123 raw particles images, each of size $768 \times 768$ pixels, with pixel size of $0.3185$ \AA.
To generate class averages from this dataset, we used the ASPIRE software package~\cite{aspire} as follows. First, all images were phase-flipped (in order to remove the phase-reversals in the CTF), down-sampled to size $129 \times 129$ pixels (hence with pixel size of $1.9$~\AA), and normalized so that the noise in each image has zero mean and unit variance. We then split the images into two independent sets, each consisting of 20,560 particle images, and all subsequent processing was applied to each set independently.

We next used the class-averaging procedure in ASPIRE~\cite{aspire} to generate 2000 class averages from each of the two sets of particle images (using the EM-based class averaging algorithm in ASPIRE). A sample of these class averages is shown in Fig.~\ref{fig:10061_classavg}. The input to our algorithm was 500 out of the 2000 class averages (by selecting every 4th image).

\begin{figure}
	\centering
	\includegraphics[width=0.15\textwidth]{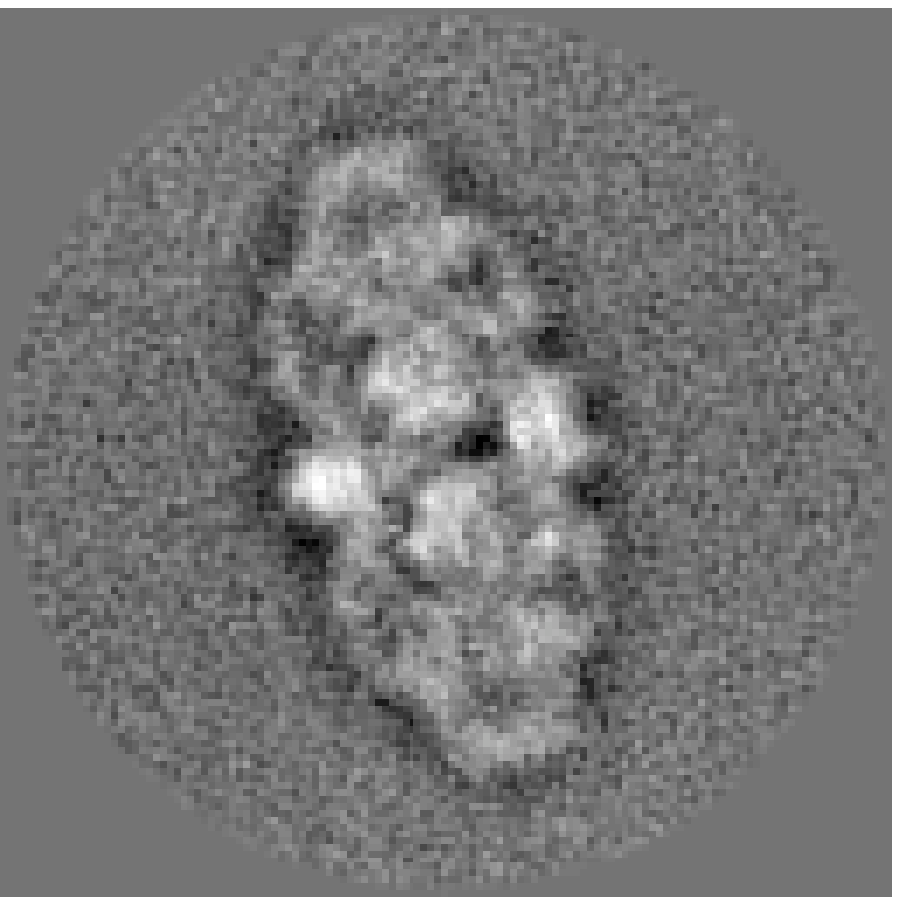}
    \includegraphics[width=0.15\textwidth]{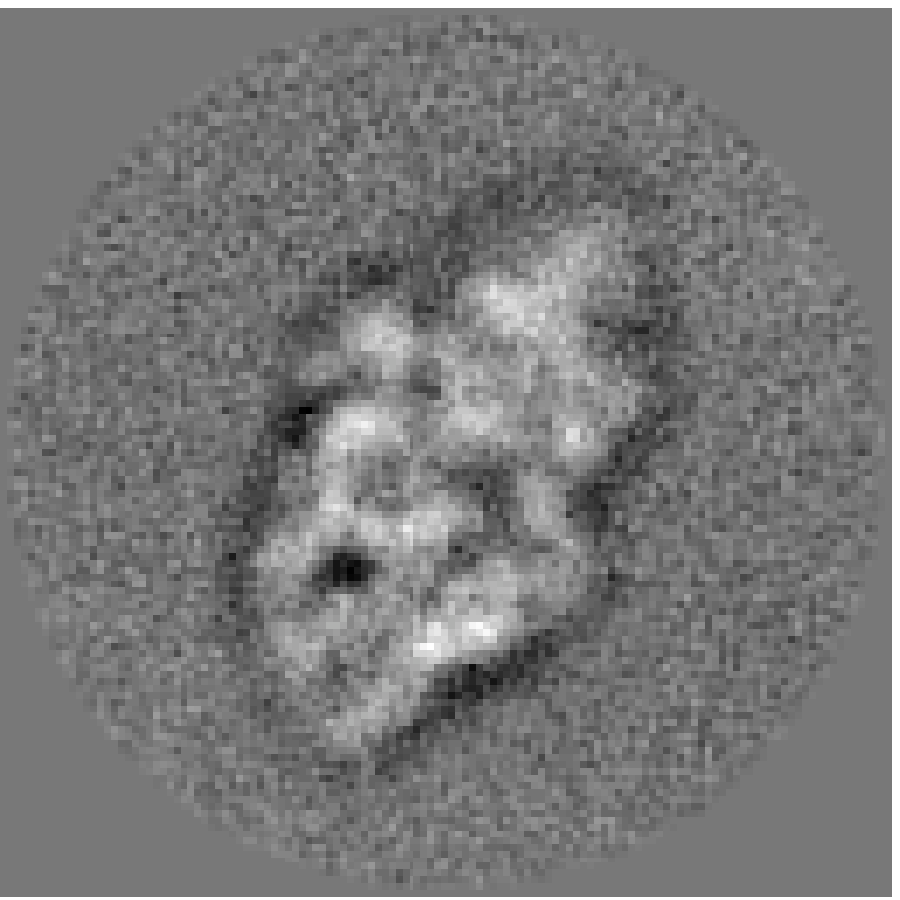}
    \includegraphics[width=0.15\textwidth]{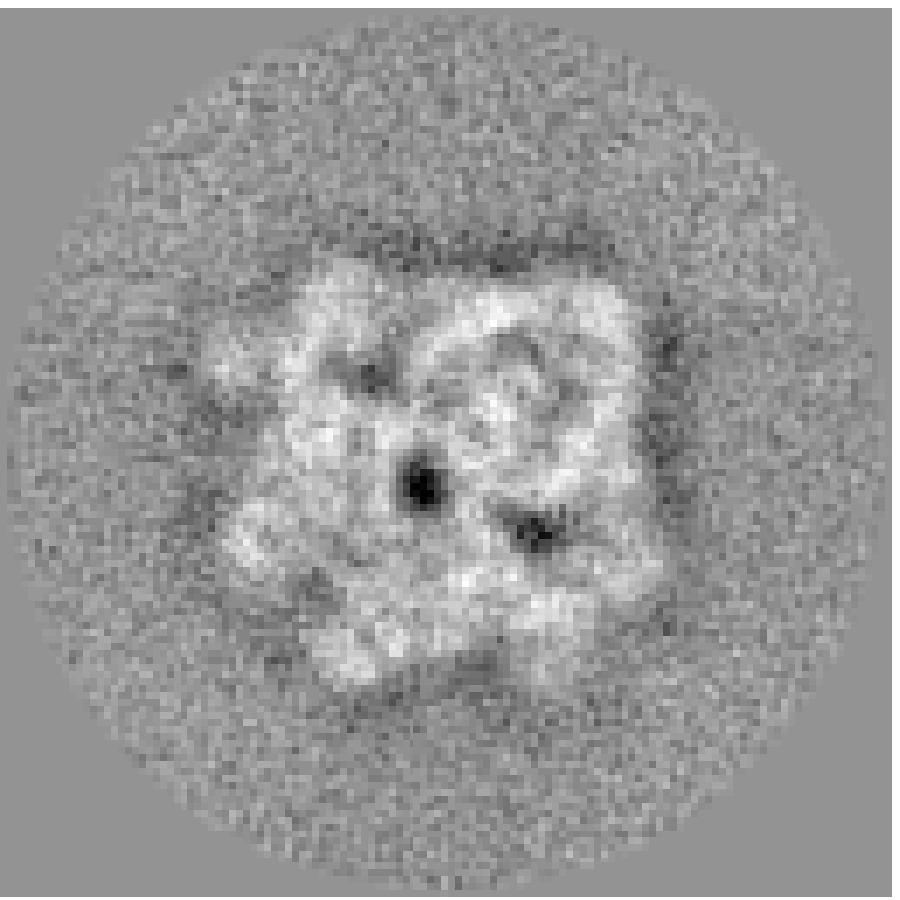}
    \includegraphics[width=0.15\textwidth]{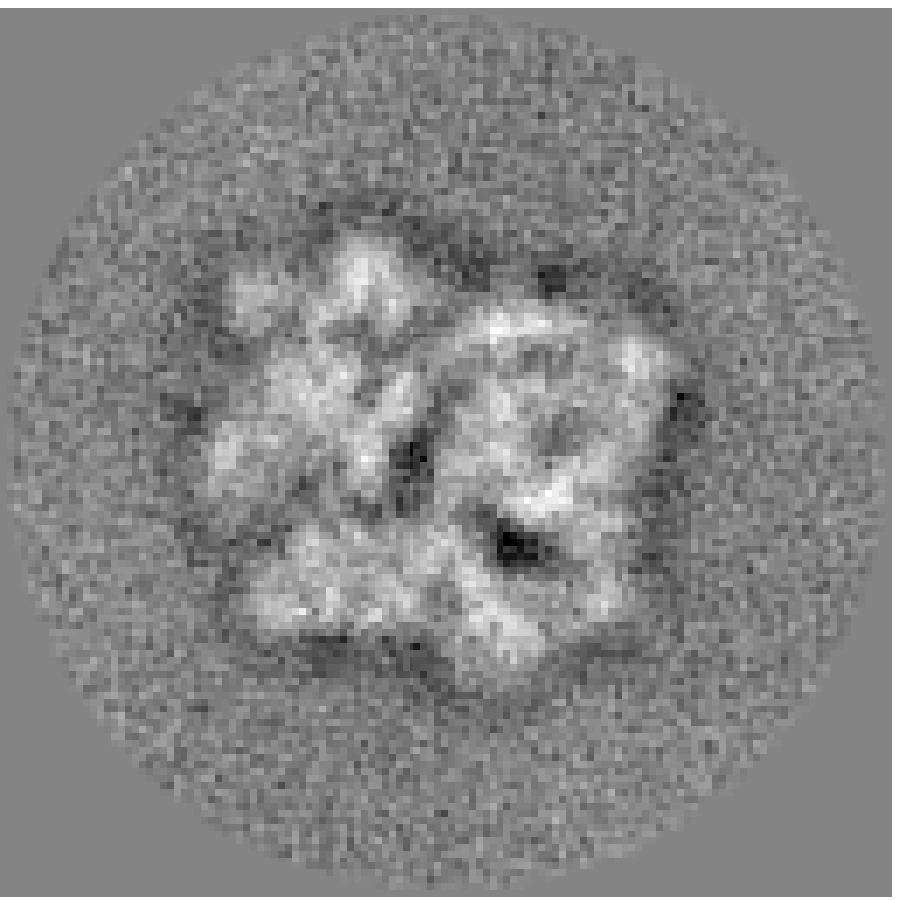}
    \includegraphics[width=0.15\textwidth]{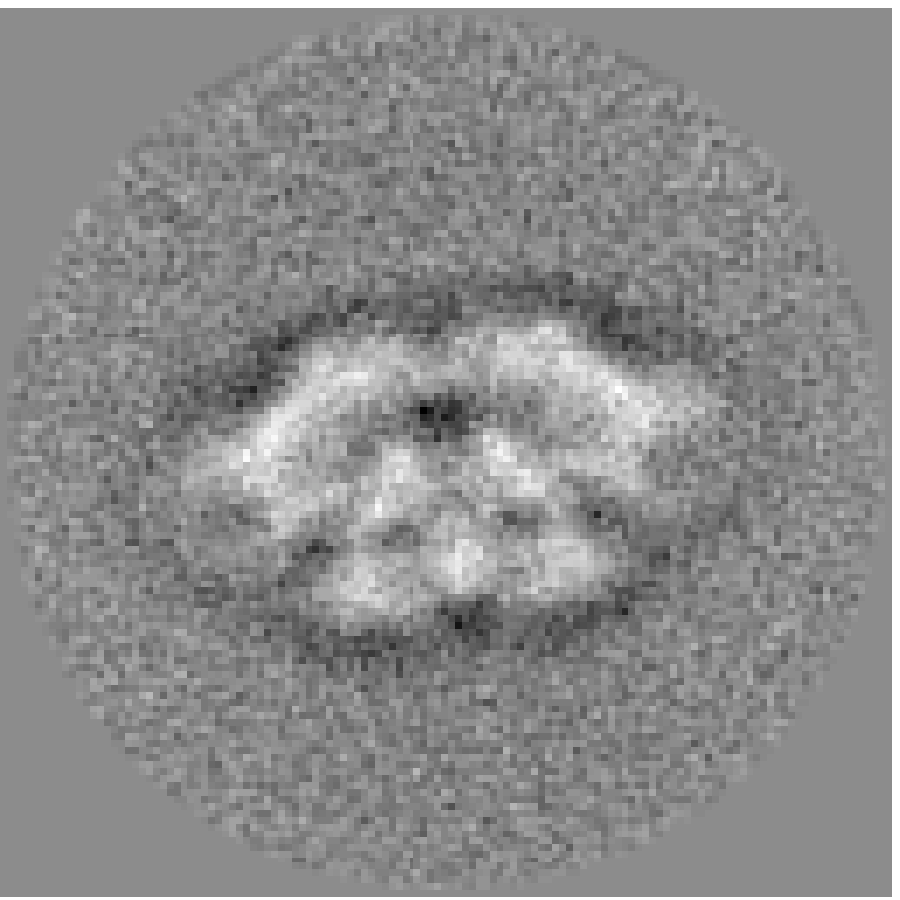}
    \includegraphics[width=0.15\textwidth]{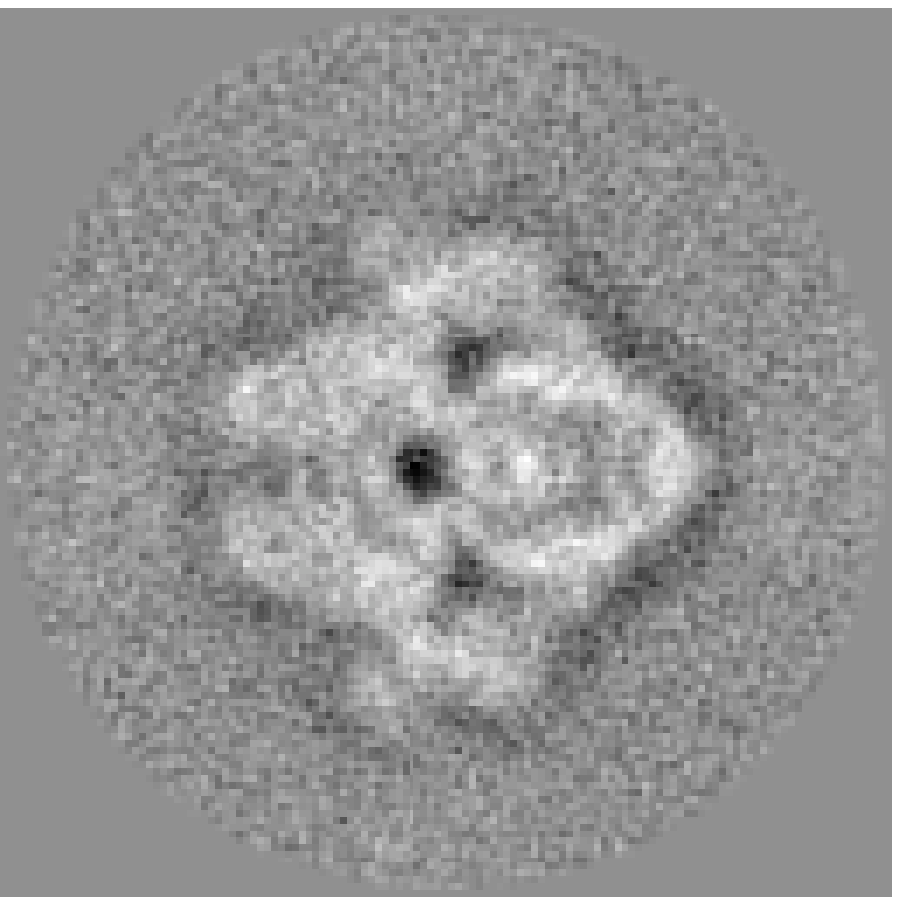}
    \caption{A sample of $129 \times 129$ class averages of the EMPIAR-$10061$ dataset~\cite{EMPIAR-10061}.}
	\label{fig:10061_classavg}
\end{figure}

Next, we used the algorithms described in this paper to estimate the rotation matrices that correspond to the 500 class averages, and reconstructed the three-dimensional density map using the class averages and their estimated rotation matrices. The resolution of the reconstructed volume, assessed by comparing the reconstructions from the two independent sets of class averages is 8.23~\AA,~using the Fourier shell correlation (FSC) $0.143$-criterion~\cite{vanHeel_Schatz} (Fig.~\ref{fig:fsc_halfmaps}). When comparing our reconstructions to a high resolution reconstruction of the molecule (EMD-$7770$~\cite{EMDB-7770}), the resolution estimated using the 0.5-criterion of the FSC is 9.88~\AA~(Fig.~\ref{fig:fsc_ref}).

\begin{figure}
	\centering
	\begin{subfigure}{0.45\textwidth}
        \centering
        \includegraphics[width=\textwidth]{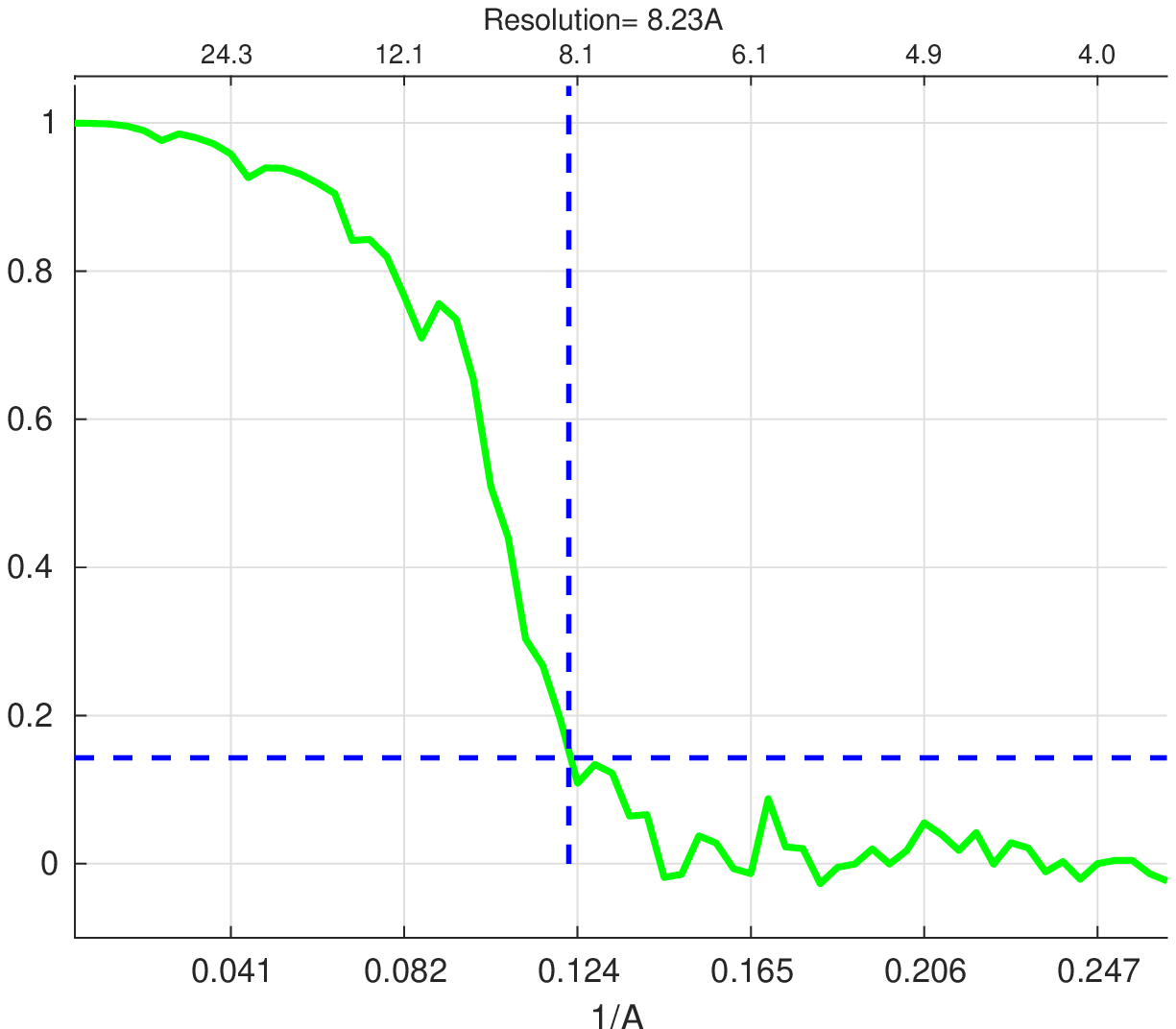}
        \caption{Two half maps}
        \label{fig:fsc_halfmaps}
    \end{subfigure}
	\begin{subfigure}{0.45\textwidth}
        \centering
        \includegraphics[width=\textwidth]{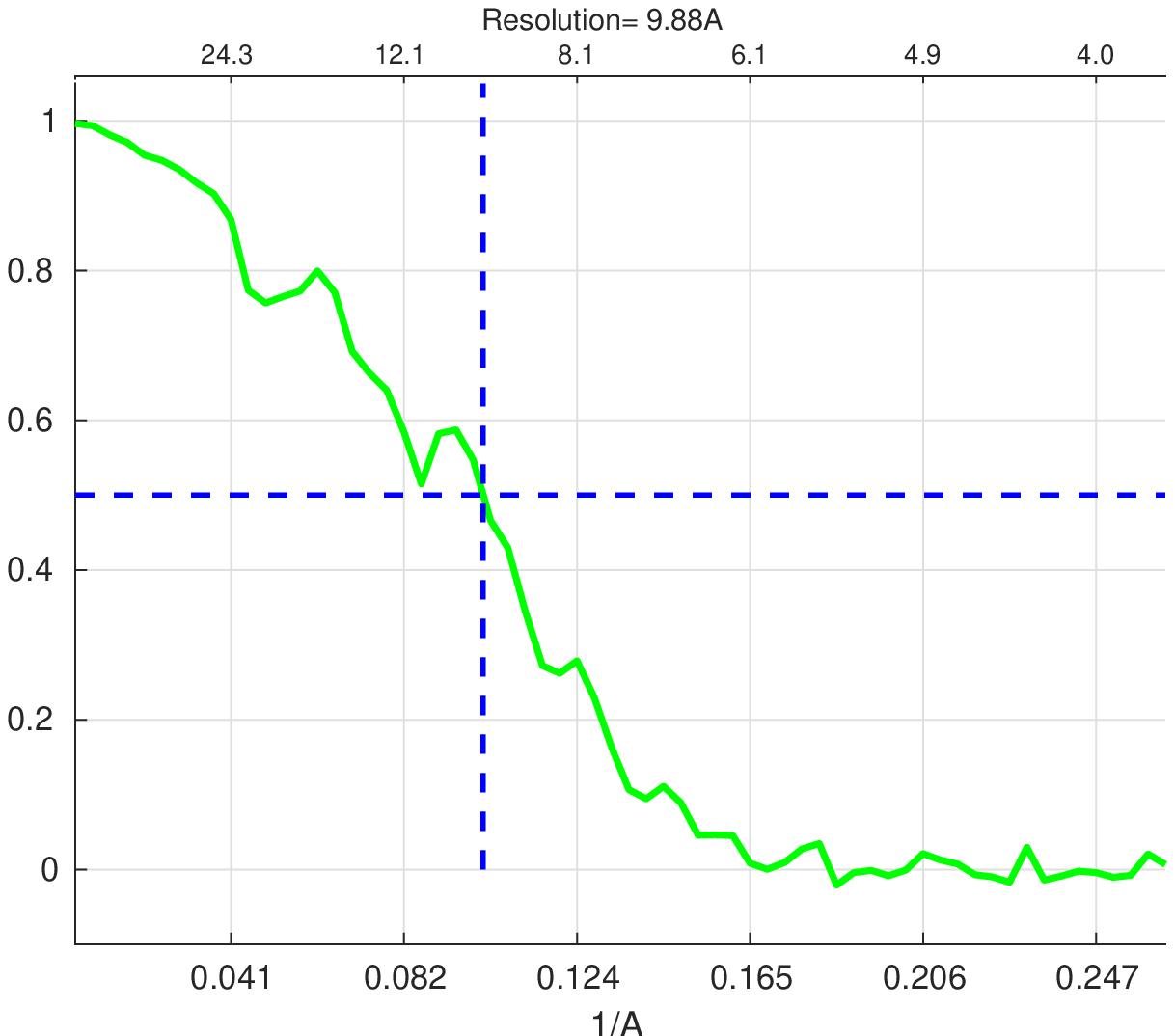}
        \caption{Half map vs reference}
        \label{fig:fsc_ref}
    \end{subfigure}
    \caption{Fourier shell correlation curves.}
\end{figure}

\section{Summary and future work}\label{sec:summary}
In this paper, we presented a procedure for estimating the orientations corresponding to a given set of projection images of a $D_2$-symmetric molecule. We have shown that the set of relative rotations between all pairs of images admits a special graph structure, and demonstrated that this structure can be exploited to recover the rotations. We then demonstrated our method by reconstructing an ab-initio model from an experimental set of cryo-EM images.

An obvious future work is to extend the proposed method to $D_n$ for $n\geq 3$. Preliminary theoretical analysis suggests that this can be achieved by combining the method of the current paper with the algorithms derived in~\cite{cn}.

\section*{Acknowledgments}
This research was supported by the European Research Council
(ERC) under the European Union’s Horizon 2020 research and innovation programme (grant
agreement 723991 - CRYOMATH).

\appendix
\section{Appendix}
\subsection{Proof of Proposition~\ref{Sec5:MainProp}}\label{proof:Sec5:MainProp}
	Let $(R_i^Tg_{\tau(m)}R_j)_{m=1}^4$ be a permutation of the 4-tuple $(R_i^Tg_mR_j)_{m=1}^4$.
	\begin{enumerate}
		\item This follows immediately from~\eqref{Sec5:Sumgs} and~\eqref{Sec5:ViVj}, by noting that the tuple $(\tau(2),\tau(3),\tau(4))$ is a permutation of $(2,3,4)$ .
		\item Suppose that $\tau(2)=1$. Then, $(\tau(1),\tau(3),\tau(4))$ is a permutation of $(2,3,4)$.
		By~\eqref{Sec5:SumgsUnord} and~\eqref{Sec5:ViVj}, we have
		\begin{equation*}
		\begin{split}
		\frac{1}{2}(R_i^Tg_{\tau(1)}R_j+R_i^Tg_{\tau(3)}R_j)&=R_i^T\frac{1}{2}(g_{\tau(1)}+g_{\tau(3)})R_j\\
		&=-R_i^TI_{\tau(4)-1}R_j=-(v_i^{\tau(4)-1})^Tv_j^{\tau(4)-1},\\
		\frac{1}{2}(R_i^Tg_{\tau(1)}R_j+R_i^Tg_{\tau(4)}R_j)&=R_i^T\frac{1}{2}(g_{\tau(1)}+g_{\tau(4)})R_j\\
		&=-R_i^TI_{\tau(3)-1}R_j=-(v_i^{\tau(3)-1})^Tv_j^{\tau(3)-1}.
		\end{split}
		\end{equation*}
		By~\eqref{Sec5:Sumgs} and~\eqref{Sec5:ViVj}, we have
		\begin{equation*}
		\begin{split}
		\frac{1}{2}(R_i^Tg_{\tau(1)}R_j+R_i^Tg_{\tau(2)}R_j)&=R_i^T\frac{1}{2}(g_{\tau(1)}+g_{\tau(2)})R_j\\
		&=R_i^T\frac{1}{2}(g_{\tau(1)}+g_1)R_j\\
		&=R_i^TI_{\tau(1)-1}R_j=(v_i^{\tau(1)-1})^Tv_j^{\tau(1)-1},
		\end{split}
		\end{equation*}
		as required. The proof for the cases $m=3,4$ is similar.
	\end{enumerate}
\qed

\subsection{Proof of Proposition~\ref{Jsync:MainProp}}\label{proof:Jsync:MainProp}
	First, we note that the set $\{g_1,g_2,g_3,g_4\}$ of~\eqref{Sec1:gpMem} forms a multiplicative group of matrices, known in literature as the Klein four-group. In particular, $g_1=I$ is the identity element of the group, and we have
	\begin{equation}\label{JSync:Klein4}
	g_{m}^2=g_1,\quad m=1,2,3,4.
	\end{equation}
	
	Now, since $R_i,R_j$ and $R_k$ are in $SO(3)$, we have
	\begin{equation}\label{JSync:Rs2GpMem}
	\begin{split}
	&R_{ij}^mR_{jk}^lR_{ki}^r=(R_i^Tg_{\tau_{ij}(m)}R_j)(R_j^Tg_{\tau_{jk}(l)}R_k)(R_k^Tg_{\tau_{ki}(r)}R_i)=I \iff\\
	&R_i^Tg_{\tau_{ij}(m)}g_{\tau_{jk}(l)}g_{\tau_{ki}(r)}R_i=I\iff
	g_{\tau_{ij}(m)}g_{\tau_{jk}(l)}g_{\tau_{ki}(r)}=I.
	\end{split}
	\end{equation}
	By~\eqref{JSync:Klein4}, each member of the group $\{g_m\}_{m=1}^4$ is its own inverse, and so for any triplet $(m,l,r)\in\{1,2,3,4\}^3$ we have that $g_{\tau_{ij}(m)}g_{\tau_{jk}(l)}g_{\tau_{ki}(r)}=I$ if and only if
	\begin{equation}\label{JSync:KleinEquation}
	g_{\tau_{ij}(m)}g_{\tau_{jk}(l)}=g_{\tau_{ki}(r)}.
	\end{equation}
	Now observe, that since $\tau_{ij},\tau_{jk},\tau_{ki}\in S_4$ are permutations, we have
	\begin{equation}\label{JSync:KleinPermEquivalence}
	\{g_{\tau_{ij}(m)}\}_{m=1}^4=\{g_{\tau_{jk}(m)}\}_{m=1}^4=\{g_{\tau_{ki}(m)}\}_{m=1}^4=\{g_m\}_{m=1}^4,
	\end{equation}
	from which it follows that there are 16 possible products $g_{\tau_{ij}(m)}g_{\tau_{jk}(l)}$ on the left-hand side of~\eqref{JSync:KleinEquation}, corresponding to the 16 products
	$R_{ij}^mR_{jk}^l$ for $m,l\in\{1,2,3,4\}$.
	Thus, there are at most 16 triplets $(m,l,r)\in\{1,2,3,4\}^3$ for which~\eqref{JSync:KleinEquation} is satisfied. On the other hand, combining the closure property of groups with~\eqref{JSync:KleinPermEquivalence} gives us that there always exists an element $g_{\tau_{ki}(r)}$ of $\{g_{\tau_{ki}(m)}\}_{m=1}^4$ such that~\eqref{JSync:KleinEquation} is satisfied. We conclude that are exactly 16 triplets $(m,l,r)\in\{1,2,3,4\}^3$ such that~\eqref{JSync:KleinEquation} is satisfied, from which by~\eqref{JSync:Rs2GpMem}, the proof is concluded.

\qed

\subsection{Proof of Theorem \ref{Clr.MainThrm}}\label{proof:Clr.MainThrm}
We assume without loss of generality, that the first ${\binom{N}{2}}$ rows of $\Omega$ correspond to the vertices $\{v_{ij}^1\}_{i<j\in[N]}$, the following  ${\binom{N}{2}}$ rows correspond to the vertices $\{v_{ij}^2\}_{i<j\in[N]}$, and the last  ${\binom{N}{2}}$ rows correspond to $\{v_{ij}^3\}_{i<j\in[N]}$.
Let $\Sigma^{+}$ be the ${\binom{N}{2}}\times {\binom{N}{2}}$ matrix given by
\begin{equation}\label{Clr.MainThrm.SigPlus}
(\Sigma^+)_{(i,j)(k,l)}=\begin{cases}
1 & |\{i,j\}\cap\{k,l\}|=1,\\
0 & \text{otherwise}.
\end{cases}
\end{equation}
In~\cite{Jsync}, it was shown that the spectrum of $\Sigma^+$ is given by
\begin{equation}\label{Clr.MainThrm.SigPlusSpec}
\begin{pmatrix}2(N-2)& N-4& -2\\ 1 & N-1& {\binom{N}{2}}-N\end{pmatrix}.
\end{equation}
We now show how to relate the spectrum of $\Omega$ to that of $\Sigma^+$.

Let $B_{rp}$ for $r,p=1,2,3$, be an ${\binom{N}{2}}\times{\binom{N}{2}}$ matrix which consists of rows $((r-1){\binom{N}{2}}+1),\ldots,r{\binom{N}{2}}$ and columns
$((p-1){\binom{N}{2}}+1),\ldots,p{\binom{N}{2}}$ of the matrix $\Omega$. That is, we partition $\Omega$ into 9 block matrices of equal dimensions.
By assumption, for any $r\in\{1,2,3\}$, both the rows and columns of the matrix $B_{rr}$ correspond to vertices of the same set $C_r$ in~\eqref{Clr.Color.Class}. Thus, by~\eqref{Clr.Omega.Def} we have that for $r\in\{1,2,3\}$,
\begin{equation}\label{Clr.MainThrm.Brr}
(B_{rr})_{(i,j)(k,l)}=\begin{cases}
1 & |\{i,j\}\cap\{k,l\}|=1,\\
0 & \text{otherwise},
\end{cases}
\end{equation}
which by~\eqref{Clr.MainThrm.SigPlus} gives that $B_{rr}=\Sigma^+$ for all $r\in\{1,2,3\}$.
Now, consider any matrix $B_{rp}$ for $r\neq p\in\{1,2,3\}$, and note that its rows correspond to the vertices of the set $C_r=\{v_{ij}^r\}_{i<j\in[N]}$, whereas its columns correspond the vertices of the set $C_p=\{v_{ij}^p\}_{i<j\in[N]}$. Again, by~\eqref{Clr.Omega.Def}, we have that for $r\neq p\in\{1,2,3\}$,
\begin{equation}\label{Clr.MainThrm.Brp}
(B_{rp})_{(i,j)(k,l)}=\begin{cases}
-1 & |\{i,j\}\cap\{k,l\}|=1,\\
\;\;\:0 & \text{otherwise}.
\end{cases}
\end{equation}
Thus, by~\eqref{Clr.MainThrm.SigPlus}, we have that $B_{rp}=-\Sigma^+$ for all $r\neq p\in\{1,2,3\}$, and we conclude that
\begin{equation}\label{Clr.MainThrm.Om2Sig}
\Omega=\begin{pmatrix*}[r]\Sigma^+& -\Sigma^+& -\Sigma^+\\ -\Sigma^+& \Sigma^+& -\Sigma^+ \\ -\Sigma^+& -\Sigma^+& \Sigma^+\end{pmatrix*}.
\end{equation}
Now, let $u\in\mathbb{R}^{\binom{N}{2}}$ be any eigenvector of $\Sigma^+$ corresponding to an eigenvalue $\lambda$.
Denote $z=(0,\ldots,0)^T\in\mathbb{R}^{\binom{N}{2}}$, and consider the column vectors of length $3{\binom{N}{2}}$
\begin{equation}\label{Clr.MainProof.uiDef}
u_1=\begin{pmatrix*}
u\\u\\u
\end{pmatrix*},\quad
u_2=\begin{pmatrix*}
u\\z\\-u
\end{pmatrix*},\quad
u_3=\begin{pmatrix*}
u\\-2u\\u
\end{pmatrix*}.
\end{equation}
By~\eqref{Clr.MainThrm.Om2Sig}, we have
\begin{equation}\label{Clr.MainThrm.uDef}
\begin{split}
\Omega u_1&=
\begin{pmatrix*}[r]
\Sigma^+u-\Sigma^+u-\Sigma^+u\\-\Sigma^+u+\Sigma^+u-\Sigma^+u\\-\Sigma^+u-\Sigma^+u+\Sigma^+u
\end{pmatrix*}=
\begin{pmatrix}
-\lambda u\\-\lambda u\\-\lambda u
\end{pmatrix}=-\lambda\begin{pmatrix}
u\\u\\u
\end{pmatrix}=-\lambda u_1,\\
\Omega u_2&=\begin{pmatrix*}[r]
\Sigma^+u-\Sigma^+z+\Sigma^+u\\-\Sigma^+u+\Sigma^+z+\Sigma^+u\\-\Sigma^+u-\Sigma^+z-\Sigma^+u
\end{pmatrix*}=\begin{pmatrix}
2\lambda u\\2\lambda z\\-2\lambda u
\end{pmatrix}=2\lambda\begin{pmatrix}
u\\z\\-u
\end{pmatrix}=2\lambda u_2,\\
\Omega u_3&=\begin{pmatrix*}[r]
\Sigma^+u+2\Sigma^+u-\Sigma^+u\\-\Sigma^+u-2\Sigma^+ u-\Sigma^+u\\-\Sigma^+u+2\Sigma^+u+\Sigma^+u
\end{pmatrix*}=\begin{pmatrix}
2\lambda u\\-4\lambda u\\2\lambda u
\end{pmatrix}=2\lambda\begin{pmatrix}
u\\-2u\\u
\end{pmatrix}=2\lambda u_3.
\end{split}
\end{equation}
This shows that if $\lambda$ is an eigenvalue of $\Sigma^+$, then $-\lambda$ and $2\lambda$ are eigenvalues of $\Omega$. Furthermore, if $u$ is an eigenvector of $\Sigma^+$ with eigenvalue $\lambda$, then $u_1$ in~\eqref{Clr.MainProof.uiDef} is an eigenvector of $\Omega$ corresponding to the eigenvalue $-\lambda$ of $\Omega$, and $u_2$ and $u_3$ in~\eqref{Clr.MainProof.uiDef} are eigenvectors of $\Omega$ corresponding to the eigenvalue $2\lambda$ of $\Omega$. Now, note that
\begin{equation*}
\langle u_2,u_3 \rangle=\langle u,u \rangle-2\langle z,u \rangle-\langle u,u \rangle=0,
\end{equation*}
and thus $u_2$ and $u_3$ in~\eqref{Clr.MainProof.uiDef} are independent for any $u\in\mathbb{R}^{\binom{N}{2}}$. Furthermore, since all eigenvalues of $\Sigma^+$ are non-zero (see~\eqref{Clr.MainThrm.SigPlusSpec}), we have that $2\lambda\neq-\lambda\neq 0$, and so $u_1$ is in a different eigenspace of $\Omega$ than $u_2$ and $u_3$. Thus, all three vectors in~\eqref{Clr.MainProof.uiDef} are independent eigenvectors of $\Omega$.
Let us denote the multiplicity of an eigenvalue $\lambda$ of the matrix $\Sigma^+$ by $m_{\Sigma^+}(\lambda)$, and the multiplicity of an eigenvalue $\mu$ of $\Omega$, by $m_\Omega(\mu)$. We also denote the three eigenvalues of $\Sigma^+$ by $\lambda_1,\lambda_2$ and $\lambda_3$. It is simple to verify, that if a pair of vectors $u$ and $v$ are independent, then so are the pairs of vectors $\{(u,u,u)^T,(v,v,v)^T\}$, $\{(u,z,-u)^T,(v,z,-v)^T\}$, and $\{(u,-2u,u)^T,(v,-2v,v)^T\}$. Thus, by~\eqref{Clr.MainThrm.uDef},
if $\lambda$ is an eigenvalue of $\Sigma^+$, then the eigenvalues $2\lambda$ and $-\lambda$ of $\Omega$ satisfy
\begin{equation}\label{Clr.MainProof.MultBounds}
m_\Omega(-\lambda)\geq m_{\Sigma^+}(\lambda),\quad m_\Omega(2\lambda)\geq 2m_{\Sigma^+}(\lambda),
\end{equation}
which gives us that
\begin{equation}\label{Clr.MainProof.SumMultLowBound}
\sum_{i=1}^3 \left ( m_\Omega(-\lambda_i)+m_{\Omega}(2\lambda_i) \right ) \geq 3\sum_{i=1}^3m_{\Sigma^+}(\lambda_i)=3{\binom{N}{2}}.
\end{equation}
On the other hand, since $\Omega$ has dimensions $3{\binom{N}{2}}\times3{\binom{N}{2}}$, we have
\begin{equation}\label{Clr.MainProof.SumMultUpBound}
\sum_{i=1}^3m_\Omega(-\lambda_i)+m_{\Omega}(2\lambda_i)\leq3{\binom{N}{2}},
\end{equation}
by which we have that
\begin{equation}\label{Clr:eigMultEquality}
m_\Omega(2\lambda_i)=2m_{\Sigma^+}(\lambda_i),\quad m_\Omega(-\lambda_i)= m_{\Sigma^+}(\lambda_i),\quad i=1,2,3,
\end{equation}
for otherwise, by~\eqref{Clr.MainProof.MultBounds} we would have a strong inequality in~\eqref{Clr.MainProof.SumMultLowBound}, which is a contradiction to~\eqref{Clr.MainProof.SumMultUpBound}.

We conclude that the set of eigenvalues of $\Omega$ is given by
$\{2\lambda_i,-\lambda_i\}_{i=1}^3$. Finally, the multiplicities of the eigenvalues of $\Omega$ in~\eqref{Clr:OmegaSpec} are computed by combining~\eqref{Clr.MainThrm.SigPlusSpec} with~\eqref{Clr:eigMultEquality}.

\qed

\subsection{Proof of Proposition~\ref{Clr.MainProp}}\label{proof:Clr.MainProp}
We begin by introducing some notation and definitions which we use in the current and subsequent proofs.
Let $P_{\sigma}$ denote the $3\times3$ permutation matrix of $\sigma\in S_3$, i.e., $P_{\sigma}$ satisfies $P_{\sigma}v=(v(\sigma(1)),v(\sigma(2)),v(\sigma(3))^T$ for any $v\in \mathbb{R}^3$.
Thus, using the notation introduced in~\eqref{Clr:w_ij} we can write
\begin{equation}\label{Clr:u_ij2p_ij}
(u_\alpha)_{ij}=P_{\sigma_{ij}}u_{3c},\quad (u_\beta)_{ij}=P_{\sigma_{ij}}u_{2c},
\end{equation}
where $u_\alpha$ and $u_\beta$ are given in Definition~\ref{Clr.uAlphDef}, and $u_{3c}$ and $u_{2c}$ are defined in~\eqref{Clr.Unmix.CVecs}.

Now, using the notation of~\eqref{Clr.Omega.BlockDef}, consider the block $\Omega_{(i,j)(j,k)}$ of $\Omega$ for some $i<j<k\in[N]$. If it were the case that $\sigma_{ij}$ and $\sigma_{jk}$ are the identity permutations, then by~\eqref{Clr.Omega.BlockDef} and~\eqref{Clr.Omega.Def}, we would have that
\begin{equation}\label{Clr.MainProp.IdBlock}
\Omega_{(i,j)(j,k)}=
\begin{pmatrix*}[r]
1&-1&-1\\-1&1&-1\\-1&-1&1
\end{pmatrix*}.
\end{equation}
However, in general, $\Omega_{(i,j)(j,k)}$ is the $3 \times 3$ block of $\Omega$ which we get by taking the entries of $\Omega$ in the rows corresponding to the triplet of vertices $(v_{ij}^{\sigma_{ij}(1)},v_{ij}^{\sigma_{ij}(2)},v_{ij}^{\sigma_{ij}(3)})$, and columns corresponding to the triplet of vertices $(v_{jk}^{\sigma_{jk}(1)},v_{jk}^{\sigma_{jk}(2)},v_{jk}^{\sigma_{kl}(3)})$. In other words, $\Omega_{(i,j)(j,k)}$ is obtained from the matrix in~\eqref{Clr.MainProp.IdBlock} by permuting its rows by $\sigma_{ij}$ and its columns by $\sigma_{jk}$, that is
\begin{equation*}
\Omega_{(i,j)(j,k)}=P_{\sigma_{ij}} \begin{pmatrix*}[r]
1&-1&-1\\-1&1&-1\\-1&-1&1
\end{pmatrix*}P_{\sigma_{jk}}^T.
\end{equation*}
By the same argument applied to $\Omega_{(i,j)(k,l)}$ whenever $|\{i,j\}\cap\{k,l\}|=1$, we have that
\begin{equation}\label{Clr:OmegaBlock}
\Omega_{(i,j)(k,l)}=P_{\sigma_{ij}} \begin{pmatrix*}[r]
1&-1&-1\\-1&1&-1\\-1&-1&1
\end{pmatrix*}P_{\sigma_{kl}}^T.
\end{equation}
The reason that $P_{\sigma_{kl}}$ in~\eqref{Clr:OmegaBlock} is transposed, is that in order to permute the columns of a matrix by $\sigma_{kl}$, one has to multiply it on the right by $P_{\sigma_{kl}}^T$.
We now prove Proposition~\ref{Clr.MainProp}
\begin{proof}[Proof of Proposition~\ref{Clr.MainProp}]
	Let us compute the Rayleigh quotient for $u_\alpha$. By the second equality in~\eqref{Clr.Omega.IdxIdentity}, we have
	\begin{equation}
	\begin{split}	
	u_{\alpha}^T\Omega u_{\alpha}=\sum_{i<j<k\in[N]}
	[&(u_\alpha)_{ij}^T\Omega_{(i,j)(j,k)}(u_\alpha)_{jk}+
	(u_\alpha)_{jk}^T\Omega_{(j,k)(i,j)}(u_\alpha)_{ij}\\
    &\qquad+(u_\alpha)_{ij}^T\Omega_{(i,j)(i,k)}(u_\alpha)_{ik}+ (u_\alpha)_{ik}^T\Omega_{(i,k)(i,j)}(u_\alpha)_{ij}\\
    &\qquad+(u_\alpha)_{jk}^T\Omega_{(j,k)(i,k)}(u_\alpha)_{ik}+
	(u_\alpha)_{ik}^T\Omega_{(i,k)(j,k)}(u_\alpha)_{jk}].
	\end{split}
	\end{equation}	
	By~\eqref{Clr:u_ij2p_ij} and~\eqref{Clr:OmegaBlock},
	and since permutation matrices are orthogonal, we have
	\begin{equation}\label{Clr.MainProof.MainEq}
	\begin{split}
	u_{\alpha}^T\Omega u_{\alpha}
	=\sum_{i<j<k\in[N]}
	&(P_{\sigma_{ij}}u_{3c})^TP_{\sigma_{ij}} \begin{pmatrix*}[r]
	1&-1&-1\\-1&1&-1\\-1&-1&1
	\end{pmatrix*}P_{\sigma_{jk}}^TP_{\sigma_{jk}}u_{3c}\\
	&+(P_{\sigma_{jk}}u_{3c})^TP_{\sigma_{jk}} \begin{pmatrix*}[r]
	1&-1&-1\\-1&1&-1\\-1&-1&1
	\end{pmatrix*}P_{\sigma_{ij}}^TP_{\sigma_{ij}}u_{3c}\\		
	&+(P_{\sigma_{ij}}u_{3c})^TP_{\sigma_{ij}} \begin{pmatrix*}[r]
	1&-1&-1\\-1&1&-1\\-1&-1&1
	\end{pmatrix*}P_{\sigma_{ik}}^TP_{\sigma_{ik}}u_{3c}\\
	&+(P_{\sigma_{ik}}u_{3c})^TP_{\sigma_{ik}} \begin{pmatrix*}[r]
	1&-1&-1\\-1&1&-1\\-1&-1&1
	\end{pmatrix*}P_{\sigma_{ij}}^TP_{\sigma_{ij}}u_{3c}\\
	&+(P_{\sigma_{jk}}u_{3c})^TP_{\sigma_{jk}} \begin{pmatrix*}[r]
	1&-1&-1\\-1&1&-1\\-1&-1&1
	\end{pmatrix*}P_{\sigma_{ik}}^TP_{\sigma_{ik}}u_{3c}\\	
	&+(P_{\sigma_{ik}}u_{3c})^TP_{\sigma_{ik}} \begin{pmatrix*}[r]
	1&-1&-1\\-1&1&-1\\-1&-1&1
	\end{pmatrix*}P_{\sigma_{jk}}^TP_{\sigma_{jk}}u_{3c}\\			
	=\sum_{i<j<k\in[N]}&6\cdot(u_{3c})^T \begin{pmatrix*}[r]
	1&-1&-1\\-1&1&-1\\-1&-1&1
	\end{pmatrix*}u_{3c}
	\end{split}	
	\end{equation}
	It is straightforward to check that each term in the sum~\eqref{Clr.MainProof.MainEq} equals exactly $24\alpha^2$. By Definition~\ref{Clr.uAlphDef}, $\alpha=(2{\binom{N}{2}})^{-\frac{1}{2}}$,
	and since there are $\binom{N}{3}$ triplets $i<j<k\in [N]$, the sum in~\eqref{Clr.MainProof.MainEq} amounts to
	\begin{equation*}
	u_{\alpha}^T\Omega u_{\alpha}={\binom{N}{3}}\cdot 24 \alpha^2=\frac{24{\binom{N}{3}}}{2{\binom{N}{2}}}=4(N-2).
	\end{equation*}
	By Theorem~\ref{Clr.MainThrm}, $\mu_c=4(N-2)$ is the leading eigenvalue of $\Omega$, and thus, $u_\alpha$ maximizes the Rayleigh quotient of $\Omega$, by which we have that $u_{\alpha}$ is in the eigenspace of $\mu_c$.
	A similar calculation for $u_{\beta}$ shows that it is also in the same eigenspace. Finally, observe that since $P_{\sigma_{ij}}$ are orthogonal for all $i<j\in[N]$, we have
	\begin{equation*}
	\begin{split}
	<u_{\alpha},u_{\beta}>&=\sum_{i<j\in[N]}
	<(u_\alpha)_{ij},(u_\beta)_{ij}>=
	\sum_{i<j\in[N]}
	<P_{\sigma_{ij}}u_{3c},P_{\sigma_{ij}}u_{2c}>\\
	&=\sum_{i<j\in[N]}<u_{3c},u_{2c}>
	=\sum_{i<j\in[N]}\alpha\cdot\beta+0\cdot(-2\beta) +\alpha\cdot(-\beta)=0.
	\end{split}
	\end{equation*}		
\end{proof}

\subsection{Proof of Lemma~\ref{Clr.Omega.IdxLemma}}\label{proof:Clr.Omega.IdxLemma}
	We begin by showing the first equality in~\eqref{Clr.Omega.IdxIdentity}.
	Fix some $i<j\in[N]$, and observe that for any $k<l\in[N]$ such that $|\{i,j\}\cap\{k,l\}|=1$, we have that
	\begin{equation}\label{SignsSync:IndexCases1}
	(i,j)(k,l)=
	\begin{cases}
	(i,j)(k,j) & j=l \text{ and }(k<i<j \text{ or } i<k<j),\\
	(i,j)(j,l) & j=k \text{ and }i<j<l,\\
	(i,j)(k,i) & i=l \text{ and }k<i<j,\\
	(i,j)(i,l) & i=k \text{ and }i<l<j \text{ or } i<j<l,
	\end{cases}
	\end{equation}
    where in the first case of~\eqref{SignsSync:IndexCases1} it cannot be that $i=k$ since then we would have that $\left \lvert \{i,j\}\cap\{k,l\} \right \rvert>1$ (and similarly for the other cases). From~\eqref{SignsSync:IndexCases1} and~\eqref{eq:def_of_A1_to_A4} we get that
	\begin{equation*}
		(i,j)(k,l)\in
		\begin{cases}
		A_{ij}^1 & j=l\text{ and }(k<i<j \text{ or } i<k<j),\\
		A_{ij}^2 & j=k\text{ and }i<j<l,\\
		A_{ij}^3 & i=l\text{ and }k<i<j,\\
		A_{ij}^4 & i=k\text{ and }i<l<j \text{ or } i<j<l.
		\end{cases}
	\end{equation*}
	This shows that for fixed $i <j$ it holds that
\begin{equation*}
\left \{ (i,j)(k,l) \; | \; k<l,\ |\{i,j\}\cap\{k,l\}|=1 \right \} \subseteq A_{ij}^1\cup A_{ij}^2\cup A_{ij}^3\cup A_{ij}^4,
\end{equation*}
and so by taking a union over all $i<j\in[N]$, we get from~\eqref{eq:def_of_A} that
	\begin{equation}\label{Clr.IdxLemmaPf.eq1}
    \begin{aligned}
    A &= \bigcup_{i<j\in[N]} \left \{ (i,j)(k,l) \; | \; k<l,\ |\{i,j\}\cap\{k,l\}|=1 \right \} \\
      &\subseteq \bigcup_{i<j\in[N]} A_{ij}^1\cup A_{ij}^2\cup A_{ij}^3\cup A_{ij}^4.
    \end{aligned}
	\end{equation}
	Conversely, suppose that $(i,j)(k,j)\in A_{ij}^1$ for some $i<j\in[N]$ and $k<j$, $k\neq i$. Then, we have that $(i,j)(k,j)=(i,j)(k,l)$ for $l=j$,  $|\{i,j\}\cap\{k,l\}|=1$ and $k<l$. Thus, by \eqref{eq:def_of_A}, for any $k\in[N]$ such that $k<j$ and $k\neq i$ we have that $(i,j)(k,j)\in A$, from which we have that $A_{ij}^1\subseteq A$. Applying a similar argument to $A_{ij}^2, A_{ij}^3$ and $A_{ij}^4$, we get that $A_{ij}^2, A_{ij}^3,A_{ij}^4\subseteq A$ for all $i<j\in [N]$, from which we get that
	\begin{equation*}
	A_{ij}^1\cup A_{ij}^2\cup A_{ij}^3\cup A_{ij}^4\subseteq A, \quad i<j\in[N].
	\end{equation*}
	Taking a union over all $i<j\in[N]$ we have
	\begin{equation}\label{eq:A contains union4}
	A\supseteq\bigcup_{i<j\in[N]}A_{ij}^1\cup A_{ij}^2\cup A_{ij}^3\cup A_{ij}^4,
	\end{equation}
	which together with \eqref{Clr.IdxLemmaPf.eq1} proves the first equality in \eqref{Clr.Omega.IdxIdentity}.
	
	Let us now show the second equality in \eqref{Clr.Omega.IdxIdentity}. Suppose that $i<j<k\in [N]$. Then, from~\eqref{eq:def_of_A1_to_A4} we have that
	\begin{equation}\label{eq:break A and Af}
	\begin{alignedat}{6}
	(i,j)(j,k)&\in A_{ij}^2, &\quad& (i,j)(i,k)&\in A_{ij}^4, &\quad& (j,k)(i,k)&\in A_{jk}^1, \\
	(j,k)(i,j)&\in A_{jk}^3, &\quad& (i,k)(i,j)&\in A_{ik}^4, &\quad& (i,k)(j,k)&\in A_{ik}^1.
	\end{alignedat}	
	\end{equation}
	By~\eqref{eq:break A and Af} and the definition of $A_{ijk}$ and $A_{ijk}^{f}$ in~\eqref{eq:A and Af}, we have
	\begin{equation*}
	A_{ijk} \subseteq  A_{ij}^2\cup A_{ij}^4 \cup A_{jk}^1, \quad A_{ijk}^f \subseteq A_{jk}^3\cup A_{ik}^4 \cup A_{ik}^1.
	\end{equation*}
	Thus, taking the union over all $i<j<k\in [N]$ gives us that
    \begin{equation}\label{Clr.Omega.6Union}
    \bigcup_{i<j<k\in[N]}A_{ijk}\cup A_{ijk}^f \subseteq \bigcup_{i<j<k\in[N]} A_{ij}^2\cup A_{ij}^4 \cup A_{jk}^1 \cup A_{jk}^3\cup A_{ik}^4 \cup A_{ik}^1.
    \end{equation}
    Now, for each $i<j<k\in[N]$ we have
    \begin{alignat*}{7}
    &A_{ij}^2 &\cup& A_{ij}^4 &\subseteq& A_{ij}^1 &\cup& A_{ij}^2 &\cup& A_{ij}^3 &\cup& A_{ij}^4,\\
    &A_{jk}^1 &\cup& A_{jk}^3 &\subseteq& A_{jk}^1 &\cup& A_{jk}^2 &\cup& A_{jk}^3 &\cup& A_{jk}^4,\\
    &A_{ik}^4 &\cup& A_{ik}^1 &\subseteq& A_{ik}^1 &\cup& A_{ik}^2 &\cup& A_{ik}^3 &\cup& A_{ik}^4,
    \end{alignat*}
    by which we have that (renaming the indices to those given in~\eqref{eq:def_of_A1_to_A4})
    \begin{equation*}
     A_{ij}^2\cup A_{ij}^4 \cup A_{jk}^1 \cup A_{jk}^3\cup A_{ik}^4 \cup A_{ik}^1\subseteq \bigcup_{i<j\in[N]}A_{ij}^1\cup A_{ij}^2 \cup A_{ij}^3 \cup A_{ij}^4.
    \end{equation*}
    Taking a union over all $i<j<k\in [N]$ we get
	\begin{equation*}
	\bigcup_{i<j<k\in[N]} A_{ij}^2\cup A_{ij}^4 \cup A_{jk}^1 \cup A_{jk}^3\cup A_{ik}^4 \cup A_{ik}^1\subseteq \bigcup_{i<j\in[N]}A_{ij}^1\cup A_{ij}^2 \cup A_{ij}^3 \cup A_{ij}^4,
	\end{equation*}
	thus, by \eqref{Clr.Omega.6Union}, we have
	\begin{equation*}
	\bigcup_{i<j<k\in[N]}A_{ijk}\cup A_{ijk}^f \subseteq \bigcup_{i<j\in[N]}A_{ij}^1\cup A_{ij}^2 \cup A_{ij}^3 \cup A_{ij}^4 \subseteq A,
	\end{equation*}
where the last inequality follows from~\eqref{eq:A contains union4}.

	Conversely, take $(i,j)(k,l) \in A$. By the first part of the proof, $(i,j)(k,l)$ belongs to one of the sets in~\eqref{eq:def_of_A1_to_A4}. If $(i,j)(k,j)\in A_{ij}^1$, then we have that
	either $k<i<j$ or $i<k<j$, and thus
	\begin{equation*}
	(i,j)(k,j)\in\begin{cases}
	\{(k,i)(i,j)\;,\;(k,i)(k,j)\;,\;(i,j)(k,j)\} & k<i<j,\\
	\{(k,j)(i,k)\;,\;(i,j)(i,k)\;,\; (i,j)(k,j)\} & i<k<j,
	\end{cases}
	\end{equation*}
	that is, $(i,j)(k,j)\in A_{kij}$ or $(i,j)(k,j)\in A_{ikj}^f$.
	This shows that (renaming the indices to the order given in~\eqref{eq:A and Af})
	\begin{equation*}
	(i,j)(k,j)\in\bigcup_{i<j<k\in[N]}A_{ijk}\cup A_{ijk}^f,
	\end{equation*}
	for either $k<i<j\in[N]$ or $i<k<j\in[N]$, by which we conclude that
	\begin{equation*}
	A_{ij}^1 \subseteq \bigcup_{i<j<k\in[N]}A_{ijk}\cup A_{ijk}^f.
	\end{equation*}
	In the same manner one can show that
	\begin{equation*}
	A_{ij}^2,A_{ij}^3,A_{ij}^4\subseteq\bigcup_{i<j<k\in[N]}A_{ijk}\cup A_{ijk}^f,
	\end{equation*}
	for all $i<j\in[N]$. Thus, we have that
	\begin{equation*}
	A_{ij}^1\cup A_{ij}^2\cup A_{ij}^3\cup A_{ij}^4\subseteq\bigcup_{i<j<k\in[N]}A_{ijk}\cup A_{ijk}^f
	\end{equation*}
	for all $i<j\in[N]$. Taking a union over all $i<j\in[N]$ and using the first equality in~\eqref{Clr.Omega.IdxIdentity} gives
	\begin{equation*}
	A = \bigcup_{i<j\in[N]}A_{ij}^1 \cup A_{ij}^2 \cup A_{ij}^3 \cup A_{ij}^4\subseteq\bigcup_{i<j<k\in[N]}A_{ijk}\cup A_{ijk}^f,
	\end{equation*}
	which concludes the proof of the second equality in \eqref{Clr.Omega.IdxIdentity}.

\qed

\subsection{Proof of Proposition~\ref{Col.Unmix.MinUniq}}\label{proof:Col.Unmix.MinUniq}
Define the sets
\begin{equation}
\Gamma_{3c}=\{P_\sigma u_{3c}\;|\; P_\sigma\in S_3\},\quad \Gamma_{2c}=\{P_\sigma u_{2c}\;|\; P_\sigma\in S_3\},
\end{equation}
where $u_{2c}$ and $u_{3c}$ were defined in~\eqref{Clr.Unmix.CVecs}, and $P_\sigma$ is the permutation matrix of $\sigma\in S_3$. That is, $\Gamma_{3c}$ is the set of all permutations of $u_{3c}$ and $\Gamma_{2c}$ is the set of all permutations of $u_{2c}$.
\begin{lemma}\label{Clr.Unmix.MainPropLemma}
	Let $u,v\in\mathbb{R}^3$ be such that $u=P_{\sigma}u_{3c}$ and  $v=P_{\sigma}u_{2c}$ for some $P_{\sigma}\in S_3$. Suppose, that $u',v'\in \mathbb{R}^3$ are such that $(u'\;v')=(u\;v)R(\theta)$
	for some $\theta\in [0,2\pi)$, $u'\in \Gamma_{3c}$ and $v'\in\Gamma_{2c}$, where $R(\theta)$ was defined in~\eqref{Clr.Unmix.RotMat}.
	Then, there exists $P_{\sigma}'\in S_3$ such that
	\begin{equation}\label{Clr.PropLemma}
	\begin{pmatrix}	|&|\\ u&v\\|&|\end{pmatrix}R(\theta)=P_{\sigma}'\begin{pmatrix}	|&|\\ \pm u& v\\|&|\end{pmatrix}.
	\end{equation}	
\end{lemma}
\begin{proof}
	First, since $P_\sigma$ is orthogonal, then
	\begin{equation}\label{Clr:PropLemma1}
	<u,v>=<P_\sigma u_{3c},P_\sigma u_{2c}>=<u_{3c},u_{2c}>=0.
	\end{equation}
	Define $W=\{(x,y,z)^T|\;x+y+z=0\}$, and note that $W$ is a linear subspace of $\mathbb{R}^3$ that contains $\Gamma_{3c}$ and $\Gamma_{2c}$. Since $(u'\;v')=(u\;v)R(\theta)$ and $R(\theta)$ is an orthogonal matrix, by~\eqref{Clr:PropLemma1} we have that
	$u'$ and $v'$ are also orthogonal vectors, and by assumption $u'$ and $v'$ are in $\Gamma_{3c}$ and $\Gamma_{2c}$, respectively, and so they are in $W$.
	Now, since $v,v'\in \Gamma_{2c}$ there exists $P_{\sigma}'\in S_3$ such that $v'=P_{\sigma}'v$.
	Since $P_{\sigma}'$ is an orthogonal matrix, we have
	\[<\pm P_{\sigma}'u,v'>=<\pm u,v>=0.\]
	Finally, since $W$ is of dimension 2, there are exactly two vectors perpendicular to $v'$ in $W$. Thus, it must be that $u'=\pm P_{\sigma}'u$, from which we get~\eqref{Clr.PropLemma}.
\end{proof}
 We now prove Proposition~\ref{Col.Unmix.MinUniq}.
\begin{proof}[Proof of Proposition~\ref{Col.Unmix.MinUniq}]
The function $f_{c}(\theta)$, given in~\eqref{Clr.Unmix.Max}, satisfies $f_c(\theta)\geq 0$ for all $\theta \in [0,2\pi)$. Suppose that $\theta$ is such that $f_c(\theta)=0$. Such a $\theta$ necessarily exists, since if we choose $\theta$ in~\eqref{eq:vtheta} such that $(v_a^{\theta},v_b^{\theta})=(\pm u_{\alpha}, u_{\beta})$ (which can be done due to~\eqref{Clr.Unmix.Rot}), then~\eqref{Clr.Unmix.Max} equals zero.
In the notation of \eqref{Clr:w_ij}, we now show that there exists $\sigma\in S_3$ such that either
\begin{equation}\label{Clr.Unmix.Prop2.Claim_ij1}
(v_a^\theta)_{ij}=(u_\alpha^\sigma)_{ij},\quad (v_b^\theta)_{ij}=(u_\beta^\sigma)_{ij},\quad i<j\in[N],
\end{equation}
or that
\begin{equation}\label{Clr.Unmix.Prop2.Claim_ij2}
(v_a^\theta)_{ij}=(-u_\alpha^\sigma)_{ij},\quad (v_b^\theta)_{ij}=(u_\beta^\sigma)_{ij},\quad i<j\in[N],
\end{equation}
from which it follows that $(v_a^\theta,v_b^\theta)\in\{\pm(u_{\alpha}^\sigma,u_{\beta}^\sigma)\;|\;\sigma\in S_3\}$.

First, we show that
\begin{equation}\label{eq:v_theta_InGamma}
(v_a^\theta)_{ij}\in\Gamma_{3c},\quad (v_b^\theta)_{ij}\in \Gamma_{2c}, \quad i<j\in[N].
\end{equation}
Indeed, for each pair $i<j\in[N]$, looking at the first square brackets in~\eqref{Clr.Unmix.Max}, we have that
	\begin{equation*}
	(M_{ij}(v_a^\theta)+m_{ij}(v_a^\theta))^2+d_{ij}(v_a^\theta)^2=0 \iff
	\left\{
	\begin{array}{ll}	
	M_{ij}(v_a^\theta)=-m_{ij}(v_a^\theta),\\
	d_{ij}(v_a^\theta)=0. \end{array}
	\right.
	\end{equation*}	
	This shows that $(v_a^\theta)_{ij}$ must be a permutation of $u_{3c}$ in~\eqref{Clr.Unmix.CVecs}, that is, $(v_a^\theta)_{ij}\in \Gamma_{3c}$. Similarly, looking at the second square brackets in~\eqref{Clr.Unmix.Max}, we have that
	\begin{equation*}	
	[(m_{ij}(v_b^{\theta})+2M_{ij}(v_b^{\theta}))^2 +(m_{ij}(v_b^{\theta})+2d_{ij}(v_b^{\theta}))^2
	+(M_{ij}(v_b^{\theta})-d_{ij}(v_b^{\theta}))^2]=0
	\end{equation*}
	\begin{equation*}
	\iff\left\{
	\begin{array}{ll}
	m_{ij}(v_b^{\theta})=-2M_{ij}(v_b^{\theta}),\\
	m_{ij}(v_b^{\theta})=-2d_{ij}(v_b^{\theta}),\\
	d_{ij}(v_b^{\theta})=M_{ij}(v_b^{\theta}),	
	\end{array}
	\right.
	\end{equation*}
	which is possible only if $(v_b^\theta)_{ij}$ is a permutation of the vector $u_{2c}$ in~\eqref{Clr.Unmix.CVecs}, i.e., $(v_b^\theta)_{ij}\in\Gamma_{2c}$, which shows \eqref{eq:v_theta_InGamma}.
	
	Now, by \eqref{Clr.Unmix.Rot} and \eqref{eq:vtheta}, the vectors $v_a^\theta$ and $v_b^\theta$ are either given by
	\begin{equation}\label{eq:Prop8Case1}
	(v_a^\theta\;\;v_b^\theta)=(v_a\; \;v_b)R(\theta)=(u_\alpha\; u_\beta)R(\varphi)R(\theta)=(u_\alpha\; u_\beta)R(\varphi+\theta),
	\end{equation}
	or by
	\begin{equation}\label{eq:Prop8Case2}
	(v_a^\theta\;\;v_b^\theta)=(v_a\; \;v_b)R(\theta)=(-u_\alpha\; u_\beta)R(\varphi)R(\theta)=(-u_\alpha\; u_\beta)R(\varphi+\theta),
	\end{equation}
	for some $\varphi\in[0,2\pi)$, where $v_a$ and $v_b$ is the pair of orthogonal eigenvectors of $\Omega$ defined in \eqref{Clr.Unmix.VaVb}. 	
	Let us assume first the case \eqref{eq:Prop8Case1}.
	In the notation of~\eqref{Clr:w_ij}, we define
	\begin{equation*}
	A_{\alpha\beta}^+=
	\begin{pmatrix} |&|\\(u_{\alpha})_{12} &(u_{\beta})_{12} \\ |&| \end{pmatrix},\quad A_{\alpha\beta}^-=
	\begin{pmatrix} |&|\\(-u_{\alpha})_{12} &(u_{\beta})_{12} \\|&|\end{pmatrix}.
	\end{equation*}	
	Also, denoting by
	\begin{equation}\label{eq:A^+Def}
	A_c^+=(u_{3c}\;\; u_{2c}),\quad A_c^-=(-u_{3c}\;\;u_{2c}),
	\end{equation}
	the $2\times3$ matrices with columns $\pm u_{3c}$ and $u_{2c}$ (defined in \eqref{Clr.Unmix.CVecs}),
	we get that by Definition~\ref{Clr.uAlphDef} and \eqref{Clr:w_ij}, for each pair $i<j\in[N]$ we have
	\begin{equation}\label{Clr.Unmix.AlphBet1}
	\begin{pmatrix} |&|\\(u_{\alpha})_{ij} &(u_{\beta})_{ij} \\ |&| \end{pmatrix}=P_{\sigma_{ij}}A_{c}^+.
	\end{equation}
	In particular, for $i=1$ and $j=2$ we have
	\begin{equation}\label{Clr.Unmix.AlphBet5}
	A_{\alpha\beta}^+=\begin{pmatrix} |&|\\(u_{\alpha})_{12} &(u_{\beta})_{12} \\ |&| \end{pmatrix}=P_{\sigma_{12}}A_{c}^+.
	\end{equation}
	Thus, by~\eqref{Clr.Unmix.AlphBet1} and~\eqref{Clr.Unmix.AlphBet5} it follows that for all $i<j\in[N]$ we have
	\begin{equation}\label{Clr:Prop771}
	\begin{pmatrix} |&|\\(u_{\alpha})_{ij} &(u_{\beta})_{ij} \\ |&| \end{pmatrix}=P_{\sigma_{ij}}P_{\sigma_{12}}^TA_{\alpha\beta}^+.
	\end{equation}	
	Now, by \eqref{eq:Prop8Case1} and \eqref{Clr:w_ij} we have that
	\begin{equation}\label{eq:rotate Aplus}
	\begin{pmatrix}
	|&|\\ (v_{a}^\theta)_{12} & (v_{b}^\theta)_{12}\\ |&|
	\end{pmatrix}=
	\begin{pmatrix} |&|\\(u_{\alpha})_{12} &(u_{\beta})_{12} \\ |&|\end{pmatrix}R(\theta+\varphi)=A^+_{\alpha\beta}R(\theta+\varphi).
	\end{equation}
	By~\eqref{Clr.Unmix.AlphBet5} we have that $(u_\alpha)_{12}=P_{\sigma_{12}}u_{3c}$ and $(u_\beta)_{12}=P_{\sigma_{12}}u_{2c}$, and by \eqref{eq:v_theta_InGamma}, we have that $(v_a^\theta)_{12}\in \Gamma_{3c}$ and $(v_b^\theta)_{12}\in \Gamma_{2c}$.	
	Thus, using~\eqref{eq:rotate Aplus}, by Lemma \ref{Clr.Unmix.MainPropLemma} there exists a permutation matrix $P_\tau$ such that either
	\begin{equation}\label{eq:useOfLemmaA1}
	A_{\alpha\beta}^+R(\theta+\varphi)=P_{\tau}A_{\alpha\beta}^+\quad \text{or}\quad A_{\alpha\beta}^+R(\theta+\varphi)=P_{\tau}A_{\alpha\beta}^-.
	\end{equation}	

	 First, assume that the case on the left of \eqref{eq:useOfLemmaA1} holds. It then follows from~\eqref{Clr:w_ij}, \eqref{eq:Prop8Case1}, \eqref{Clr:Prop771} and~\eqref{Clr.Unmix.AlphBet5} that
	\begin{equation*}
	\begin{split}
	\begin{pmatrix}
	|&|\\ (v_{a}^\theta)_{ij} & (v_{b}^\theta)_{ij}\\ |&|
	\end{pmatrix}&=
		\begin{pmatrix} |&|\\(u_{\alpha})_{ij} &(u_{\beta})_{ij} \\ |&|\end{pmatrix}R(\theta+\varphi)=P_{\sigma_{ij}}P_{\sigma_{12}}^TA_{\alpha\beta}^+R(\theta+\varphi)\\
		&=P_{\sigma_{ij}}P_{\sigma_{12}}^TP_{\tau}A_{\alpha\beta}^+=P_{\sigma_{ij}}(P_{\sigma_{12}}^TP_{\tau}P_{\sigma_{12}})A_c^+.
	\end{split}
	\end{equation*}
	Writing $P_{\sigma}=P_{\sigma_{12}}^TP_{\tau}P_{\sigma_{12}}$, we get by \eqref{eq:A^+Def} that for all $i<j\in[N]$
	\begin{equation}\label{eq:UnmixPropCubeDeGras}
	\begin{pmatrix}
	|&|\\ (v_{a}^\theta)_{ij} & (v_{b}^\theta)_{ij}\\ |&|
	\end{pmatrix}=
	P_{\sigma_{ij}}P_{\sigma}A_{c}^+\\
	=P_{\sigma_{ij}}\begin{pmatrix}
	u_{3c}(\sigma(1))&u_{2c}(\sigma(1))\\
	u_{3c}(\sigma(2))&u_{2c}(\sigma(2))\\
	u_{3c}(\sigma(3))&u_{2c}(\sigma(3))
	\end{pmatrix}.
	\end{equation}
	By \eqref{Clr.Unmix.uSig} and \eqref{Clr:w_ij}, we have that
	\begin{equation}\label{Clr.Unmix.Prop2.MainIdentity1}
	\begin{split}
	P_{\sigma_{ij}}
	\begin{pmatrix}
	u_{3c}(\sigma(1)) &u_{2c}(\sigma(1))\\u_{3c}(\sigma(2))& u_{2c}(\sigma(2))\\u_{3c}(\sigma(3))& u_{2c}(\sigma(3))\\
	\end{pmatrix}&=
	P_{\sigma_{ij}}\begin{pmatrix}
	u_\alpha^\sigma(v_{ij}^{1})&u_\beta^\sigma(v_{ij}^{1})\\u_\alpha^\sigma(v_{ij}^{2})&u_\beta^\sigma(v_{ij}^{2})\\u_\alpha^\sigma(v_{ij}^{3})&u_\beta^\sigma(v_{ij}^{3})
	\end{pmatrix}\\&=
	\begin{pmatrix}
	u_\alpha^\sigma(v_{ij}^{\sigma_{ij}(1)})&u_\beta^\sigma(v_{ij}^{\sigma_{ij}(1)})\\u_\alpha^\sigma(v_{ij}^{\sigma_{ij}(2)})&u_\beta^\sigma(v_{ij}^{\sigma_{ij}(2)})\\u_\alpha^\sigma(v_{ij}^{\sigma_{ij}(3)})&u_\beta^\sigma(v_{ij}^{\sigma_{ij}(3)})
	\end{pmatrix}=
	\begin{pmatrix}
	|&|\\(u_\alpha^{\sigma})_{ij}&(u_\beta^{\sigma})_{ij}\\|&|
	\end{pmatrix}
	\end{split}
	\end{equation}
	for all $i<j\in[N]$. The last two equations show that \eqref{Clr.Unmix.Prop2.Claim_ij1} holds, which proves the proposition for the case $\eqref{eq:Prop8Case1}$, when the identity on the left of \eqref{eq:useOfLemmaA1} holds.

	If~$\eqref{eq:Prop8Case1}$ holds as well as the case on the right of~\eqref{eq:useOfLemmaA1}, where $A_{\alpha\beta}^+R(\theta+\varphi)=P_{\tau}A_{\alpha\beta}^-$, then by repeating the latter calculation with $A_{\alpha\beta}^+$ and $A_c^+$ replaced by $A_{\alpha\beta}^-$ and $A_c^-$, we get that
	\begin{equation}\label{Clr.Unmix.Prop2.MainIdentity2}
	\begin{split}
	\begin{pmatrix}
	|&|\\ (v_{a}^\theta)_{ij} & (v_{b}^\theta)_{ij}\\ |&|
	\end{pmatrix}
	&=P_{\sigma_{ij}}P_{\sigma}A_{c}^-\\
	&=P_{\sigma_{ij}}\begin{pmatrix}
	-u_{3c}(\sigma(1))&u_{2c}(\sigma(1))\\
	-u_{3c}(\sigma(2))&u_{2c}(\sigma(2))\\
	-u_{3c}(\sigma(3))&u_{2c}(\sigma(3))
	\end{pmatrix}=\begin{pmatrix}
	|&|\\ (-u_\alpha^\sigma)_{ij} & (u_\beta^\sigma)_{ij}\\ |&|
	\end{pmatrix},
	\end{split}
	\end{equation}
	for all $i<j\in[N]$, i.e., that \eqref{Clr.Unmix.Prop2.Claim_ij2} holds, which proves the proposition for the case $\eqref{eq:Prop8Case1}$ when the identity on the right of \eqref{eq:useOfLemmaA1} holds. This concludes the proof for the case $\eqref{eq:Prop8Case1}$.

	In the case where~$\eqref{eq:Prop8Case2}$ holds, we get by the same method of proof
that either \eqref{Clr.Unmix.Prop2.Claim_ij2} or \eqref{Clr.Unmix.Prop2.Claim_ij1} hold, which proves the proposition for this case.

To conclude, we have shown that given $\theta$ which minimizes~\eqref{Clr.Unmix.Max}, the vectors $v_a^{\theta}$ and $v_b^{\theta}$ defined in~\eqref{eq:vtheta}, must satisfy either $(v_a^{\theta},v_b^{\theta}) = (u_{\alpha}^{\sigma}, u_{\beta}^{\sigma})$ or $(v_a^{\theta},v_b^{\theta}) = (-u_{\alpha}^{\sigma}, u_{\beta}^{\sigma})$ for some $\sigma \in S_{3}$.

\end{proof}

\subsection{Proof of Proposition~\ref{SignsSync:MainProp}}\label{proof:SignsSync:MainProp}
	Suppose that~\eqref{SignsSync:CycleId} holds, and fix an arbitrary $n\in[N]$. Then, the $(i,j)^{th}$ $3\times3$ block of $H$ is given by $s_{ij}v_i^Tv_j=s_{in}s_{nj}v_i^Tv_j$. Thus, since $s_{nj}=s_{jn}$ for all $j\in[N]$, we have
	\begin{equation}
	H=(v_n^s)^Tv_n^s,\quad v_n^s=(s_{n1}v_1,\ldots,s_{nN}v_N),
	\end{equation}
	which gives~\eqref{SignsSync:HnDecomp} and shows that $H$ indeed has rank 1.
	
	Now suppose that~\eqref{SignsSync:CycleId} does not hold, and assume without loss of generality that $s_{12}s_{23}=-s_{13}$. Denote the $m^{th}$ entry of a row $v_i$ by $v_i(m)$, $i\in[N]$.
	The rank 1 matrix $v_1^Tv_2$ is non-zero, thus, there exist $r,l\in\{1,2,3\}$ such that $v_1(r),v_2(l)\neq 0$, that is, the $r^{th}$ and $l^{th}$ entries of the vectors $v_1$ and $v_2$, respectively, are non-zero.
	By~\eqref{SignsSync:v_ii}, we have that $s_{11}=1$, and thus, the first three rows of $H$ are given by the $3\times3N$ matrix
	\begin{equation*}
	(v_1^Tv_1,s_{12}v_1^Tv_2,s_{13}v_1^Tv_3\ldots,s_{1N}v_1^Tv_N)=v_1^T(v_1,s_{12}v_2,s_{13}v_3\ldots,s_{1N}v_N).
	\end{equation*}
	Similarly, the next three rows of $H$ are given by the $3\times3N$ matrix
	\begin{equation*}
	(s_{21}v_2^Tv_1,v_2^Tv_2,s_{23}v_2^Tv_3,\ldots,s_{2N}v_2^Tv_N)=v_2^T(s_{21}v_1,v_2,s_{23}v_3\ldots,s_{2N}v_N).
	\end{equation*}
	Thus, since each vector $v_i$ is of length 3, rows number $r$ and $3+l$ of $H$ are given by
	\begin{gather}
	v_1(r)(v_1,s_{12}v_2,s_{13}v_3,\ldots,s_{1N}v_N), \label{eq:v1}\\
	v_2(l)(s_{21}v_1,v_2,s_{23}v_3,\ldots,s_{2N}v_N). \label{eq:v2}
	\end{gather}
	Multiplying~\eqref{eq:v1} by $\frac{1}{v_1(r)}$ and~\eqref{eq:v2} by $\frac{s_{12}}{v_2(l)}$, by our assumption we get
	\begin{equation*}
	\begin{split}
	(v_1,s_{12}v_2,s_{13}&v_3,\ldots,s_{1N}v_N),\\
	(v_1,s_{12}v_2,-s_{13}&v_3,\ldots,s_{12}s_{2N}v_N).
	\end{split}
	\end{equation*}
	Since $v_3\neq 0$, the latter two vectors are linearly dependent only if $s_{13}=-s_{13}$, which is impossible. Therefore, rows number $r$ and $3+l$ of $H$  given in~\eqref{eq:v1} and~\eqref{eq:v2} are linearly independent, which implies that $\operatorname{rank}(H)\geq2$.

\qed

\subsection{Proof of Proposition~\ref{SignsSync:SignsProp}}\label{proof:SignsSync:SignsProp}

	Fix a pair of indices $i<j\in[N]$. We begin by deriving an expression for $(Su_s)_{ij}$, which is the entry of $Su_s$ corresponding to the $(i,j)^{th}$ row of $S$. By~\eqref{SignsSync:SignsSyncMat} and the first equality in~\eqref{Clr.Omega.IdxIdentity} of Lemma~\ref{Clr.Omega.IdxLemma}, we have
	\begin{equation}\label{SignsSync:eigVecSums}
	\begin{split}
	(Su_s)_{ij}=&\sum_{k<j\in[N],k\neq i}(\widetilde{s}_{ij}\widetilde{s}_{kj})\widetilde{s}_{kj}+\sum_{j<k\in[N]}(\widetilde{s}_{ij}\widetilde{s}_{jk})\widetilde{s}_{jk}\\&+\sum_{k<i\in[N]}(\widetilde{s}_{ij}\widetilde{s}_{ki})\widetilde{s}_{ki}+\sum_{i<k\in[N],k\neq j}(\widetilde{s}_{ij}\widetilde{s}_{ik})\widetilde{s}_{ik}\\	
	=&\sum_{k<j\in[N],k\neq i}\widetilde{s}_{ij}+\sum_{j<k\in[N]}\widetilde{s}_{ij}+\sum_{k<i\in[N]}\widetilde{s}_{ij}+\sum_{i<k\in[N],k\neq j}\widetilde{s}_{ij}\\
	=&\sum_{k\neq i,j}\widetilde{s}_{ij}+\sum_{k\neq i,j}\widetilde{s}_{ij}=2\cdot \sum_{k\neq i,j}\widetilde{s}_{ij},	
	\end{split}
	\end{equation}
	where the 4 sums to the right of the first equality in~\eqref{SignsSync:eigVecSums} correspond to the 4 sets $A_{ij}^m$ of Lemma~\ref{Clr.Omega.IdxLemma}.	
	Since there are exactly $N-2$ indices $k\in[N]$ for which $k\neq i,j$, we conclude from~\eqref{SignsSync:eigVecSums} that $(Su_s)_{(i,j)}=2(N-2)\widetilde{s}_{ij}$ for all $i<j\in[N]$. Thus, $u_s$ is an eigenvector of $S$ corresponding to the eigenvalue $2(N-2)$.
	Next, we show that $2(N-2)$ is simple. Define a partition of the set $\{\widetilde{s}_{ij}\}_{i<j\in[N]}$ of~\eqref{SignsSync:s_ijTildeDef} into two disjoint sets
	\begin{equation}\label{SignsSync:SignClass}
	S^-=\{\widetilde{s}_{ij}\;|\;\widetilde{s}_{ij}=-1\},\quad S^+=\{\widetilde{s}_{ij}\;|\;\widetilde{s}_{ij}=1\},
	\end{equation}
	and note that a pair $\widetilde{s}_{ij}$ and $\widetilde{s}_{kl}$ are in the same set of~\eqref{SignsSync:SignClass} \textit{iff} $\widetilde{s}_{ij}\widetilde{s}_{kl}=1$.
	Thus, by~\eqref{SignsSync:SignsSyncMat} and~\eqref{SignsSync:SignClass}, the matrix $S$ is given by
	\begin{equation}
	(S)_{(i,j)(k,l)}=
	\begin{cases}
	\begin{aligned}
	1& & &|\{i,j\}\cap\{k,l\}|=1 \text{ and } \widetilde{s}_{ij} \text{ and }\\
	& & &\widetilde{s}_{kl} \text{ are in the same set of~\eqref{SignsSync:SignClass}},\\
	-1& & &|\{i,j\}\cap\{k,l\}|=1 \text{ and } \widetilde{s}_{ij} \text{ and }\\
	& &  &\widetilde{s}_{kl} \text{ are in different sets of~\eqref{SignsSync:SignClass}},\\
	0& & &\text{otherwise}.
	\end{aligned}
	\end{cases}
	\end{equation}
	In~\cite{Jsync}, it was shown that the leading eigenvalue of $S$ is simple and is given by $2(N-2)$, which concludes the proof.

\qed

\bibliographystyle{plain}
\bibliography{D21}

\end{document}